\pdfoutput=1
\RequirePackage{ifpdf}
\ifpdf 
\documentclass[pdftex]{sigma}
\else
\documentclass{sigma}
\fi

\numberwithin{equation}{section}

\newtheorem{Theorem}{Theorem}[section]
\newtheorem*{Theorem*}{Theorem}
\newtheorem{Corollary}[Theorem]{Corollary}
\newtheorem{Lemma}[Theorem]{Lemma}
\newtheorem{Proposition}[Theorem]{Proposition}
\newtheorem{con}[Theorem]{Conjecture}
 { \theoremstyle{definition}
\newtheorem{Definition}[Theorem]{Definition}

\newtheorem{Remark}[Theorem]{Remark} }

\newcommand{\bbC}{\mathbb{C}}

\newcommand{\calD}{\mathcal{D}}

\newcommand{\sfa}{\mathsf{a}}
\newcommand{\sfb}{\mathsf{b}}

\newcommand{\seteq}{\mathbin{:=}}

\newcommand{\Nk}{{\mathsf N}}

\newcommand{\ignore}[1]{}

\def\Has{{H_{\sfa\sfb}}}
\def\SS{{\mathcal{H}_{\mathrm S}}}
\def\Bor{\mathcal{B}}
\def\dfrac#1#2{{\displaystyle\frac{#1}{#2}}}

\begin{document}
\allowdisplaybreaks

\newcommand{\arXivNumber}{2211.16772}

\renewcommand{\PaperNumber}{089}

\FirstPageHeading

\ShortArticleName{Non-Stationary Difference Equation and Affine Laumon Space}

\ArticleName{Non-Stationary Difference Equation\\ and Affine Laumon Space:\\
Quantization of Discrete Painlev\'e Equation}

\Author{Hidetoshi~AWATA~$^{\rm a}$, Koji~HASEGAWA~$^{\rm b}$, Hiroaki~KANNO~$^{\rm ac}$, Ryo~OHKAWA~$^{\rm de}$, \newline Shamil~SHAKIROV~$^{\rm fg}$, Jun'ichi~SHIRAISHI~$^{\rm h}$ and Yasuhiko~YAMADA~$^{\rm i}$}

\AuthorNameForHeading{H.~Awata et~al.}

\Address{$^{\rm a)}$~Graduate School of Mathematics, Nagoya University, Nagoya 464-8602, Japan}
\EmailD{\href{mailto:awata@math.nagoya-u.ac.jp}{awata@math.nagoya-u.ac.jp}, \href{mailto:kanno@math.nagoya-u.ac.jp}{kanno@math.nagoya-u.ac.jp}}

\Address{$^{\rm b)}$~Mathematical Institute, Tohoku University, Sendai 980-8578, Japan}
\EmailD{\href{mailto:kojihas2@gmail.com}{kojihas2@gmail.com}}

\Address{$^{\rm c)}$~Kobayashi-Maskawa Institute, Nagoya University, Nagoya 464-8602, Japan}

\Address{$^{\rm d)}$~Osaka Central Advanced Mathematical Institute, Osaka Metropolitan University,\\
\hphantom{$^{\rm d)}$}~Osaka 558-8585, Japan}
\EmailD{\href{mailto:ohkawa@kurims.kyoto-u.ac.jp}{ohkawa.ryo@omu.ac.jp}}

\Address{$^{\rm e)}$~Research Institute for Mathematical Sciences, Kyoto University, Kyoto 606-8502, Japan}
\EmailD{\href{mailto:ohkawa@kurims.kyoto-u.ac.jp}{ohkawa@kurims.kyoto-u.ac.jp}}

\Address{$^{\rm f)}$~University of Geneva, Switzerland}
\EmailD{\href{mailto:shamil.shakirov@unige.ch}{shamil.shakirov@unige.ch}}

\Address{$^{\rm g)}$~Institute for Information Transmission Problems, Moscow, Russia}

\Address{$^{\rm h)}$~Graduate School of Mathematical Sciences, University of Tokyo,\\
\hphantom{$^{\rm h)}$}~Komaba, Tokyo 153-8914, Japan}
\EmailD{\href{mailto:shiraish@ms.u-tokyo.ac.jp}{shiraish@ms.u-tokyo.ac.jp}}

\Address{$^{\rm i)}$~Department of Mathematics, Kobe University, Rokko, Kobe 657-8501, Japan}
\EmailD{\href{mailto:yamaday@math.kobe-u.ac.jp}{yamaday@math.kobe-u.ac.jp}}

\ArticleDates{Received December 06, 2022, in final form October 22, 2023; Published online November 09, 2023}

\Abstract{We show the relation of the non-stationary difference equation proposed by one of the authors and the quantized discrete Painlev\'e VI equation. The five-dimensional Seiberg--Witten curve associated with the difference equation has a consistent four-dimensional limit. We also show that the original equation can be factorized as a coupled system for a pair of functions $\bigl(\mathcal{F}^{(1)}, \mathcal{F}^{(2)}\bigr)$, which is a consequence of the identification of the Hamiltonian as a translation element in the extended affine Weyl group. We conjecture that the instanton partition function coming from the affine Laumon space provides a solution to the coupled system.}


\Classification{14H70; 81R12; 81T40; 81T60}

\begin{flushright}
\begin{minipage}{70mm}
\it Dedicated to the memories of our friends,\\ Omar Foda and Yaroslav Pugai
\end{minipage}
\end{flushright}

\section{Introduction}

The conformal blocks, more precisely, the matrix elements or the traces
of the intertwiners among the Verma modules of the Virasoro algebra,
or the chiral algebra in general define special functions which are
ubiquitous in mathematics and physics. As special functions originated from
the representation theory of the symmetry, the hypergeometric series and
Nekrasov function \cite{N} to mention a few,
the conformal blocks should satisfy sufficiently simple but significant
equations.
The Belavin--Polyakov--Zamolodchikov (BPZ) equation for the Virasoro conformal block
with degenerate field insertion is a celebrated example \cite{Belavin:1984vu}.

For the deformed Virasoro algebra \cite{Shiraishi:1995rp} such an equation is expected
to be some (non-stationary) difference equation.
Though the conformal blocks allow integral representations of Dotsenko--Fateev type
or the deformed matrix model type,
the desired equations were not known for a long time,
more than a quarter century after the discovery of the algebra.
While most attempts to work out explicit form of the expected difference equation were not successful,
the non-stationary Ruijsenaars function has been proposed \cite{Shi} (see also \cite{Langmann:2020utd}).
In an appropriate limit, we can see that it reduces to the non-stationary affine Toda equation
which is a difference equation that involves the $q$-exponent of the Laplace operator
(the $q$-Borel transformation $\Bor$, see Definition \ref{qBorel}).
Recently a remarkable progress has been made by one of the authors \cite{Shakirov:2021krl} based on AGT correspondence \cite{Alday:2009aq}.
Namely a non-stationary difference equation was discovered
for the Nekrasov partition function of the five-dimensional gauge theory with a surface defect.
The AGT correspondence tells that if the theory has four matter hypermultiplets
in the fundamental representation, the partition function agrees with the genus zero
five point function with one degenerate field coming from the defect \cite{Alday:2009fs,Marshakov:2010fx,Taki:2010bj}.
In \cite{Shakirov:2021krl} the five-dimensional lift of AGT correspondence \cite{Awata:2009ur}
was applied, where the surface defect is realized by imposing the Higgsing condition on
the ${\rm SU}(2) \times {\rm SU}(2)$ quiver gauge theory. Hence, the non-stationary difference equation
proposed in \cite{Shakirov:2021krl} is regarded as a $q$-deformed version of BPZ equation.
In the decoupling limit of the matter multiplets it reproduces the non-stationary affine Toda equation
for the non-stationary Ruijsenaars function.
A distinguished feature of these equations is the appearance of the $q$-Borel transformation $\Bor$.

The quantization of the Painlev\'e equations (or isomonodromic deformations more generally) has been studied for many years.
One of the motivation of such studies was its relation to the conformal field theories.
This relation plays a key role in the recent studies of AGT correspondence (see \cite{Nekrasov:2021tik} and references therein).
Also, there has been a large progress in the study of discrete (or difference) analog of Painlev\'e equations in the last decades.
The discrete Painlev\'e equations are classified into additive, multiplicative (i.e.,~$q$-), and elliptic-difference equations~\cite{Sakai:2001},
and each class corresponds to gauge theories in four, five and six dimensions;
\def\rr{\rightarrow}
\def\se{\searrow}
\def\ne{\nearrow}
\arraycolsep=1pt
\[{
\begin{array}{@{}lcccccccccccccccccc}
{\rm elliptic{\colon}}\quad &E_8^{(1)},\\
&&&&&&&&&&&&&&&&A_1^{(1)}\\[-3mm]
&&&&&&&&&&&&&&&\ne\\
{\rm multiplicative{\colon}}\quad
&E_8^{(1)}&\rr&E_7^{(1)}&\rr&E_6^{(1)}&\rr&{\color{red}D_5^{(1)}}&\rr&A_4^{(1)}
&&\rr&A_{2+1}^{(1)}&\rr&A_{1+1}^{(1)}&\rr&A_1^{(1)}&\rr&A_0^{(1)},\\[5mm]
{\rm additive{\colon}}\quad &E_8^{(1)}&\rr&E_7^{(1)}&\rr&E_6^{(1)}&&\rr&&D_4^{(1)}
&&\rr&A_3^{(1)}&\rr&A_{1+1}^{(1)}&\rr&A_1^{(1)}&\rr&A_0^{(1)}\\
&&&&&&&&&&&&&\se&&\se\\
&&&&&&&&&&&&&&A_2^{(1)}&\rr&A_1^{(1)}&\rr&A_0^{(1)}.\\
\end{array}
}\]
Such a correspondence can be easily seen at classical level, however, the understanding at quantum level is incomplete so far.
Fortunately, for the $q$-difference Painlev\'e VI equation \cite{Jimbo-Sakai}
relevant for this paper, a natural quantization
was worked out in \cite{Hasegawa} (see also \cite{HasegawaLax,Kuroki})
based on the extended affine Weyl group symmetry of type \smash{$D_5^{(1)}$}.
Recall that the B\"acklund transformations for the discrete Painlev\'e equations are
generated by the affine Weyl group and the automorphisms of the Dynkin diagram,
which act on the dynamical variables as birational transformations.
In this paper, the discrete Painlev\'e VI equation always (except for Appendix~\ref{App.A}) means the {\it quantized} one
in the sense of \cite{Hasegawa}, where the dynamical variables $(F,G)$ are non-commutative
and the time evolution is defined by the adjoint action of the Hamiltonian.
Since the prefix~$q$- is already used for {\it classical} $q$-difference analogue of Painlev\'e VI equation,\footnote{In this paper, $q$-Painlev\'e VI equation, in contrast to $qq$-Painlev\'e~VI, means the classical difference equation obtained by Jimbo--Sakai \cite{Jimbo-Sakai}.}
we will call the {\it quantization} of the equation $qq$-Painlev\'e VI for short,
namely we use the double ``$q$'' standing for the $q$-difference and the quantization.
We warn the reader that the use of $qq$- does not mean any direct connection to the $qq$-character
introduced in~\cite{Nekrasov:2015wsu}. But there is a similarity in the sense that the full $\Omega$-background
parameters $(q,t)$ \cite{N} are turned on in both cases. The fact that one of the $\Omega$-background parameters,
say $q$, plays the role of the quantization parameter of integrable systems is the same as
the Nekrasov--Shatashvili limit~\cite{Nekrasov:2009rc}, which corresponds to the autonomous limit of the Painlev\'e equations.
The recent paper \cite{Shakirov:2021krl} gives us valuable lessons on the problem of the quantization of the discrete Painlev\'e
equation. The problem is also discussed from the viewpoint of the cluster integrable system \cite{Bershtein:2017swf}.
The relevant cluster algebras are associated with the BPS quiver of five-dimensional superconformal
field theories \cite{Bonelli:2020dcp}.

As we have mentioned above, in \cite{Hasegawa} by constructing a representation of the extended Weyl group \smash{$\widetilde{W}\bigl(D_5^{(1)}\bigr)$}
on the non-commutative variables $(F,G)$, explicit forms of the $qq$-Painlev\'e equation were
derived both in the Heisenberg and the Schr\"odinger forms.

\begin{Definition}[Heisenberg form of $qq$-Painlev\'e VI]
\begin{align*}
& \overline{F} F = q b_7 b_8
\frac{G+ b_5}{G+ b_7}
\frac{G+ b_6}{G+ b_8},\qquad
G \underline{G}=q b_3 b_4
\frac{F+b_1}{F+ b_3}
\frac{F+ b_2}{F+ b_4},
\end{align*}
where $\overline{F} =T\cdot F$, $\underline{G}=T^{-1}\cdot G$ and $T$ is a translation element in $\widetilde{W}\bigl(D_5^{(1)}\bigr)$.
$\mathbf{b}=(b_1, \dots, b_8)$ are the standard parameters for the $q$-Painlev\'e VI equation (see Appendix \ref{App.A}).
\end{Definition}

\begin{Definition}[Schr\"odinger form of $qq$-Painlev\'e VI]
\begin{align}\label{Sch-form}
H_{\mathrm{VI}}\cdot \mathsf{u}\bigl(b,G,Q| q,t^{-1}\bigr)=\mathsf{u}\bigl(b,G,Q|q,t^{-1}\bigr),
\end{align}
with the Hamiltonian given by
\begin{align*}
H_{\mathrm{VI}}:={}&
\frac{1}{\varphi\bigl(-qb_5 G^{-1}\bigr)\varphi\bigl(-q b_6G^{-1}\bigr)
\varphi\bigl(-b_7^{-1}G\bigr) \varphi\bigl(-b_8^{-1} G\bigr) } \theta\bigl(F^{-1}G;q\bigr)^{-1} \\
&\times
\frac{1}
{\varphi\bigl(q t^{-1/2} b_1G^{-1}\bigr)\varphi\bigl(q t^{-1/2} b_2 G^{-1}\bigr)
\varphi\bigl(t^{-1/2} b_3^{-1} G\bigr) \varphi\bigl(t^{-1/2} b_4^{-1} G\bigr) }\\
&\times
\theta\bigl(F^{-1}G;q\bigr)^{-1}T_{t^{1/4},\mathbf{b}},
\end{align*}
where $\varphi(x) := (x;q)_\infty$ and $\theta(X;q)=(X;q)_\infty (q/X;q)_\infty$.
$T_{t^{1/4},\mathbf{b}}$ is a shift operator of the parameters $\mathbf{b}$ and
$Q$ appears as a parameter of wave functions.
\end{Definition}

In this paper, we first show that the non-stationary equation proposed in \cite{Shakirov:2021krl}
is successfully identified with the $qq$-Painlev\'e VI equation. Namely, we prove
\begin{Proposition}\label{Thm1}
By an appropriate gauge transformation from $\mathsf{U}(\Lambda,x)$ to $\mathsf{u}\bigl(b,G,Q| q,t^{-1}\bigr)$,
the non-stationary difference equation in {\rm \cite{Shakirov:2021krl}}
\begin{align}\label{Shamil-Eq}
\mathsf{U} (t \Lambda,x)=\mathcal{A}_1(\Lambda,x) \cdot \Bor \cdot
\mathcal{A}_2(\Lambda,x) \cdot \Bor \cdot \mathcal{A}_3 (\Lambda,x) \mathsf{U}\left(\Lambda,\frac{x}{t q Q}\right)
\end{align}
is transformed to the $qq$-Painlev\'e VI equation \eqref{Sch-form}, where $\Bor$ is the $q$-Borel transformation
and~$\mathcal{A}_i(\Lambda,x)$ are multiplications of $\varphi(x)$ and $\Phi(x):=(x;q,t)_\infty$ {\rm (}see Section {\rm \ref{sec:conjecture})}.
\end{Proposition}
In contrast to the original form of the $qq$-Painlev\'e equation, the double infinite product ${\Phi(x):=(x;q,t)_\infty}$
arises as a consequence of the gauge transformation.
In \eqref{Shamil-Eq}, the parameter $x$ is related to the dynamical variable $G$ and $\Lambda$ plays the role of the time variable.
(See~Section~\ref{sec:dictinary} for the dictionary of variables between Painlev\'e side and the gauge theory side.)
$\mathsf{U} (\Lambda,x)$ is regarded as a formal power series in $x$ and $\Lambda/x$,
which is motivated by the following conjecture that the Nekrasov partition function solves the above equation.
Hence, the virtue of the gauge transformation in Proposition \ref{Thm1} is due to the conjecture
that the Nekrasov partition functions, which allow a combinatorial description, provide solutions to \eqref{Shamil-Eq}.
\begin{con}[\cite{Shakirov:2021krl},  Conjecture \ref{conjecture}]\label{Shamil-conjecture}
Let $Z(\Lambda,x)$ be the Nekrasov partition function of five-dimensional ${\rm SU}(2) \times {\rm SU}(2)$ gauge
theory $($see Definition {\rm \ref{A_2quiverZ})}. If we define a function $\Psi(\Lambda,x)$ by imposing the Higgsing condition on $Z(\Lambda,x)$,
then it gives a solution $\mathsf{U}(\Lambda,x)=\Psi(\Lambda,x)$ to the equation \eqref{Shamil-Eq}.
\end{con}

The discrete time evolution of the $q$-Painlev\'e VI equation is given by a translation element
in~\smash{$D_5^{(1)}$} root lattice, which is orthogonal to the symmetry \smash{$D_4^{(1)}$} of discrete $P_{\rm VI}$.
If we write the translation element in terms of the generators of the extended affine Weyl groups,
it is a~product of two factors which are exchanged by the automorphism $\tau$ of \smash{$D_5^{(1)}$} Dynkin graph.
The factorization of the original difference equation as the coupled system reflects this fact.
The decomposition of the discrete time evolution (or the Hamiltonian)
by the B\"acklund transformations implies that the original equation,
which is of the second order in $\Bor$ (i.e., in the $q$-exponent of the Laplace operator),
can be rewritten as a coupled system of the first order difference equations in $\Bor$,
as was already suggested in the original paper \cite{Shakirov:2021krl}.

\begin{Proposition}\label{Thm2}
The following coupled system is gauge equivalent to the non-stationary difference equation of {\rm \cite{Shakirov:2021krl}}
and hence, to the $qq$-Painlev\'e VI equation
\begin{gather*}
\mathsf{V}^{(1)}=
\frac{\Phi\bigl(q t^{-1} b_2/b_4\bigr)\Phi(b_1/b_3)}{\Phi( t b_6/b_8)\Phi(q b_5/b_7) }
\frac{1}{\varphi(- q b_6/G)\varphi(- G/b_8)}
 \nonumber \\
\phantom{\mathsf{V}^{(1)}=}{} \times \bigl(\widetilde{\Bor} \widetilde{T}_{\mathsf{p},b} \bigr)
\frac{1}{\varphi\bigl(-\mathsf{p}^{-1} q b_2/G\bigr)\varphi\bigl(-\mathsf{p}^{-1} G/b_4\bigr)}
\widetilde{T}_{\mathsf{p},b} \mathsf{V}^{(2)},\\
\widetilde{T}_{\mathsf{p},b} \mathsf{V}^{(2)} =
\frac{
\Phi\bigl(\mathsf{p}^{-2} t b_6/b_8\bigr)\Phi\bigl(\mathsf{p}^{-2}q b_5/b_7\bigr)}{\Phi\bigl(\mathsf{p}^{-2} q b_2/b_4\bigr)\Phi\bigl( \mathsf{p}^{-2} t b_1/b_3\bigr)
}
\frac{1}{\varphi\bigl(-\mathsf{p}^{-1}q b_1/G\bigr)\varphi\bigl(-\mathsf{p}^{-1}G/b_3\bigr)} \nonumber \\
\phantom{\widetilde{T}_{\mathsf{p},b} \mathsf{V}^{(2)} =}{}\times \bigl(\widetilde{\Bor} \widetilde{T}_{\mathsf{p},b} \bigr)
\frac{1}{\varphi(-q b_5/G)\varphi(-G/b_7)}\mathsf{V}^{(1)}.
\end{gather*}
\end{Proposition}

To motivate an analogous conjecture to Conjecture \ref{Shamil-conjecture},
let us recall that the instanton counting with a surface defect allows another
description in terms of the affine Laumon spaces
\cite{Alday:2010vg,AFKMY,FFNR}.
In this method the partition functions are identified with
the conformal blocks of the affine Kac--Moody algebra (the current algebra)
{\it without degenerate fields} \cite{Kozcaz:2010yp}. For example, in the present case
the parameter $x$, which originally comes from the insertion point of the degenerate field,
is replaced by the ${\rm SU}(2)$ spin variable of $\widehat{\mathfrak{sl}}_2$.
The existence of the surface defect is taken into account by introducing the
orbifold action \cite{FR,Kanno:2011fw}.
The relation between the two methods for incorporating a surface defect is discussed
in \cite{Frenkel:2015rda} from the viewpoint of integrable systems.
In fact, the role of the affine Kac--Moody algebra was already revealed in~\cite{Braverman:2004vv,Braverman:2004cr},
where a pertinent theorem was proved to demonstrate
that the prepotential of the Seiberg--Witten theory is obtained from the leading term
of the Nekrasov partition function.
We conjecture that the solutions to the coupled system are provided by the $K$-theoretic
instanton partition function derived from the equivariant character of the affine Laumon space~\cite{FFNR}.\looseness=1

\begin{con}[Conjecture \ref{secondconjecture}]
The partition function \eqref{Nf4} gives a solution to the coupled system in Proposition {\rm \ref{Thm2}}
by the following specialization of parameters:
\begin{align*}
&\mathcal{F}^{(1)}=
f\left(  \begin{matrix}
q^{1/2}b_4/b_8,&q^{1/2} b_6/b_2\vspace{1mm}\\
\bigl(\mathsf{p}^2Q\bigr)^{-1/2},&\bigl(\mathsf{p}^2Q\bigr)^{1/2}\vspace{1mm}\\
q^{-1/2} b_2/b_5,&q^{-1/2}b_7/b_4\end{matrix}\, \Bigg|\,
q^{1/2} \mathsf{p}^{-1}\mathsf{t} G^{-1}, q^{-1/2} \mathsf{p}^{-1}\mathsf{t} G\,\Bigg|\,q,t^{-1/2}\right), \vspace{1mm}\\
&\mathcal{F}^{(2)}=
f\left(  \begin{matrix}
q^{1/2}b_4/b_8,&q^{1/2} b_6/b_2\vspace{1mm}\\
\bigl(\mathsf{p}^2Q\bigr)^{-1/2},&\bigl(\mathsf{p}^2Q\bigr)^{1/2}\vspace{1mm}\\
q^{-1/2} b_1/b_6,&q^{-1/2}b_8/b_3\end{matrix} \,\Bigg|\,
- q^{1/2} \mathsf{t} G^{-1},- q^{-1/2} \mathsf{t} G\,\Bigg|\,q,t^{-1/2}\right).
\end{align*}
\end{con}

The point here is that due to the symmetry of the translation element $T$ which defines the discrete time evolution of
the $qq$-Painlev\'e VI equation, the pair of solutions $\bigl(\mathcal{F}^{(1)}, \mathcal{F}^{(2)}\bigr)$ comes from
the common instanton partition function of the affine Laumon space with different specialization of parameters.
An action of the quantum toroidal algebra of $A_r$ type
on the equivariant cohomology group and the equivariant $K$ group
of the affine Laumon space can be defined geometrically \cite{Negut:2011aa}.
In four-dimensional case (cohomological version) it has been shown that the instanton partition function
satisfies the Knizhnik--Zamolodchikov (KZ) equation for the affine Kac--Moody algebra
\cite{Nekrasov:2017gzb,Nekrasov:2021tik}.
Hence, the non-stationary difference equation
should be derived as a KZ type equation for the quantum affine algebra $U_{q}\bigl(\widehat{\mathfrak{sl}}_2\bigr)$ or
more likely $U_{q}\bigl(\widehat{\mathfrak{gl}}_2\bigr)$.
Moreover, since the affine Laumon space has elliptic cohomology,
it seems an interesting and challenging problem
to generalize our non-stationary difference equation to the elliptic case,
which might guide us for ascending some more Sakai's geometric classification scheme of the discrete Painlev\'e equations
in the quantized setting.

\looseness=1 The paper is organized as follows.
In the next section, we first summarize the non-stationary difference equation proposed in \cite{Shakirov:2021krl}.
We make a gauge transformation to rewrite it in a form which is natural from the viewpoint of the $qq$-Painlev\'e VI
equation. We also propose a dictionary between the variables on the gauge theory side and those on the Painlev\'e side.
In Section \ref{sec3}, we show that the adjoint action of the Hamiltonian involving the $q$-Borel transformation correctly
reproduces the Heisenberg form of the $qq$-Painlev\'e VI equation.
In Section \ref{sec4}, following \cite{Hasegawa}, we first recapitulate the quantization of the $q$-Painlev\'e VI equation
focusing on the representation of the extended affine Weyl group \smash{$\widetilde{W}\bigl(D_5^{(1)}\bigr)$} on the space of $q$-commutative dynamical variables~$(F,G)$.
We then make a comparison of the Hamiltonian in \cite{Hasegawa} constructed from the representation of \smash{$\widetilde{W}\bigl(D_5^{(1)}\bigr)$} and
that of the non-stationary difference equation of \cite{Shakirov:2021krl} which involves the $q$-Borel transformation.
In Section \ref{4dlimit}, we introduce the five-dimensional quantum Seiberg--Witten curve.
We show the quantum Seiberg--Witten curve allows a four-dimensional limit and it is consistent with our previous result \cite{AFKMY}.
This is regarded as a good support for the conjecture in \cite{Shakirov:2021krl}.
In Section \ref{Laumon}, we propose a coupled system which is gauge equivalent to the $qq$-Painlev\'e VI equation.
Finally, we conjecture the instanton partition functions of the affine Laumon space provide a solution
$\bigl(\mathcal{F}^{(1)}, \mathcal{F}^{(2)}\bigr)$ to the coupled system.
As a consequence of the fact that the translation element $T$ is given as the square of a certain element in
the extended Weyl group (see \eqref{Eq.(6.8)}),
$\mathcal{F}^{(1)}$ and $\mathcal{F}^{(2)}$ are obtained from a common instanton partition function with two kinds of
specialization of parameters, which are related by the automorphism~$\tau$.
A~summary of the discrete Painlev\'e VI equation is provided in Appendix~\ref{App.A}.
Some of notations and conventions for the discrete Painlev\'e VI equation are fixed there.
A few examples for supporting our conjecture in Section \ref{Laumon} are presented in Appendix~\ref{App.B}.
The four-dimensional limit for a factorized form of the Hamiltonian is discussed in Appendix~\ref{sec:App.C}.

{\it Note added}:
A proof of Conjecture \ref{secondconjecture} is provided in \cite{Awata:2023nuc}.

\section{Non-stationary difference equation}
\label{sec:conjecture}

\begin{Definition}
Let $T_{a,\Lambda}$ and $T_{b,x}$ be the shift operators acting on
the variables $\Lambda$ and $x$ by
$T_{a,\Lambda} f(\Lambda,x) \seteq f(a\Lambda,x)$,
$T_{b,x} f(\Lambda,x) \seteq f(\Lambda,bx)$.
Let $\vartheta_x\seteq x \partial_x$ be the Euler operator in $x$.
We have $\vartheta_x x=x(\vartheta_x+1)$, indicating that
$p \seteq q^{\vartheta_x}$ acts as the $q$-shift operator $q^{\vartheta_x}=T_{q,x}$.
\end{Definition}

\begin{Definition}\label{qBorel}
Set $\Bor \seteq q^{\vartheta_x(\vartheta_x+1)/2}$.
We define the action of $\Bor$ on a formal Laurent series in~$x$
as the $q$-Borel transformation:
\[
\Bor\bigg(\sum_{n} c_n x^n\bigg)=\sum_n q^{n(n+1)/2} c_n x^n.
\]
\end{Definition}

The fundamental relations among $x$, $p = q^{\vartheta_x}$ and $\Bor = q^{\vartheta_x(\vartheta_x+1)/2}$ are
\begin{equation}\label{adGamma}
p x=q xp, \qquad \Bor p=p\Bor, \qquad \Bor x= p x \Bor.
\end{equation}
One can see the last relation by looking at the action on $x^n$. In fact,
both sides give the same result; $q^{\frac{1}{2}(n+1)(n+2)} x^{n+1}$.
The $q$-Borel transformation $\Bor$ (see \cite[Section~2]{Garoufalidis:2022wij} and references therein)
plays a significant role in the non-stationary difference equation in \cite{Shakirov:2021krl}.

It is convenient to introduce the notations
$\varphi(x)\in\mathbb{Q}(q)[[x]]$ and $\Phi(x)\in\mathbb{Q}(q,t)[[x]]$
for the infinite products:
\begin{align*}
\varphi(x):={}& (x;q)_\infty = \prod_{n=0}^\infty \bigl(1-q^n x\bigr)
=\exp\left(-\sum_{n=1}^\infty\frac{1}{n} \frac{1}{1-q^n} x^n \right), \\
\Phi(x) :={}& (x;q,t)_\infty = \prod_{n,m=0}^\infty \bigl(1-q^nt^m x\bigr)
=\exp\left(-\sum_{n=1}^\infty\frac{1}{n} \frac{1}{\bigl(1-q^n\bigr)\bigl(1-t^n\bigr)} x^n \right) .
\end{align*}
They satisfy
\begin{equation}\label{t-difference}
\frac{\Phi(x)}{\Phi(tx)} = \varphi(x), \qquad
\Phi(ta\Lambda)^{-1} T_{t,\Lambda}^{-1} \Phi(ta\Lambda) = \varphi(a\Lambda)T_{t,\Lambda}^{-1}.
\end{equation}
We also use the standard notation for the $q$-shifted factorial
\begin{align*}
(u;q)_n=\prod_{i=0}^{n-1}\bigl(1-u q^i\bigr),\qquad n\in \mathbb{Z}_{\geq 0}.
\end{align*}
See \cite{hypergeometric} for useful formulas for $(u;q)_n$.


\subsection{Five point function with a degenerate field}

The correlation functions of the chiral primary fields $\Phi_{\Delta}(z)$ are the most fundamental objects
in two-dimensional conformal field theory with the energy-momentum tensor $T(z)$ (the generating currents of the Virasoro algebra).
The BPZ equation describes the response of the correlation functions
under the insertion of the descendant fields created by the action of the Virasoro algebra.
The BPZ equation for the five point function on $\mathbb{P}^1$:
\[
\Psi_{\mathrm{CFT}}(\Lambda, x) := \langle \Phi_{\Delta_4}(\infty) \Phi_{\Delta_3}(1)
\Phi_{\Delta_{(2,1)}}(x) \Phi_{\Delta_2}(\Lambda) \Phi_{\Delta_1}(0) \rangle_{\mathbb{P}^1}
\]
with a level two degenerate field, say $\Phi_{\Delta_{(2,1)}(x)}$, at $x$ is the linear differential equation of the form
\begin{equation}\label{Heun}
\bigl(\partial_x^2 + a(\Lambda, x) \partial_x + b(\Lambda, x) + c(\Lambda, x) \partial_{\Lambda}\bigr)
\Psi_{\mathrm{CFT}}(\Lambda, x) =0,
\end{equation}
which has regular singularities at $\{ 0, \Lambda, x, 1 \}$ and hence,
is identified with the non-stationary Heun equation.\footnote{The instanton partition function we are going to discuss
is expanded in $x$ and $\Lambda/x$. Hence, here we assume the radial ordering with $|x|<1$ and $|\Lambda/x| <1$.}
We will reserve $t$ for one of the equivariant parameters
($\Omega$ background) of the torus action on $\mathbb{C}^2$ and $\Lambda$ plays the role of ``time'' variable in \eqref{Heun}.
In the non-stationary case, the constant part of the Heun operator involves the time derivative $\partial_\Lambda$.
Up to the gauge transformation with the factor
\[
x^{\alpha} (x-1)^{\beta} (x- \Lambda)^\gamma,
\]
the equation \eqref{Heun} agrees with the quantization of the Painlev\'e VI equation.
The BPZ equation is also obtained from the deformed Seiberg--Witten curve of four-dimensional supersymmetric gauge theory
in the Nekrasov--Shatashvili limit as the non-stationary Schr\"odinger equation \mbox{\cite{Bullimore:2014awa,Poghossian:2016rzb}}.
This is also regarded as the quantization of (continuous, additive) Painlev\'e VI equation~\cite{AFKMY}.
What we are going to discuss in this paper is an uplift of these stories to the triality of
the deformed Virasoro algebra, discrete Painlev\'e equation and five-dimensional supersymmetric gauge theories.

Recall that in the AGT correspondence $(r+3)$ point conformal blocks on the genus zero curve
are identified with the instanton partition functions of the linear quiver gauge theory of type $A_{r}$.
Let us consider the five-dimensional uplift of the AGT correspondence.
The instanton partition function is expressed in terms of the $K$-theoretic Nekrasov factor
$\Nk_{\lambda,\mu}(u)=\Nk_{\lambda,\mu}(u|q,\kappa)$ defined~by
\begin{align*}
\Nk_{\lambda,\mu}(u|q,\kappa)
&=
\prod_{(i,j)\in \lambda}\bigl(1-u q^{\lambda_i-j}\kappa^{-\mu'_j+i-1}\bigr) \cdot
\prod_{(k,l)\in \mu}\bigl(1-u q^{-\mu_k+l-1}\kappa^{\lambda'_l-k}\bigr),
\end{align*}
or equivalently
\begin{align*}
 \Nk_{\lambda,\mu}(u|q,\kappa)=
 \prod_{j\geq i\geq 1}
\bigl(u q^{-\mu_i+\lambda_{j+1}} \kappa^{-i+j};q\bigr)_{\lambda_j-\lambda_{j+1}}\cdot
\prod_{\beta\geq \alpha \geq 1}
\bigl(u q^{\lambda_{\alpha}-\mu_\beta} \kappa^{\alpha-\beta-1};q\bigr)_{\mu_{\beta}-\mu_{\beta+1}}.
\end{align*}
Here $q$ and $\kappa$ are regarded independent indeterminates.
The Nekrasov factor $\Nk_{\lambda,\mu}(u)$ depends on a pair of partitions $(\lambda,\mu)$, namely
$\lambda=(\lambda_1,\lambda_2,\ldots)$ is non-increasing non-negative integers with finitely many positive parts.
$\lambda'$ denotes the conjugate of $\lambda$.
In \cite{Shakirov:2021krl}, the instanton partition function of five-dimensional ${\rm SU}(2) \times {\rm SU}(2)$ theory
with four fundamental matter multiplets and one bi-fundamental matter multiplet is considered.
On the deformed conformal block side this corresponds to the five point function on $\mathbb{P}^1$.
\begin{Definition}\label{A_2quiverZ}
\begin{align}
\mathcal{Z}( \Lambda,x)
:={}&\sum_{\nu_1,\nu_2,\mu_1,\mu_2\in \mathsf{P}}
\mathfrak{p}_1^{|\nu_1|+|\nu_2|}
\mathfrak{p}_2^{|\mu_1|+|\mu_2|}\nonumber\\
&\times
\prod_{1\leq a,b\leq 2}
\frac{\Nk_{\varnothing,\nu_b}\!\! \bigl(v \mathfrak{f}_a^+/\mathfrak{n}_b|q,t^{-1}\bigr)
\Nk_{\nu_a,\mu_b} \bigl(w \mathfrak{n}_a/\mathfrak{m}_b|q,t^{-1}\bigr)
\Nk_{\mu_a,\varnothing} \bigl(v \mathfrak{m}_a/\mathfrak{f}^-_b|q,t^{-1}\bigr)}
{\Nk_{\nu_a,\nu_b} \bigl( \mathfrak{n}_a/\mathfrak{n}_b|q,t^{-1}\bigr)
\Nk_{\mu_a,\mu_b} \bigl( \mathfrak{m}_a/\mathfrak{m}_b|q,t^{-1}\bigr) },\!\!\label{Z}
\end{align}
where ${\mathsf P}$ denotes the set of all partitions and $\varnothing$ is the empty partition.
\end{Definition}
{\samepage The following parametrization was used in \cite{Shakirov:2021krl}:
\begin{align}
&v=q^{1/2}t^{-1/2},\qquad w=v\phi_1, \qquad
\mathfrak{p}_1=v^{-2}T_2\phi_2 x,\qquad
\mathfrak{p}_2=v^{-2}\frac{T_4 \Lambda}{\phi_1 x}, \nonumber \\
& \mathfrak{n}_1=1,\qquad \mathfrak{n}_2=Q,\qquad
 \mathfrak{m}_1=1,\qquad \mathfrak{m}_2=\phi_1\phi_2 Q, \nonumber\\
 &\mathfrak{f}_1^+=T_1Q,\qquad
 \mathfrak{f}_2^+=T_2^{-1},\qquad
 \mathfrak{f}_1^-=T_3^{-1},\qquad
 \mathfrak{f}_2^-=T_4 \phi_1\phi_2 Q. \label{Sh-parameters}
\end{align}}
The coefficients of the expansion depend on parameters $(Q, \phi_1, \phi_2, T_1, \dots, T_4)$ and
the equivariant parameters $(q,t)$ of the torus action on $\mathbb{C}^2$.
The parameters $(Q, \phi_1, \phi_2)$ correspond to the equivariant parameters of
the Cartan subalgebra ${\rm U}(1) \times {\rm U}(1) \subset {\rm SU}(2) \times {\rm SU}(2)$ of the gauge group
and the mass of the bi-fundamental matter.\footnote{If we extend the gauge group to ${\rm U}(2) \times {\rm U}(2)$,
the (exponentiated) mass parameter of the bi-fundamental matter may be identified with the equivariant parameter
of the relative ${\rm U}(1)$ factor of the gauge group ${\rm U}(2) \times {\rm U}(2)$. Note that the diagonal ${\rm U}(1)$ factor decouples.}
On the other hand the mass parameters~$M_i$ of the fundamental hypermultiplets are related to
$\log T_i$ up to the appropriate shifts of $\log q = \epsilon_1$, $\log t = - \epsilon_2$, $\log Q = -2a$.
The instanton partition function $\mathcal{Z}( \Lambda,x)$ is a formal power series in $(x, \Lambda/x)$, where
they are related to the insertion points of the intertwiners up the ${\rm SL}(2,\mathbb{R})$ transformation.

Let us consider the degenerate conformal block with the insertion of a level two degenerate field.
Then one of the external Liouville momentum has a special value and the degenerate fusion rule
tells that there are two allowed values for the intermediate momentum.
According to the AGT correspondence this imposes {\it two} conditions on parameters $(Q, \phi_1, \phi_2)$
on the quiver gauge theory side, which is often referred to the Higgsing condition.
In the present case the conditions are explicitly $\phi_1= q^{-1/2} t^{3/2}$, $\phi_2 = q^{1/2} t^{-1/2}$
\big(or $\phi_1= q^{-1/2} t^{1/2}$, $\phi_2 = q^{3/2} t^{-1/2}$\big). Recall that there are two possibilities for the
intermediate momentum. As a consequence, one of the equivariant parameters, say $t$, is transmuted to
the Higgsing mass parameter. Later we will see $t^{1/4} = \mathsf{p}$ becomes a basic shift parameter, or
the non-autonomous parameter ($t \to 1$ is the autonomous limit) on the
Painlev\'e side. In the five brane web realization of the surface defect, the parameter $t$ is identified with the
volume of $S^3$ connecting $\mathrm{NS}5$ brane and $\mathrm{D}5$ brane which are non-intersecting. Note that
$\mathrm{D}3$ brane can wrap $S^3$ in type IIB string theory.

After imposing the Higgsing condition
the instanton partition function
\begin{equation}
\Psi(\Lambda, x) = \mathcal{Z}( \Lambda,x) \vert_{\phi_1= q^{-1/2} t^{3/2}, \phi_2 = q^{1/2} t^{-1/2}}
\end{equation}
has parameters $(Q; T_1, \dots, T_4; \Lambda, x ; q,t)$.
Physically this is the instanton partition function of the ${\rm SU}(2)$ gauge theory with a surface
defect.\footnote{In general the gauge theory allows several types of the surface defect according to
the breaking patters of the total gauge group ${\rm SU}(N)$ which are labelled by partitions of $N$.
But for ${\rm SU}(2)$ the breaking pattern is
unique ${\rm SU}(2) \to {\rm U}(1)$. Mathematically the breaking pattern defines a parabolic structure
along a defect (or a divisor).} The parameter $x$ is the insertion point of the degenerate field.
Then, the conjecture in \cite{Shakirov:2021krl} says
\begin{con}[\cite{Shakirov:2021krl}]\label{conjecture}
\begin{align}\label{non-stationary}
\Psi(t \Lambda,x)=\mathcal{A}_1(\Lambda,x) \cdot \Bor \cdot
\mathcal{A}_2(\Lambda,x) \cdot \Bor \cdot \mathcal{A}_3 (\Lambda,x) \Psi\left(\Lambda,{x\over t q Q}\right),
\end{align}
where
\begin{align*}
&\mathcal{A}_1(\Lambda,x) =
{1\over \varphi(T_1 t v x)}
{\Phi\bigl(T_3 t^2 v \Lambda x^{-1}\bigr) \over \Phi\bigl(T_3 q v \Lambda x^{-1}\bigr) }
{\Phi\bigl(T_4 t^2 v \Lambda x^{-1}\bigr) \over \Phi\bigl(T_4 t^2 v^{-1} \Lambda x^{-1}\bigr) },\\
&\mathcal{A}_2(\Lambda,x) =
{\varphi(q T_2T_3 \Lambda)\varphi(t T_1 T_4 \Lambda) \over
\varphi(-T_1T_2 x)\varphi\bigl(-Q^{-1} x\bigr)
\varphi\bigl(-T_3T_4 Q q t \Lambda x^{-1}\bigr)
\varphi\bigl(-q \Lambda x^{-1}\bigr)},\\
&\mathcal{A}_3(\Lambda,x) =
{1\over \varphi\bigl(T_2Q^{-1}q^{-1} v x\bigr)}
{\Phi\bigl(T_3 Q q^2 v \Lambda x^{-1}\bigr) \over \Phi\bigl(T_3 Q q^2 v^{-1} \Lambda x^{-1}\bigr) }
{\Phi\bigl(T_4 Q t^3 v \Lambda x^{-1}\bigr) \over \Phi\bigl(T_4 Q q^2 v^{-1} \Lambda x^{-1}\bigr) },
\end{align*}
with $v := q^{1/2} t^{-1/2}$ and $\Bor$ is the $q$-Borel transformation.
\end{con}
In \cite{Shakirov:2021krl}, several evidences for the conjecture have been provided. For example it was proved
for a special choice of mass parameters (external Liouville momenta)\footnote{Recall that the Higgsing condition is $(\phi_1, \phi_2)=\bigl(tv^{-1}, v\bigr)$.}
{\samepage
\[(T_1, T_2, T_3, T_4) = \bigl(vt^{-1}, v^{-1}, v^{-1}, vt^{-1}\bigr),\] where the solution is expressed in terms of the Macdonald polynomials.}

In this paper, we will investigate the structure of the difference equation \eqref{non-stationary},
which is independent of the validity of Conjecture \ref{conjecture}.
Hence, let us replace $\Psi(\Lambda,x)$, which was defined by the Nekrasov partition function,
with a generic function $\mathsf{U}(\Lambda,x)$ and write the equation as follows:
\begin{align}\label{qq-PVI}
\mathsf{U}(t \Lambda,x)=\mathcal{A}_1(\Lambda,x) \cdot \Bor \cdot
\mathcal{A}_2(\Lambda,x) \cdot \Bor \cdot \mathcal{A}_3 (\Lambda,x) \mathsf{U}\left(\Lambda,{x\over t q Q}\right).
\end{align}
The difference equation \eqref{qq-PVI} is called non-stationary, since there appears the shift $\Lambda \to t\Lambda$ on the left hand side.
It was pointed out \cite{Shakirov:2021krl} that in the mass decoupling limit $T_1, T_2, T_3, T_4 \to 0$,
the difference equation \eqref{qq-PVI} reduces to the non-stationary relativistic affine Toda equation \cite{Shi}.
\begin{Remark}
In the mass decoupling limit the difference equation \eqref{qq-PVI} degenerates to
\[
\mathsf{U}_{\mathrm{Toda}}(t\Lambda, x) = \widehat{H} \mathsf{U}_{\mathrm{Toda}}(\Lambda, x), \qquad
\widehat{H} = \Bor \cdot \frac{1}{\varphi\bigl(- Q^{-1}x\bigr)\varphi\bigl(- q\Lambda x^{-1}\bigr)} \cdot \Bor \cdot T_{qtQ,x}^{-1}.
\]
\end{Remark}
The name ``Toda'' comes from the following fact.
\begin{Proposition}[\cite{Shakirov:2021krl,Shi}]
The operator $\widehat{H}$ commutes with the quantum Hamiltonian of the two-particle relativistic affine Toda system:
\[
\hat{H}_{\mathrm{Toda}}= T_{q,x} + tQ T_{q,x}^{-1} + tx + \Lambda x^{-1}.
\]
\end{Proposition}


One can confirm the existence and uniqueness of the solution to
the non-stationary difference equation \eqref{qq-PVI}:
\begin{Proposition}
The equation \eqref{qq-PVI} has a formal series solution of the form
\[
\mathsf{U}(\Lambda,x)=\sum_{i,j=0}^{\infty} c_{i,j}x^i (\Lambda/x)^j,
\]
and it is unique up to a normalization.
\end{Proposition}
\begin{proof}
It is easy to see that the operator $T_{t, \Lambda}-{\mathcal A}_1 \cdot \Bor \cdot {\mathcal A}_2 \cdot \Bor \cdot
{\mathcal A}_3 T_{tqQ, x}^{-1}$ is ``triangular" in the sense that it sends the monomials $x^i (\Lambda/x)^j$ to linear combinations of
$x^{i+m} (\Lambda/x)^{j+n}$ $(m,n \geq 0)$. Moreover the leading coefficient
with $(m,n)=(0,0)$ is $t^j-q^{(i-j)(i-j+1)} (t q Q)^{j-i}$, which is nonvanishing for generic parameters $t$, $q$ and $Q$.
Hence, the coefficients $c_{i,j}$ are uniquely solved order by order with respect to $i+j$, once the initial value $c_{0,0}$ is fixed.
\end{proof}

We can eliminate $\Phi$-factors (the double infinite products) completely from the non-stationary difference equation \eqref{qq-PVI},
by the following gauge transformation:
\begin{equation}\label{2Painleve}
\mathsf{u}(\Lambda,x)= \Phi(qt T_2T_3\Lambda)\Phi\bigl(t^2 T_1T_4\Lambda\bigr)
\mathcal{A}_3(t\Lambda,t q Q x)\mathsf{U}(t\Lambda,x).
\end{equation}
Using the relations
\begin{align*}
&\mathcal{A}_3 (t\Lambda,t q Q x)\mathcal{A}_1(\Lambda,x)
= {1\over
\varphi(T_1v t x) \varphi(T_2v t x)
\varphi\bigl(T_3 vt \Lambda x^{-1}\bigr) \varphi\bigl(T_4 vt \Lambda x^{-1}\bigr) },\\
&{\Phi(q T_2T_3 \Lambda)\Phi(t T_1T_4 \Lambda) \over
\Phi(qt T_2T_3 \Lambda)\Phi\bigl(t^2 T_1T_4 \Lambda\bigr) }=\varphi(qT_2T_3\Lambda)\varphi(tT_1T_4\Lambda),
\end{align*}
which follow from \eqref{t-difference}, we obtain
\begin{Proposition}
The difference equation \eqref{qq-PVI} is gauge equivalent to
\begin{align}\label{qq-PVI-2}
H\mathsf{u}(\Lambda,x)=\mathsf{u}(\Lambda,x),
\end{align}
with the Hamiltonian
\begin{align}\label{SH}
H={}&
{1\over
\varphi(T_1v t x) \varphi(T_2v t x)
\varphi\bigl(T_3v t \Lambda x^{-1}\bigr) \varphi\bigl(T_4 v t\Lambda x^{-1} \bigr) }
 \Bor \nonumber \\
&\times
{1 \over
\varphi(-T_1T_2 x)\varphi\bigl(-Q^{-1} x\bigr)
\varphi\bigl(-T_3T_4 Q q t \Lambda x^{-1}\bigr)
\varphi\bigl(-q \Lambda x^{-1}\bigr)}
 \Bor  T_{q t Q,x}^{-1}T_{t,\Lambda}^{-1}.
\end{align}
\end{Proposition}

\begin{Remark}\label{rmk:elimination}
By a further gauge transformation,
the parameter $Q$ can be completely removed from the Hamiltonian \eqref{SH}.
\end{Remark}

To see this, first we note that
\[
x^{-c} \Bor x^c=x^{-1} (px)^c \Bor=q^{\frac{c(c+1)}{2}}p^c \Bor,
\]
for any complex parameter $c$, at least formally. Then we have
\begin{align}
x^{-c}H x^c&=q^{c(c+1)} a(x) p^c \cdot\Bor\cdot b(x) p^c \cdot\Bor\cdot T_{qtQ,x}^{-1}T_{t,\Lambda}^{-1}\nonumber \\
&=q^{c(c+1)} a(x) \cdot\Bor\cdot b(q^c x) \cdot\Bor\cdot T_{q^{2c} (q t Q)^{-1},x} T_{t,\Lambda}^{-1},
\label{eq:cHc}
\end{align}
where for simplicity we have defined
\begin{gather}
a(x)^{-1} :=\varphi(T_1 v t x)\varphi(T_2 v t x)\varphi\left(T_3 v t \frac{\Lambda}{x}\right)\varphi\left(T_4 v t \frac{\Lambda}{x}\right),
\nonumber \\
b(x)^{-1} :=\varphi(-T_1T_2 x)\varphi\bigl(-Q^{-1}x\bigr)\varphi\left(-T_3T_4 Q q t \frac{\Lambda}{x}\right)\varphi\left(-q \frac{\Lambda}{x}\right).
\nonumber
\end{gather}
Taking the exponent $c$ as $q^c=(q t Q)^{\frac{1}{2}}$, we have
\[
x^{-c}H x^c=q^{c(c+1)} a(x) \cdot \Bor \cdot b\bigl((q t Q)^{\frac{1}{2}} x\bigr) \cdot \Bor \cdot T_{t,\Lambda}^{-1}.
\]
By a rescaling
\begin{equation}\label{Q-elimination}
x \to \left(\frac{Q}{t q}\right)^{\frac{1}{2}} x, \qquad
\Lambda \to \frac{Q}{q t} \Lambda, \qquad
T_i \to Q^{-\frac{1}{2}}T_i,
\end{equation}
the functions in \eqref{eq:cHc} become
\begin{align*}
&a(x)^{-1} \to \varphi(T_1 x)\varphi(T_2 x)\varphi\left(T_3 \frac{\Lambda}{x}\right)\varphi\left(T_4\frac{\Lambda}{x}\right),
\nonumber \\
&b\bigl((q t Q)^{\frac{1}{2}} x\bigr)^{-1}\to \varphi(-T_1T_2 x)\varphi(-x)\varphi\left(-T_3T_4 \frac{\Lambda}{x}\right)\varphi\left(-t^{-1} \frac{\Lambda}{x}\right),
\end{align*}
hence, $Q$ is eliminated from these factors. The remaining overall factor $q^{c(c+1)}$ still depends on~$Q$, but it
can be removed by a further gauge transformation $H \to \Lambda^{-k} H \Lambda^k$ with $t^k=q^{c(c+1)}$.


\subsection{Dictionary of five-dimensional gauge/Painlev\'e correspondence}
\label{sec:dictinary}

In this paper, we will employ the following relations of the parameters in \cite{Shakirov:2021krl} and
the root variables for the $q$-Painlev\'e VI equation (see Appendix~\ref{App.A}):
\begin{alignat}{4}
&a_0=(T_4/T_3)^{1/4},\qquad&&
a_1=(T_1/T_2)^{1/4},\qquad&&
a_2=(t \Lambda T_1T_4)^{-1/4},&\nonumber \\
&a_3=(t^2 \Lambda T_3T_4)^{1/4},\qquad&&
a_4=(Q T_1T_2)^{1/4},\qquad&&
a_5=(t QT_3T_4)^{-1/4}.&\label{dictionary}
\end{alignat}
Namely,
\begin{equation}\label{mass-invert}
\sqrt{Q}T_1={a_1^2 a_4^2},\qquad
\sqrt{Q}T_2={ a_4^2\over a_1^2 },\qquad
\sqrt{tQ}T_3={1\over a_0^2a_5^2},\qquad
\sqrt{tQ}T_4={a_0^2\over a_5^2},
\end{equation}
and
\begin{align}
& \mathsf{p}:= {\rm e}^{\delta} = a_0a_1a_2^2a_3^2a_4a_5 = t^{1/4},\qquad
\mathsf{t} := a_3^2a_4a_5 = t^{1/2} \Lambda (T_1T_2T_3T_4)^{1/4}, \nonumber \\
\label{Lambda-invert}
& {t\Lambda\over Q}=a_3^4 a_5^4 .
\end{align}
In \cite{Hasegawa}, $\mathsf{t}$ is identified with the ``time'' variable.
The parameter $\mathsf{p}$ defines the shift parameter of the discrete time evolution.
On the gauge theory side $\Lambda$ plays the role of the corresponding time variable.
We have to accept that \eqref{dictionary} is not invertible. In fact, \eqref{mass-invert} and the last equation of \eqref{Lambda-invert}
involve the parameter $Q$. But there is no way to fix $Q$ from \eqref{dictionary}.
This is due to the fact that there is a scaling symmetry for the parameters $(T_i,\Lambda,Q)$ on the gauge theory side,\footnote{In Remark \ref{rmk:elimination}, we use the scaling transformation \eqref{Q-elimination}. We can see the same transformation
eliminates $Q$ in~\eqref{mass-invert} and the last equation of \eqref{Lambda-invert}.}
while the root variables $a_i$ are regarded as ``inhomogeneous'' coordinates
of the ambient ten-dimensional Picard lattice of
$\mathbb{P}^1 \times \mathbb{P}^1$ with eight points blow-ups (or $\mathbb{P}^2$ with nine points blow-ups).

A related issue is the fact that the parameter $Q$ is not invariant under the extended affine Weyl group of $D_5^{(1)}$.
This property was first pointed out in \cite{Bershtein:2016aef} in the case of the discrete Painlev\'e~III$_3$ equation,
which corresponds to the pure ${\rm SU}(2)$ Yang--Mills theory, by examining the tau-function.
The transformation law of $Q$ should be fixed by the invariance of the Hamiltonian under the action of the extended affine Weyl group.
For example, the Weyl reflections~$r_4$ and~$r_5$ roughly exchange $Q$ with $\Lambda^{\pm 1}$.
On the other hand, $Q$ is invariant under the Weyl reflections~$r_0$ and $r_1$. Hence,
the dictionary \eqref{dictionary} tells that the exchange of $T_1$ and~$T_2$ corresponds to
the Weyl reflection $r_1\colon(a_1, a_2) \to \bigl(a_1^{-1}, a_1a_2\bigr)$. Similarly the exchange of
$T_3$ and~$T_4$ gives the action of $r_0\colon(a_0, a_2) \to (a_0^{-1}, a_0a_2)$ on the Painlev\'e side.

The parameter $x$ is related to a pair of $q$-commuting dynamical variables $(F,G)$ as
\begin{align*}
&G=-\xi,\qquad \mbox{\rm i.e., as multiplication by $-\xi$},\\
&F=\bigl(q t^{1/2} Q\bigr)^{-1/2} \xi q^{-\vartheta_\xi},
\end{align*}
where
\[
\xi=\left({q^2 \Lambda^2 T_3T_4\over T_1T_2}\right)^{1/4} x^{-1}
\]
and $\vartheta_\xi=\xi \partial/\partial \xi$.
Note that the dynamical variables $(F,G)$ satisfy
\[
FG=q^{-1}GF.
\]

\begin{table}[t]\centering\renewcommand{\arraystretch}{1.1}
\begin{tabular}{|c|c|c|}
\hline
Parameters & Higgsed quiver theory \cite{Shakirov:2021krl} & $qq$-Painlev\'e VI equation \\ \hline
$T_i$ & (dressed) mass parameters & root variables of the outer nodes \\ \hline
$Q$ & ${\rm SU}(2)$ Coulomb modulus & parameter of solutions \\ \hline
$\Lambda$ & instanton expansion parameter & root variable of the inner node \\ \hline
$x$ & position of degenerate field insertion & dynamical variable (coordinate) \\ \hline
$q$ & $\Omega$ background along the defect & quantization parameter \\ \hline
$t$ & $\Omega$ background orthogonal to the defect & non-autonomous parameter \\ \hline
\end{tabular}
\caption{Dictionary between the quiver gauge theory and Painlev\'e VI equation.}
\end{table}

By the above dictionary the Hamiltonian \eqref{SH} in terms of the variables of Painlev\'e VI equation takes
the following formula (for the definition of the variables $b_i$ see Definition \ref{b-variables} in Appendix~\ref{App.A}):
\begin{align}
H_{\mathrm{VI}}={}&
{1\over
\varphi\bigl(-q b_5 G^{-1}\bigr) \varphi\bigl(-q b_6 G^{-1}\bigr)
\varphi\bigl(- b_7^{-1} G \bigr) \varphi\bigl(- b_8^{-1} G \bigr) }
 \widetilde{\Bor} \nonumber \\
&\times
{1 \over
\varphi\bigl(q \mathsf{p}^{-2} b_1 G^{-1}\bigr)\varphi\bigl(q \mathsf{p}^{-2} b_2 G^{-1}\bigr)
\varphi\bigl(\mathsf{p}^{-2} b_3^{-1} G\bigr)
\varphi\bigl(\mathsf{p}^{-2} b_4^{-1} G\bigr)}  \widetilde{\Bor}
 T_{t^{1/2},x}^{-1}T_{t,\Lambda}^{-1}.\label{OHG}
\end{align}
For later convenience, by writing \smash{$T_{qtQ,x}^{-1} = \bigl(T_{qt^{1/2}Q,x}^{-1/2}\bigr)^2 T_{t^{1/2},x}^{-1}$},
we have distributed \smash{$T_{qt^{1/2}Q}^{-1/2}$} and combined it with $\Bor$ to define
$\widetilde{\Bor} = \Bor \cdot T_{(qt^{1/2}Q)^{-1/2},x}= T_{(qt^{1/2}Q)^{-1/2},x} \cdot \Bor$.
Since $G \sim -x^{-1}$, we have the \smash{$\bigl(qt^{1/2}Q\bigr)^{\pm 1/2}$}-shift of the arguments of $\varphi$ involving $b_1, \dots, b_4$
after commuting them with \smash{$T_{qt^{1/2}Q,x}^{-1/2}$} in \eqref{OHG}, compared with \eqref{SH}.

In terms of the variables $b_i$, the symmetry of the Hamiltonian \eqref{OHG},
exchanging $(b_1, \dots, b_4)$ and $(b_5, \dots, b_8)$,
becomes more manifest.\footnote{The additional factor $\mathsf{p}^{-2}= t^{-1/2}$ for $b_1, \dots, b_4$ should be related to
the rescaling of the $qq$-Painlev\'e equation to be discussed in the next section.}
The symmetry is related to the diagram automorphism $\tau$ of $D_5^{(1)}$, see Figure~\ref{D5} in Appendix~\ref{App.A}.
Geometrically this is the exchange of $\mathbb{P}^1$ in the product $\mathbb{P}^1 \times \mathbb{P}^1$
that appears in the description of the space of initial conditions, see Figure~\ref{bvar} in Appendix~\ref{App.A}.
We will get back to this point in Section \ref{Laumon}.

\section[Heisenberg form of the qq-Painlev\'e VI equation]{Heisenberg form of the $\boldsymbol{qq}$-Painlev\'e VI equation}\label{sec3}

We show the operator \eqref{SH} gives the $qq$-Painlev\'e VI equation
as the Heisenberg equation.
First, we consider a general linear operator $\Has $ of the form
\begin{gather}\label{Hasegawa}
\Has =A(x) \cdot \Bor \cdot B(x) \cdot \Bor.
\end{gather}
The operator $H$ in (\ref{SH}) is recovered from $\Has$ as $H=\Has T_{qtQ,x}^{-1}$,
where
\begin{align}
&A(x)=\frac{1}{\varphi(\sfa_1 x)\varphi(\sfa_2 x)\varphi\bigl(\frac{q}{\sfa_3x}\bigr)\varphi\bigl(\frac{q}{\sfa_4x}\bigr)},
\nonumber \\
&B(x)=\frac{1}{\varphi(-\sfb_1 x)\varphi(-\sfb_2 x)\varphi\bigl(-\frac{q}{\sfb_3x}\bigr)\varphi\bigl(-\frac{q}{\sfb_4x}\bigr)},\label{eq:AB-function}
\\
&\mathsf{a}_1 = q^{1/2} t^{1/2} T_1, \qquad \mathsf{a}_2 = q^{1/2} t^{1/2} T_2, \qquad
 \mathsf{a}_3 = q^{1/2} t^{-1/2} T_3^{-1} \Lambda^{-1},\nonumber\\
 & \mathsf{a}_4 = q^{1/2} t^{-1/2} T_4^{-1} \Lambda^{-1}, \nonumber \\
&\mathsf{b}_1 = T_1T_2, \qquad \mathsf{b}_2 = Q^{-1}, \qquad
\label{SYdictionary2}
\mathsf{b}_3 = t^{-1} Q^{-1} T_3^{-1}T_4^{-1} \Lambda^{-1}, \qquad \mathsf{b}_4 = \Lambda^{-1}.
\end{align}
Since, we will consider only the adjoint actions of $\Has$ on $x$, $p$ in this subsection, we have dropped
the factors independent of $x$.

\begin{Proposition}
Putting $a(x)\seteq A(qx)A(x)^{-1}$, $b(x)\seteq B(qx)B(x)^{-1}$, we have
\begin{equation}\label{eq:Hab}
\Has^{-1} x^{-1} \Has =b\bigl(p^{-1}x\bigr) x^{-1}p^2, \qquad \Has x^{-1}p \Has ^{-1}=x^{-1}p^{-1}a(x).
\end{equation}
\end{Proposition}

\begin{proof}
Since $\Bor^{-1} x^{-1} \Bor=x^{-1}p$, $\Bor^{-1} x \Bor=p^{-1}x$, we have
\begin{align*}
\Has^{-1} x^{-1} \Has&= \Bor^{-1}B(x)^{-1}\Bor^{-1}A(x)^{-1}x^{-1}A(x)\Bor B(x)\Bor\\
&= \Bor^{-1}B(x)^{-1}\Bor^{-1}x^{-1}\Bor B(x)\Bor\\
&= \Bor^{-1}B(x)^{-1}x^{-1}p B(x)\Bor=\Bor^{-1}B(x)^{-1}B(qx)x^{-1} \Bor p \\
&=  \Bor^{-1}b(x)x^{-1} \Bor p=b\bigl(p^{-1}x\bigr)x^{-1}p^2.
\end{align*}
Similarly, using $\Bor x^{-1} p\Bor^{-1}=x^{-1}$. $\Bor x^{-1} \Bor^{-1} =x^{-1}p^{-1}$, we have
\begin{align*}
\Has x^{-1}p \Has^{-1}&=  A(x)\Bor B(x) \Bor x^{-1}p \Bor^{-1}B(x)^{-1}\Bor^{-1} A(x)^{-1}\\
&= A(x)\Bor B(x) x^{-1} B(x)^{-1}\Bor^{-1} A(x)^{-1}\\
&= A(x)\Bor x^{-1} \Bor^{-1} A(x)^{-1}
=A(x)x^{-1}p^{-1} A(x)^{-1}\\
&= x^{-1}p^{-1}A(qx)A(x)^{-1}=x^{-1}p^{-1}a(x).
\end{align*}
Hence, the desired equations are proved.
\end{proof}

\begin{Corollary} We put
\[
f \seteq x^{-1} p, \qquad g \seteq x^{-1}, \qquad {\rm i.e.,} \qquad x= g^{-1},
\qquad p=g^{-1} f, \qquad gf =q fg,
\]
then we have
\begin{gather}\label{eq:qP6AB}
\Has f \Has^{-1}=\frac{\sfa_3\sfa_4}{qf} \dfrac{(g-\sfa_1)(g-\sfa_2)}{(g-\sfa_3)(g-\sfa_4)}, \qquad
\Has^{-1} g \Has=\dfrac{(f+\sfb_1)(f+\sfb_2)}{(f+\sfb_3)(f+\sfb_4)}\frac{\sfb_3\sfb_4}{q g}.
\end{gather}
\end{Corollary}

\begin{proof}
For the function $A(x)$, $B(x)$ in \eqref{eq:AB-function}, we see
\[
a(x)=\dfrac{(1-\sfa_1 x)(1-\sfa_2 x)}{(1-\sfa_3x)(1-\sfa_4x)}\sfa_3\sfa_4x^2, \qquad
b(x)=\dfrac{(1+\sfb_1 x)(1+\sfb_2 x)}{(1+\sfb_3x)(1+\sfb_4x)}\sfb_3\sfb_4x^2.
\]
Then from \eqref{eq:Hab}, we have
\begin{gather*}
\Has f \Has^{-1}=x^{-1}p^{-1}a(x)=\frac{\sfa_3\sfa_4}{q f}\dfrac{(1-\sfa_1/g)(1-\sfa_2/g)}{(1-\sfa_3/g)(1-\sfa_4/g)},  \\
\Has^{-1} g \Has=b\bigl(p^{-1} x\bigr)x^{-1}p^2=\dfrac{(1+\sfb_1/f)(1+\sfb_2/f)}{(1+\sfb_3/f)(1+\sfb_4/f)}\frac{\sfb_3\sfb_4}{q g},
\end{gather*}
as desired.
\end{proof}

The equation \eqref{eq:qP6AB} can be viewed as the Heisenberg form of the $qq$-Painlev\'e VI equation.
We will show the equation coincides with that obtained in \cite{Hasegawa}.
To do this, we regard the $qq$-Painlev\'e VI equation as a certain algebra automorphism
on a skew field $K$\footnote{This skew field is essentially the same as $\mathbb{F}$ to be introduced in Section \ref{sec41}.}
generated by the $q$-commuting variables $x$, $p$ and equivariant parameters $T_1,\ldots, T_4,Q,q,t$.
For any $X \in K$ we define its discrete time evolution as{\samepage
\begin{alignat*}{3}
&X \mapsto \overline{X}:=U X U^{-1}, \qquad&& (\text{up $=$ forward evolution}) &\\
&X \mapsto \underline{X}:=U^{-1} X U, \qquad &&(\text{down $=$ backward evolution}) &\\
&U=\Has T_{qtQ,x}^{-1} T_{t,\Lambda}^{-1}.&&&
\end{alignat*}}%
Note that $\overline{\Lambda}=t^{-1} \Lambda$, hence
\begin{equation*}
\overline{(\sfa_1,\sfa_2,\sfa_3,\sfa_4,\sfb_1,\sfb_2,\sfb_3,\sfb_4)}
=(\sfa_1,\sfa_2,t \sfa_3,t \sfa_4,\sfb_1,\sfb_2,t \sfb_3, t \sfb_4).
\end{equation*}
Using $T_{q t Q,x}f=(qtQ)^{-1}f T_{q t Q,x}$ and
$T_{q t Q,x}g=(qtQ)^{-1}gT_{q t Q,x}$, we can rewrite the equations~\eqref{eq:qP6AB} as
\begin{align}
\overline{f}&=\Has T_{t,\Lambda}^{-1}T_{qtQ,x}^{-1} f T_{t,\Lambda}T_{qtQ,x}\Has^{-1}
=q t Q \Has f \Has^{-1}=t Q a_3 a_4 f^{-1} \dfrac{(g-a_1)(g-a_2)}{(g-a_3)(g-a_4)},\nonumber \\
\underline{g}
&=T_{t,\Lambda}T_{qtQ,x} \Has^{-1} g \Has T_{t,\Lambda}^{-1}T_{qtQ,x}^{-1}=T_{t,\Lambda}T_{qtQ,x} \dfrac{(f+\sfb_1)(f+\sfb_2)}{(f+\sfb_3)(f+\sfb_4)}\frac{\sfb_3\sfb_4}{q} g^{-1}
T_{t,\Lambda}^{-1}T_{qtQ,x}^{-1} \nonumber \\
&=t^{-1}Q \sfb_3\sfb_4 \dfrac{(f+qtQ \sfb_1)(f+qtQ\sfb_2)}{(f+qQ\sfb_3)(f+qQ\sfb_4)}g^{-1}.\label{eq:qP6ABfg}
\end{align}

\begin{Theorem}
By a suitable rescaling $F=\alpha f$, $G=- \beta g$,
the equation \eqref{eq:qP6ABfg} can be written in the form of {\rm\cite{Hasegawa}}\footnote{See Proposition \ref{Hasegawa-form} for a derivation of \eqref{eq:qqP6-Hasegawa}.}
\begin{align}
& \overline{F}=
q {\mathsf{p}}^2{\mathsf{t}}^{-2}
{\bigl(G+{\mathsf{t}}{\mathsf{p}}^{-1}a_1^2 \bigr)
\bigl(G+{\mathsf{t}}{\mathsf{p}}^{-1}a_1^{-2}\bigr)
\over \bigl(G+{\mathsf{t}}^{-1}{\mathsf{p}}a_0^{2}\bigr)
 \bigl(G+{\mathsf{t}}^{-1}{\mathsf{p}}a_0^{-2}\bigr)} F^{-1},\nonumber \\
& \underline{G} =
 q {\mathsf{t}}^{-2}G^{-1}
{ \bigl(F+{\mathsf{t}}a_4^2\bigr) \bigl(F+{\mathsf{t}}a_4^{-2}\bigr)
\over \bigl(F+{\mathsf{t}}^{-1}a_5^2\bigr)\bigl(F+{\mathsf{t}}^{-1}a_5^{-2}\bigr)} ,\label{eq:qqP6-Hasegawa}
\end{align}
and $\overline{\mathsf{t}}=\mathsf{p}^{-2}\mathsf{t}$.
\end{Theorem}
\begin{proof}
From \eqref{mass-invert} and \eqref{SYdictionary2} the parameters are related as
\begin{align*}
&\sfa_1=a_1^2 a_4^2 \sqrt{\frac{q t}{Q}},\qquad
\sfa_2=\frac{a_4^2}{a_1^2}\sqrt{\frac{q t}{Q}},\qquad
\sfa_3=a_0^2 a_5^2 \frac{\sqrt{qQ}}{\Lambda},\qquad
\sfa_4=\frac{a_5^2}{a_0^2} \frac{\sqrt{qQ}}{\Lambda}, \\
&\sfb_1=\frac{a_4^4}{Q},\qquad
\sfb_2=\frac{1}{Q},\qquad
\sfb_3=\frac{a_5^4}{\Lambda },\qquad
\sfb_4=\frac{1}{\Lambda }, \qquad
\mathsf{t}= \frac{a_4}{a_5}\sqrt{ \frac{t \Lambda}{Q}}, \qquad
\mathsf{p} =t^{\frac{1}{4}}.
\end{align*}
Hence, from \eqref{eq:qP6ABfg}, we have
\begin{align*}
&\overline{F}=
\frac{\overline{\alpha}\alpha a_5^4 q Q^2 t}{\Lambda ^2}
\frac{\left(G+a_1^2\frac{a_4^2 \beta \sqrt{t}}{\sqrt{qQ}}\right)
\left(G+a_1^{-2}\frac{a_4^2 \beta \sqrt{t}}{\sqrt{qQ}}\right)}
{\left(G+a_0^2\frac{a_5^2 \beta \sqrt{Q}}{\Lambda \sqrt{q}}\right)
\left(G+a_0^{-2}\frac{a_5^2 \beta \sqrt{Q}}{\Lambda \sqrt{q}}\right)}F^{-1}, \\
&\underline{G}=
\frac{a_5^4 \beta \underline{\beta} Q}{\Lambda ^2 t}G^{-1}
\frac{\bigl(F+a_4^4\alpha t\bigr)(F+\alpha t)}
{\left(F+a_5^4\frac{\alpha Q}{\Lambda}\right)\left(F+\frac{\alpha Q}{\Lambda}\right)}.
\end{align*}
By choosing rescaling parameters $\alpha$, $\beta$ as
\begin{align}\label{time-depnedent-scaling}
\alpha = \frac{1}{a_4 a_5}\sqrt{\frac{\Lambda}{tQ}}, \qquad
\beta = \frac{1}{a_4 a_5} \frac{\sqrt{q\Lambda}}{t^{\frac{1}{4}}},
\end{align}
we obtain the desired results \eqref{eq:qqP6-Hasegawa}.
\end{proof}

We should emphasize that the scaling parameters \eqref{time-depnedent-scaling} are
time dependent, since they involve the parameter $\Lambda$, which is regarded as
the time coordinate in the non-stationary difference equation \eqref{qq-PVI}.


\section{Extended affine Weyl group action on quantum variables}\label{sec4}

In the last section we have shown that after the appropriate gauge transformation,
the adjoint action of the Hamiltonian \eqref{SH} of the non-stationary difference equation \eqref{qq-PVI}
correctly reproduces the discrete time evolution of the $q$-Painlev\'e VI equation.
On the other hand, in \cite{Hasegawa} a~quantization of the $q$-Painlev\'e VI equation was performed
by constructing a representation of the extended affine Weyl group \smash{$\widetilde{W}\bigl(D_5^{(1)}\bigr)$} on the space of
$q$-commuting dynamical variables.
Thus, we have two Hamiltonians \eqref{SH} and $\mathcal{H}_{\mathrm{Has}}$ (see \eqref{Hasegawa-operator})
obtained in \cite{Hasegawa}, which give the same adjoint action on $q$-commuting quantum variables.
Each Hamiltonian has its own advantage. The advantage of \eqref{SH} is that it is expressed in terms of the $q$-Borel transformation $\Bor$
with the explicit formula \eqref{qBorel}, which allows us to work out the wave function to
the Schr\"odinger form $qq$-Painlev\'e VI equation \eqref{qq-PVI} in the form of a formal series.
The advantage of $\mathcal{H}_{\mathrm{Has}}$ is that it is expressed as
a composition of generators of \smash{$\widetilde{W}\bigl(D_5^{(1)}\bigr)$} as shown in \eqref{Hasegawa-operator}.
It is an interesting problem how the symmetry \smash{$\widetilde{W}\bigl(D_5^{(1)}\bigr)$} of the $qq$-Painlev\'e VI acts
on the wave function which is given by the Nekrasov partition function.
We expect that the comparison of two Hamiltonians will give us a clue to solve this problem.
In this section we first recapitulate the quantization of the $q$-Painlev\'e VI equation in \cite{Hasegawa}.
Then by comparing two Hamiltonians we try to see the corresponding element to $\Bor$
in the representation of \smash{$\widetilde{W}\bigl(D_5^{(1)}\bigr)$} constructed in \cite{Hasegawa}.
Unfortunately, at the moment, we are not completely successful in this task.

\subsection{Coxeter relations}\label{sec41}

To formulate the quantization of the discrete $P_{\rm VI}$ equation reviewed in Appendix~\ref{App.A},
we first quantize the commutative canonical pair of variables $f$, $g$.
Let $F$ and $G$ be non commutative variables satisfying the $q$-commutation relation;
\begin{align}\label{FG-qcom}
FG=q^{-1}GF.
\end{align}
Recall the notation $\mathbb{ K}=\mathbb{C}(a)=\mathbb{C}(a_0,\ldots,a_5)$
for the rational function field in the root variables~$a_i$.
Let $\mathbb{ K}\langle F,G\rangle$ be the $\mathbb{ K}$-algebra generated by $F,G$
with the relation \eqref{FG-qcom}. It is known that~$\mathbb{ K}\langle F,G\rangle$ is
an Ore domain (see \cite[Section 2]{Kuroki} and references therein).
Denote by~$\mathbb{F}=\mathbb{ K}( F,G)$ the quotient skew field of $\mathbb{ K}\langle F,G\rangle$.
Note that $\mathbb{F}$ is generated by $a_0,\ldots,a_5,F$ and~$G$.
For any formal power series $h(z)$ in $z$, we use the formula
\begin{equation}\label{AdFG}
\operatorname{Ad}(F) \cdot h(G) = h\bigl(q^{-1}G\bigr), \qquad \operatorname{Ad}(G) \cdot h(F) = h(qF),
\end{equation}
in our computations.

As is summarized in Appendix~\ref{App.A}, the time evolution of the $q$-Painlev\'e VI equation is
derived from the translation element $T$ in the extended affine Weyl group $\widetilde{W}$. Hence,
we need an action of $\widetilde{W}$ on the quantum pair of dynamical variables $(F,G)$.
\begin{Definition}
Define the actions of $r_0,\ldots,r_5,\sigma_{01},\sigma_{45},\tau\in \widetilde{W}$ on
the generators of $\mathbb{F}$ by the rules:
\begin{align*}
&r_0\colon\ (a_0,a_1,a_2,a_3,a_4,a_5,F,G)
\mapsto
\bigl(a_0^{-1},a_1,a_0a_2,a_3,a_4,a_5,F,G\bigr),\\
&r_1\colon\ (a_0,a_1,a_2,a_3,a_4,a_5,F,G)
\mapsto
\bigl(a_0,a_1^{-1},a_1a_2,a_3,a_4,a_5,F,G\bigr),\\
&r_2\colon\ (a_0,a_1,a_2,a_3,a_4,a_5,F,G)
\mapsto
\left(a_0a_2,a_1a_2,a_2^{-1},a_2a_3,a_4,a_5,F{a_0a_1^{-1}G+a_2^2\over a_0a_1^{-1}a_2^2G+1},G\right),\\
&r_3\colon\ (a_0,a_1,a_2,a_3,a_4,a_5,F,G)
\mapsto
\left(a_0,a_1,a_2a_3,a_3^{-1},a_3a_4,a_3a_5,F,{a_3^2 a_4a_5^{-1} F+1\over a_4a_5^{-1}F+a_3^2}G\right),\\
&r_4\colon\ (a_0,a_1,a_2,a_3,a_4,a_5,F,G)
\mapsto
\bigl(a_0,a_1,a_2,a_3a_4,a_4^{-1},a_5,F,G\bigr),\\
&r_5\colon\ (a_0,a_1,a_2,a_3,a_4,a_5,F,G)
\mapsto
\bigl(a_0,a_1,a_2,a_3a_5,a_4,a_5^{-1},F,G\bigr),\\
&\sigma_{01}\colon\ (a_0,a_1,a_2,a_3,a_4,a_5,F,G)
\mapsto
\bigl(a_1^{-1},a_0^{-1},a_2^{-1},a_3^{-1},a_4^{-1},a_5^{-1},qF^{-1},G\bigr),\\
&\sigma_{45}\colon\ (a_0,a_1,a_2,a_3,a_4,a_5,F,G)
\mapsto
\bigl(a_0^{-1},a_1^{-1},a_2^{-1},a_3^{-1},a_5^{-1},a_4^{-1},F,qG^{-1}\bigr),\\
&\tau\colon\ (a_0,a_1,a_2,a_3,a_4,a_5,F,G)
\mapsto
\bigl(a_5^{-1},a_4^{-1},a_3^{-1},a_2^{-1},a_1^{-1},a_0^{-1},G,F\bigr).
\end{align*}
Then extend them to the actions on $\mathbb{F}$ by
the requirements: (1) $r_0,\ldots , r_5$ are the ring homomorphisms,
(2) $\sigma_{01}$, $\sigma_{45}$, $\tau$ are the ring anti-homomorphisms.
\end{Definition}
Note that we have defined $\sigma_{01}$, $\sigma_{45}$, $\tau$ as ring anti-homomorphisms.
Compare it with Definition~\ref{birational-cl} in Appendix~\ref{App.A}, where they are ring homomorphisms.

\begin{Proposition}\label{Cox-q}
These actions on $\mathbb{F}$ are compatible with the group structure of the
extended affine Weyl group $\widetilde{W}$. Namely they satisfy
the Coxeter relations given in Proposition {\rm \ref{Cox}}.
\end{Proposition}

One can check the Coxeter relations by straightforward calculations.
We only present two nontrivial cases.
\begin{Lemma}
We have
\begin{align*}
&r_3(F)=F,\qquad r_2r_3r_2 (F)=r_3r_2 (F),\qquad r_2(G)=G,\qquad r_3r_2r_3 (G)=r_2r_3 (G).
\end{align*}
Hence, we have the Coxeter relations
\begin{align*}
&r_2r_3r_2 (F)=r_3r_2r_3 (F),\qquad r_2r_3r_2 (G)=r_3r_2r_3 (G).
\end{align*}
\end{Lemma}

\begin{proof}
We have\footnote{The variables $b_i$ are defined in Appendix~\ref{App.A} (see Definition \ref{b-variables}).}
\begin{align*}
&r_2(F)=\sqrt{b_5\over b_7}F{G+b_7\over G+b_5},\qquad r_3r_2(F)=\sqrt{b_3b_5\over b_1b_7}F
{{F+b_3\over F+b_1}G+b_7 \over{F+b_3\over F+b_1}G+b_1^{-1}b_3 b_5 },\\
&r_2r_3r_2(F)=\sqrt{b_3b_5\over b_1b_7}F
{G+b_7\over G+b_5}{A\over B},
\end{align*}
where
\begin{align*}
&A={F {G+b_7\over G+b_5}+b_3 \over
F {G+b_7\over G+b_5}+b_1b_5^{-1} b_7 }G+b_5,\qquad
B={F {G+b_7\over G+b_5}+b_3 \over
F {G+b_7\over G+b_5}+b_1b_5^{-1} b_7 }G+b_1^{-1}b_3 b_5.
\end{align*}
Note that we have $AB=BA$.

Using
\begin{align*}
&{F+b_3\over F+b_1}G+b_7 ={1\over F+b_1}\bigl(F(G+b_7)+b_3\bigl(G+b_1b_3^{-1}b_7\bigr)\bigr),\\
&
{F+b_3\over F+b_1}G+b_1^{-1}b_3 b_5=
{1\over F+b_1}\bigl( F\bigl(G+b_1^{-1}b_3b_5\bigr)+b_3(G+b_5)\bigr),
\end{align*}
we obtain
\begin{align*}
r_3r_2(F)=\sqrt{b_3b_5\over b_1b_7}F
{1\over F\bigl(G+b_1^{-1}b_3b_5\bigr)+b_3(G+b_5)} \bigl(F(G+b_7)+ b_3\bigl(G+b_1b_3^{-1}b_7\bigr)\bigr).
\end{align*}

On the other hand, from
\begin{align*}
&A=C^{-1} \bigl(F(G+b_7)+ b_3\bigl(G+b_1b_3^{-1}b_7\bigr)\bigr),\\
&B=C^{-1} \bigl(F\bigl(G+b_1^{-1}b_3b_5\bigr)+b_3(G+b_5)\bigr){G+b_7\over G+b_5},\\
&C= F {G+b_7\over G+b_5}+b_1b_5^{-1} b_7,
\end{align*}
we have
\begin{align*}
r_2r_3r_2(F)=\sqrt{b_3b_5\over b_1b_7}F
{1\over F\bigl(G+b_1^{-1}b_3b_5\bigr)+b_3(G+b_5)} \bigl(F(G+b_7)+ b_3\bigl(G+b_1b_3^{-1}b_7\bigr)\bigr),
\end{align*}
indicating that we have $ r_2r_3r_2 (F)=r_3r_2 (F)$.
\end{proof}

\begin{Lemma}
We have
\begin{align*}
&\sigma_{01}r_2(F,G)=\sigma_{01} \left(F {a_0a_1^{-1}G+a_2^2\over a_0a_1^{-1}a_2^2G+1},G\right)
= \left( {a_0a_1^{-1}G+a_2^{-2}\over a_0a_1^{-1}a_2^{-2}G+1}qF^{-1},G\right),\\
&r_2\sigma_{01}(F,G)=r_2 \bigl(qF^{-1},G\bigr)
= \left(q {a_0a_1^{-1}a_2^2G+1\over a_0a_1^{-1}G+a_2^2}F^{-1},G\right),
\end{align*}
and
\begin{align*}
&\sigma_{01}r_3(F,G)=\sigma_{01} \left(F ,{a_3^2a_4a_5^{-1}F+1\over a_4a_5^{-1}F+a_3^2}G\right)
= \left( qF^{-1},G{a_3^{-2}a_4^{-1}a_5 qF^{-1}+1\over a_4^{-1}a_5qF^{-1}+a_3^{-2}} \right),\\
&r_3\sigma_{01}(F,G)=r_3\bigl(qF^{-1},G\bigr)
= \left(qF^{-1},{a_3^2a_4a_5^{-1}F+1\over a_4a_5^{-1}F+a_3^2}G\right),
\end{align*}
indicating that we have $\sigma_{01}r_2(F,G)=r_2\sigma_{01}(F,G)$,
$\sigma_{01}r_3(F,G)=r_3\sigma_{01}(F,G)$.
\end{Lemma}

\begin{Proposition}\label{Hasegawa-form}
Let $T=r_2r_1r_0r_2 \sigma_{01} r_3r_4r_5r_3 \sigma_{45}$
be the translation element in $\widetilde{W}$.
Writing $\overline{F} =T\cdot F$, $\underline{G}=T^{-1}\cdot G$ for short,
we have the $qq$-Painlev\'e VI equation
\begin{align}
\label{qP6F}
\overline{F}F&=q\mathsf{p}^2\mathsf{t}^{-2}
{G+ \mathsf{t}\mathsf{p}^{-1}a_1^2 \over
G+ \mathsf{t}^{-1}\mathsf{p}a_0^2}
{G+ \mathsf{t}\mathsf{p}^{-1}a_1^{-2} \over
G+ \mathsf{t}^{-1}\mathsf{p}a_0^{-2}},\\
\label{qP6G}
G\underline{G}&=q\mathsf{t}^{-2}
{F+ \mathsf{t}a_4^2 \over
F+ \mathsf{t}^{-1}a_5^2}
{F+ \mathsf{t}a_4^{-2} \over
F+ \mathsf{t}^{-1}a_5^{-2}},
\end{align}
where the discrete shift of the time variable is $T(\mathsf{t}) = \mathsf{p}^{-2} \mathsf{t}$.
\end{Proposition}
Note that compared with the classical version, there appears the factor $q$ on the right hand side.

\begin{proof}
We can compute the action of the reflections and the diagram automorphisms as follows:
\begin{align*}
F&\mathop{\longmapsto}^{\sigma_{45}}F
\mathop{\longmapsto}^{r_3}F
\mathop{\longmapsto}^{r_5}F
\mathop{\longmapsto}^{r_4}F
\mathop{\longmapsto}^{r_3}F\mathop{\longmapsto}^{\sigma_{01}}q F^{-1}\mathop{\longmapsto}^{r_2}q {a_0a_1^{-1}a_2^2 G+1\over a_0a_1^{-1}G+a_2^2}F^{-1}\\
&\mathop{\longmapsto}^{r_0}q {a_0a_1^{-1}a_2^2 G+1\over a_0^{-1}a_1^{-1}G+a_0^2a_2^2}F^{-1}\mathop{\longmapsto}^{r_1}q {a_0a_1^{3}a_2^2 G+1\over a_0^{-1}a_1G+a_0^2a_1^2a_2^2}F^{-1}\\
&\mathop{\longmapsto}^{r_2}q
 {a_0a_1^{3}a_2^2 G+1\over a_0^{-1}a_1G+a_0^2a_1^2a_2^2}
 {a_0a_1^{-1}a_2^2 G+1\over a_0a_1^{-1}G+a_2^2}F^{-1}= q \mathsf{p}^2\mathsf{t}^{-2}
 {G+\mathsf{t}\mathsf{p}^{-1}a_1^2 \over G+\mathsf{t}^{-1}\mathsf{p}a_0^2}
 {G+\mathsf{t}\mathsf{p}^{-1}a_1^{-2} \over G+\mathsf{t}^{-1}\mathsf{p}a_0^{-2}}F^{-1}.
\end{align*}
Similarly for $G$,
\begin{align*}
G& \mathop{\longmapsto}^{r_2}G
\mathop{\longmapsto}^{r_1}G
\mathop{\longmapsto}^{r_0}G
\mathop{\longmapsto}^{r_2}G
\mathop{\longmapsto}^{\sigma_{01}}G\mathop{\longmapsto}^{r_3}
{a_3^2a_4 a_5^{-1}F+1 \over a_4a_5^{-1}F+a_3^2}
G \mathop{\longmapsto}^{r_4}
{a_3^2a_4 a_5^{-1}F+1 \over a_4^{-1}a_5^{-1}F+a_3^2a_4^2}
G \\
&\mathop{\longmapsto}^{r_5}
{a_3^2a_4 a_5^{3}F+1 \over a_4^{-1}a_5F+a_3^2a_4^2a_5^2}
G \mathop{\longmapsto}^{r_3}
{a_3^2a_4 a_5^{3}F+1 \over a_4^{-1}a_5F+a_3^2a_4^2a_5^2}
{a_3^2a_4 a_5^{-1}F+1 \over a_4a_5^{-1}F+a_3^2}
G \\
&\mathop{\longmapsto}^{\sigma_{45}}
qG^{-1}
{a_3^{-2}a_4 a_5^{-1}F+1 \over a_4a_5^{-1}F+a_3^{-2}}
{a_3^{-2}a_4^{-3} a_5^{-1}F+1 \over a_4^{-1}a_5F+a_3^{-2}a_4^{-2}a_5^{-2}}
=
qG^{-1} \mathsf{t}^{-2}
{F+\mathsf{t} a_4^2 \over F+\mathsf{t}^{-1} a_5^2}
{F+\mathsf{t} a_4^{-2} \over F+\mathsf{t}^{-1} a_5^{-2}}.\tag*{\qed}
\end{align*}\renewcommand{\qed}{}
\end{proof}

In terms of the variables $b_i$ defined by Definition \ref{b-variables},
we can write the equations \eqref{qP6F} and~\eqref{qP6G}
as follows:
\begin{align*}
& \overline{F} F = q b_7 b_8
{G+ b_5 \over
G+ b_7}
{G+ b_6 \over
G+ b_8},\qquad
G \underline{G}=q b_3 b_4
{F+b_1 \over
F+ b_3}
{F+ b_2\over
F+ b_4}.
\end{align*}
By taking the conjugation (the adjoint action) by $F$ or $G^{-1}$ (see \eqref{AdFG}), we can also
write the equations in the following manner:
\begin{align*}
& F \overline{F}= q^{-1} \tilde{b}_7 \tilde{b}_8
{G+ \tilde{b}_5 \over
G+ \tilde{b}_7}
{G+ \tilde{b}_6 \over
G+ \tilde{b}_8},\qquad
\underline{G} G =q^{-1} \tilde{b}_3 \tilde{b}_4
{F+ \tilde{b}_1 \over
F+ \tilde{b}_3}
{F+ \tilde{b}_2\over
F+ \tilde{b}_4},
\end{align*}
where $\tilde{b}_i= q b_i$.


\subsection{Adjoint action and Yang--Baxter relation}

In the Heisenberg form of the $qq$-Painlev\'e VI equation the Hamiltonian acts
on the dynamical variables $(F,G)$ by the adjoint action.
Hence, we have to work out the adjoint action of the affine Weyl group generators including the diagram automorphism
$\sigma = \sigma_{01} \sigma_{45}$ in the translation (see Appendix~\ref{App.A}).
The fundamental part is to realize the birational transformation of the non-commutative variables $(F,G)$
by the Weyl reflections $r_2$ and $r_3$ as the adjoint action.
We can achieve it by using the following function \cite{Hasegawa}:
\begin{Definition}
For $X \in\mathbb{F} $, and $z\in\mathbb{K}$,
i.e., when we have $zX=Xz$, set
\begin{align*}
&\theta(X;q)=(X;q)_\infty (q/X;q)_\infty,\\
&R(z,X)={(-X;q)_\infty \bigl(-q X^{-1};q\bigr)_\infty \over \bigl(-z^{-1} X;q\bigr)_\infty \bigl(-z^{-1}q X^{-1};q\bigr)_\infty }
={\theta(-X;q) \over \bigl(-z^{-1} X;q\bigr)_\infty \bigl(-z^{-1}q X^{-1};q\bigr)_\infty }.
\end{align*}
\end{Definition}
Recall that $\mathbb{F}=\mathbb{ K}( F,G)$ is the quotient skew field of $\mathbb{ K}\langle F,G\rangle$.

\begin{Lemma}The function $R(z,X)$ satisfies the following formulas:
\begin{align*}
&R(z,X)=R\bigl(z,qX^{-1}\bigr),\\
&R(z,X)R\bigl(z^{-1},X\bigr)={\theta(-X;q)\theta(-X;q) \over \theta(-z X;q)\theta\bigl(-z^{-1}X;q\bigr)}=
R(z,qX)R\bigl(z^{-1},qX\bigr).
\end{align*}
\end{Lemma}
We will need the following Yang--Baxter relation to check the Coxeter relations among the adjoint actions of $R_i$
to be defined shortly (see Proposition \ref{AdR}) \cite{Hasegawa}.
\begin{Proposition}[\cite{Fateev}]
We have the Yang--Baxter equation
\begin{align}\label{FZ}
R(x,F)R(xy,G)R(y,F)=R(y,G)R(xy,F)R(x,G).
\end{align}
\end{Proposition}
The proof given in \cite{Faddeev:1993pe} was based on the Ramanujan's summation formula,
which implies the expansion formula of $R(z,X)$.
\begin{Remark}
Ramanujan's summation formula for the bilateral basic hypergeometric series
\begin{align*}
&{}_1\psi_1(a;b;q;z)=\sum_{n=-\infty}^\infty {(a;q)_n\over (b;q)_n}z^n ={(q;q)_\infty (b/a;q)_\infty (az;q)_\infty (q/az;q)_\infty
\over (b;q)_\infty (q/a;q)_\infty (z;q)_\infty (b/az;q)_\infty },
\qquad |b/a|<z<1,
\end{align*}
gives
\begin{align*}
{(q;q)_\infty \bigl(q/z^2;q\bigr)_\infty \over(q/z;q)_\infty (q/z;q)_\infty
} R(z,X)=
\sum_{n=-\infty}^\infty {(z;q)_n\over (q/z;q)_n}\bigl(-z^{-1}X\bigr)^n.
\end{align*}
\end{Remark}
Essentially the same relation as \eqref{FZ} is proved in \cite{Moriyama:2021mux},
where an elementary proof by the Heine's formula is provided.
It is also worth mentioning that the Yang--Baxter relation \eqref{FZ} is
closely related to the quantum dilogarithmic identities.
\begin{Proposition}[\cite{Kirillov}]
We have the five term identity
\begin{align*}
(-F;q)_\infty (-G;q)_\infty =(-G;q)_\infty (-FG;q)_\infty (-F;q)_\infty .
\end{align*}
\end{Proposition}

Let us begin with the action of the affine Weyl group. The action on the root variables is easily
obtained by introducing the dual letters.
\begin{Definition}\label{dualletter}
Let $\partial_0,\ldots,\partial_5$ be the dual letters associated with the simple roots satisfying~$[\partial_j,\alpha_k]=a_{jk}$.
Set
\begin{align*}
\rho_i= {\rm e}^{{\pi \over 2}\sqrt{-1} \alpha_i \partial_i}.
\end{align*}
Later we also use the dual letters associated with the fundamental weights satisfying
${[\partial_j^\prime,\alpha_k]=\!\delta_{jk}}$
for a realization of the adjoint action of the diagram automorphism.
\end{Definition}

\begin{Lemma}
The action of the affine Weyl group on $\mathbb{K}=\mathbb{C}(a)$ is realized by the adjoint action
\begin{align*}
r_i \cdot a_j= \rho_i a_j \rho_i^{-1}.
\end{align*}
\end{Lemma}

Among the Weyl reflections $r_i$, only $r_2$ and $r_3$ act on $(F,G)$ non-trivially.
We can show they are realized by the adjoint action of $R(z,X)$.
\begin{Lemma}\label{L314}
\begin{align*}
&R\left(a_2^2, {a_0\over a_1}G \right) F R\left(a_2^2, {a_0\over a_1}G \right)^{-1}=
F{{a_0\over a_1}G+a_2^2 \over {a_0\over a_1}a_2^2G+1}=r_2\cdot F,\\
&
R\left(a_3^2, {a_5\over a_4}F \right) G R\left(a_3^2, {a_5\over a_4}F \right)^{-1}=
{{a_5\over a_4}a_3^2F+1 \over {a_5\over a_4}F+a_3^2} G=r_3\cdot G.
\end{align*}
\end{Lemma}
\begin{proof}
For any parameter $z$, the computation goes as follows:
\begin{align*}
& R(z, G ) F R(z, G )^{-1}
=F R(z, qG) R(z, G )^{-1}=
F{1+z^{-1} G \over 1+z^{-1}G^{-1}}G^{-1}=
F{G+z \over zG+1},\\
&R (z, F ) G R(z, F )^{-1}
= R (z, F ) R(z, qF )^{-1}G =
{1+z^{-1} F^{-1}\over 1+z^{-1}F}FG={zF+1\over F+z} G.\tag*{\qed}
\end{align*}\renewcommand{\qed}{}
\end{proof}
Combining these two lemmas, we obtain the following result.
\begin{Proposition}\label{AdR}
If we define
\begin{align*}
&R_i=\rho_i,\qquad i=0,1,4,5,\\
&R_2=R\left(a_2^2,{a_0\over a_1}G\right)\rho_2,\qquad
R_3=R\left(a_3^2,{a_4\over a_5}F\right)\rho_3,
\end{align*}
the action of the affine Weyl group on the skew field $\mathbb{F}$ is realized by the adjoint action
\begin{align*}
r_i \cdot u=\operatorname{Ad} (R_i) u=R_i u R_i^{-1}, \qquad 0\leq i\leq 5,\quad u\in \mathbb{F}.
\end{align*}
\end{Proposition}
As we have mentioned before, one can check $\operatorname{Ad} (R_i)$ satisfy the Coxeter relation
by using the Yang--Baxter relation \eqref{FZ}.

Next let us consider the adjoint representation of the diagram automorphism
$\sigma :=\sigma_{01}\sigma_{45}$.\footnote{Though $\sigma_{01}$ and $\sigma_{45}$ are anti-homomorphisms,
the composition $\sigma$ is a ring homomorphism.}
Since
\begin{equation*}
\sigma\colon\ (a_0,a_1,a_2,a_3,a_4,a_5,F,G)\longmapsto
\bigl(a_1,a_0,a_2,a_3,a_5,a_4,qF^{-1},qG^{-1}\bigr),
\end{equation*}
the action on the root variables $a_i$ is realized by
\begin{equation*}
\operatorname{Ad}(\rho_2\rho_1\rho_0\rho_2\rho_3\rho_4\rho_5\rho_3)\colon\ (a_0,a_1,a_2,a_3,a_4,a_5,F,G)\longmapsto
\bigl(a_1,a_0,\mathsf{p}a_2,\mathsf{p}^{-1}a_3,a_5,a_4,F,G\bigr)
\end{equation*}
up to the scaling of $(a_2, a_3)$, which is nothing but the action of the translation element $T$.
On the other hand, the adjoint representation of $\sigma$ on the dynamical variables $(F,G)$
is given by suitable combinations of the theta functions with non-commutative variables in $\mathbb{F}$.
\begin{Lemma}\label{ad-theta}
We have
\begin{equation*}
\operatorname{Ad}\bigl(\theta(F^{-1}G;q)^{-1}\bigr)F= -G, \qquad
\operatorname{Ad}\bigl( \theta(F^{-1}G;q)^{-1}\bigr)G= -q^{-1}F^{-1}G^2.
\end{equation*}
\end{Lemma}
\begin{proof}
Similarly to the proof of Lemma \ref{L314}, we can compute
\begin{gather*}
\theta\bigl(F^{-1}G;q\bigr)^{-1} F  = F \frac{1}{\bigl(qF^{-1}G;q\bigr)_\infty \bigl(G^{-1}F;q\bigr)_\infty} \\
\hphantom{\theta\bigl(F^{-1}G;q\bigr)^{-1} F}{}
= F \frac{1-F^{-1}G}{1- G^{-1}F} \theta\bigl(F^{-1}G;q\bigr)^{-1} = -G \cdot \theta\bigl(F^{-1}G;q\bigr)^{-1}, \\
\theta\bigl(F^{-1}G;q\bigr)^{-1} G  = G \frac{1}{\bigl(qF^{-1}G;q\bigr)_\infty \bigl(G^{-1}F;q\bigr)_\infty} \\
\hphantom{\theta\bigl(F^{-1}G;q\bigr)^{-1} G}{}
= - G F^{-1}G \theta\bigl(F^{-1}G;q\bigr)^{-1}= -q^{-1} F^{-1} G^2 \cdot \theta\bigl(F^{-1}G;q\bigr)^{-1}.\tag*{\qed}
\end{gather*}\renewcommand{\qed}{}
\end{proof}

\begin{Lemma}
For any $z\in\mathbb{K}$, i.e., when $zF=Fz$ and $zG=Gz$,
we have
\begin{align*}
& \operatorname{Ad}\bigl(\bigl(\theta(z G;q) \theta\bigl( z^{-1}G;q\bigr) \bigr)^{-1}\bigr)F=FG^2,
\qquad \operatorname{Ad}\bigl(\bigl(\theta(z G;q) \theta\bigl( z^{-1}G;q\bigr) \bigr)^{-1}\bigr)F^{-1}=G^{-2}F^{-1},\\
& \operatorname{Ad}\bigl(\bigl(\theta(z G;q) \theta\bigl( z^{-1}G;q\bigr) \bigr)^{-1}\bigr)G=G .
\end{align*}
\end{Lemma}
Note that the right hand sides of the above relations are independent of $z$.
\begin{proof}
The second equation follows from the first. The third one is trivial.
We can check the first equation as follows:
\begin{align*}
\bigl(\theta(z G;q) \theta\bigl( z^{-1}G;q\bigr) \bigr)^{-1} F &=
F \frac{(1-zG)\bigl(1-z^{-1}G\bigr)}{\bigl(1-z^{-1}G^{-1}\bigr)\bigl(1-zG^{-1}\bigr)}\bigl(\theta(zG;q) \theta\bigl( z^{-1}G;q\bigr) \bigr)^{-1} \\
&= FG^2 \cdot \bigl(\theta(zG;q) \theta\bigl( z^{-1}G;q\bigr) \bigr)^{-1}.\tag*{\qed}
\end{align*} \renewcommand{\qed}{}
\end{proof}

Combining these lemmas, we find
\begin{align*}
& \operatorname{Ad}\bigl( \bigl(\theta( z G;q) \theta\bigl( z^{-1}G;q\bigr) \bigr)^{-1}
\theta\bigl(F^{-1}G;q\bigr)^{-1}\bigr) F
=-G,\\
& \operatorname{Ad}\bigl( \bigl(\theta( z G;q) \theta\bigl( z^{-1}G;q\bigr) \bigr)^{-1}
\theta\bigl(F^{-1}G;q\bigr)^{-1}\bigr) G
=-q F^{-1}.
\end{align*}
Hence, the square of this adjoint action\footnote{We can specialize the parameter $z$ of the first and the second actions
at our disposal.} agrees with the action of $\sigma$ on $(F,G) \mapsto \bigl(q F^{-1}, q G^{-1}\bigr)$.

On the other hand, if we define
\begin{equation*}
S:= {\rm e}^{\frac{\pi}{2}\sqrt{-1}(\alpha_0 - \alpha_1)(\partial'_0 - \partial'_1)}
 {\rm e}^{\frac{\pi}{2}\sqrt{-1}(\alpha_4 - \alpha_5)(\partial'_4 - \partial'_4)},
\end{equation*}
the adjoint action $\mathrm{Ad}(S)$ gives the same action as
\[(\rho_2\rho_1\rho_0\rho_2\rho_3\rho_4\rho_5\rho_3)^{-1}T_{\mathsf{p},a_2}T_{\mathsf{p},a_3}^{-1}\colon\
(a_0, a_1, a_2, a_3, a_4, a_5) \mapsto (a_1, a_0, a_2, a_3, a_5, a_4)\]
 on the root variables.
Hence, the desired realization of $\sigma$ by the adjoint action is
\begin{Proposition}
If we define
\begin{align*}
\Sigma
:={}&
\left(\theta\left(- {a_0\over a_1}G;q\right) \theta\left( -{a_1\over a_0}G;q\right) \right)^{-1} \theta\bigl(F^{-1}G;q\bigr)^{-1}
\nonumber \\
&\times
\left(\theta\left( {a_4\over a_5}G;q\right) \theta\left( {a_5\over a_4}G;q\right) \right)^{-1}\theta\bigl(F^{-1}G;q\bigr)^{-1}
\cdot S,
\end{align*}
the action of the element $\sigma=\sigma_{01}\sigma_{45}$ on the skew field $\mathbb{F}$
is realized by the adjoint action
\begin{align*}
\sigma\cdot u=\operatorname{Ad}(\Sigma) u,\qquad u\in \mathbb{F}.
\end{align*}
\end{Proposition}


\subsection{Comparison of the Hamiltonians}

Following \cite{Hasegawa}, we define the Hasegawa operator for the $qq$-Painlev\'e VI equation by
\begin{align}\label{Hasegawa-operator}
\mathcal{H}_{\mathrm{Has}} := R_2R_1R_0R_2R_3R_4R_5R_3\Sigma.
\end{align}

\begin{Proposition}\label{comparison}
We have
\begin{align}\label{OH}
\mathcal{H}_{\mathrm{Has}} ={}&
\frac{1}{\varphi\bigl(-qb_5 G^{-1}\bigr)\varphi\bigl(-q b_6G^{-1}\bigr)
\varphi\bigl(-b_7^{-1}G\bigr) \varphi\bigl(-b_8^{-1} G\bigr) } \theta\bigl(F^{-1}G;q\bigr)^{-1} \nonumber \\
&\times
\frac{1}
{\varphi\bigl(q\mathsf{p}^{-2} b_1G^{-1}\bigr)\varphi\bigl(q\mathsf{p}^{-2}b_2 G^{-1}\bigr)
\varphi\bigl(\mathsf{p}^{-2} b_3^{-1} G\bigr) \varphi\bigl(\mathsf{p}^{-2} b_4^{-1} G\bigr) }
\nonumber \\
&\times\theta\bigl(F^{-1}G;q\bigr)^{-1}T_{\mathsf{p},a_2}T_{\mathsf{p},a_3}^{-1},
\end{align}
and this Hamiltonian $\mathcal{H}_{\mathrm{Has}}$ and the Hamiltonian in \eqref{SH}
have the same adjoint action on variables $x$, $p$.
\end{Proposition}

\begin{proof}
We proceed as follows:
\begin{align}
&R_2R_1R_0R_2R_3R_4R_5R_3\Sigma \nonumber \\
&\quad=R\Big(a_2^2,{a_0\over a_1}G\Big)\rho_2\rho_1\rho_0
R\Big(a_2^2,{a_0\over a_1}G\Big)\rho_2
R\Big(a_3^2,{a_4\over a_5}F\Big)\rho_3\rho_4\rho_5 R\Bigl(a_3^2,{a_4\over a_5}F\Bigr)\rho_3 \Sigma \nonumber \\
&\quad=
R\Bigl(a_2^2,{a_0\over a_1}G\Bigr)
R\Bigl((a_0a_1a_2)^2,{a_1\over a_0}G\Bigr)
R\Bigl(\mathsf{p}^2(a_3a_4a_5)^{-2},{a_4\over a_5}F\Bigr)
R\Bigl(\mathsf{p}^2a_3^{-2},{a_5\over a_4}F\Bigr) \nonumber \\
 &\qquad\times
 \Bigl(\theta\Bigl(- {a_0\over a_1}G;q\Bigr) \theta\Bigl( -{a_1\over a_0}G;q\Bigr) \Bigr)^{-1} \theta\bigl(F^{-1}G;q\bigr)^{-1}
\Bigl(\theta\Bigl( {a_4\over a_5}G;q\Bigr) \theta\Bigl( {a_5\over a_4}G;q\Bigr) \Bigr)^{-1} \nonumber \\
&\qquad\times\theta\bigl(F^{-1}G;q\bigr)^{-1} \cdot T_{\mathsf{p},a_2}T_{\mathsf{p},a_3}^{-1} \nonumber \\
&\quad=
R\Bigl(a_2^2,{a_0\over a_1}G\Bigr)R\Bigl((a_0a_1a_2)^2,{a_1\over a_0}G\Bigr)
 \Bigl(\theta\Bigl(- {a_0\over a_1}G;q\Bigr) \theta\Bigl( -{a_1\over a_0}G;q\Bigr) \Bigr)^{-1} \theta(F^{-1}G;q)^{-1} \nonumber \\
 &\qquad\times
R\Bigl(\mathsf{p}^2(a_3a_4a_5)^{-2},-{a_4\over a_5}G\Bigr)
 R\Bigl(\mathsf{p}^2a_3^{-2},-{a_5\over a_4}G\Bigr)
 \Bigl(\theta\Bigl( {a_4\over a_5}G;q\Bigr) \theta\Bigl( {a_5\over a_4}G;q\Bigr) \Bigr)^{-1} \nonumber \\
 &\qquad\times\theta\bigl(F^{-1}G;q\bigr)^{-1} \cdot T_{\mathsf{p},a_2}T_{\mathsf{p},a_3}^{-1} \nonumber \\
&\quad=
{1\over \bigl(-{a_0\over a_1a_2^2} G;q\bigr)_\infty \bigl(-{a_1\over a_0a_2^2} qG^{-1};q\bigr)_\infty }
{1\over \bigl(-{1\over a_0^3a_1a_2^2} G;q\bigr)_\infty \bigl(-{1\over a_0a_1^3a_2^2} qG^{-1};q\bigr)_\infty }\nonumber\\
&\qquad\times\theta(F^{-1}G;q)^{-1}
{1\over \bigl({a_3^2a_4a_5^3\over \mathsf{p}^{2} } G;q\bigr)_\infty \bigl({a_3^2a_4^3a_5\over \mathsf{p}^{2} } qG^{-1};q\bigr)_\infty }
{1\over \bigl({a_3^2a_4\over \mathsf{p}^{2}a_5 } G;q\bigr)_\infty \bigl({a_3^2a_5\over \mathsf{p}^{2} a_4} qG^{-1};q\bigr)_\infty }\nonumber\\
&\qquad\times\theta\bigl(F^{-1}G;q\bigr)^{-1}T_{\mathsf{p},a_2}T_{\mathsf{p},a_3}^{-1}.\label{RRRR-deriv}
\end{align}
Note that the combination of the shift operators $T_{\mathsf{p},a_2}T_{\mathsf{p},a_3}^{-1}$ keeps
the constraints $a_0a_1a_2^2a_3^2a_4a_5 = t^{1/4}$ intact.
Hence, the desired equivalence is proved, if we can show the adjoint actions of~$\theta\bigl(F^{-1}G;q\bigr)^{-1}$ and $\widetilde{\Bor}=T_{(qt^{1/2}Q)^{-1/2},x} \cdot \Bor$
are the same. The adjoint action of $\theta\bigl(F^{-1}G;q\bigr)^{-1}$ is already given
in Lemma~\ref{ad-theta}. On the other hand, the fundamental commutation relation \eqref{adGamma} for
$\Bor$ implies
\begin{equation*}
\operatorname{Ad}\Bor \bigl(\tilde{F}\bigr) = - \tilde{G}, \qquad \operatorname{Ad}\Bor \bigl(\tilde{G}\bigr) = - q^{-1} \tilde{F}^{-1} \tilde{G}^2
\end{equation*}
for $\tilde{F} = x^{-1} p$, $\tilde{G} = - x^{-1}$.
Recall that we have rescaled $\tilde{F}$ and $\tilde{G}$. This is the reason why we have to combine $\Bor$ with
the shift operator $T_{(qt^{1/2}Q)^{-1/2},x}$ whose adjoint action produces the necessary multiplication factor for $\tilde{F}$ and $\tilde{G}$.
\end{proof}

The last equality in \eqref{RRRR-deriv}
is derived from the cancellation of the theta functions coming from $R_2$, $R_3$ (the Weyl reflections
with respect to the inner nodes labeled by $2$ and $3$)
and $\Sigma$ (the diagram automorphism which exchanges the external nodes
$0 \leftrightarrow 1$, $4 \leftrightarrow 5$).
Hence, it is not straightforward to see the correspondence
of each factor in \eqref{OH} to the generators of the extended affine Weyl group appearing in \eqref{Hasegawa-operator}.
\begin{Remark}
Though $\theta\bigl(F^{-1}G;q\bigr)^{-1}$ and $\widetilde{\Bor}$ have the same adjoint action on the variables $(F,G)$,
this does not necessarily mean that they are the same as operators acting on some space of functions.
In fact, contrary to $\widetilde{\Bor}$ or $\Bor$, we do not know how to define the action
of the operator~$\theta\bigl(F^{-1}G;q\bigr)^{-1}$ on the space of formal Laurent series in $x$.
\end{Remark}

As is well known, the Painlev\'e equations are derived from the isomonodromic deformation of linear system.
It is an interesting problem to find the Lax operators which give rise to the $qq$-Painlev\'e VI equation.
Note that the monodromy problem is naturally associated with the Yang--Baxter relation of the universal $R$ matrix.
In fact, in \cite{HasegawaLax} the universal $R$ matrix of \smash{$U_q\bigl(\widehat{\mathfrak{sl}}_2\bigr)$} was used to
define local Lax matrices for the $qq$-Painlev\'e VI equation, which is a good starting point to work out the problem
completely.

\section{Quantum Seiberg--Witten curve}\label{4dlimit}

\subsection[Gauge transformation and U(1) factor]{Gauge transformation and $\boldsymbol{{\rm U}(1)}$ factor}

By the following gauge transformation
\begin{align}\label{StoS}
&{1\over \Phi\bigl(t^2 T_1 T_3 \Lambda\bigr)\Phi(qt T_2 T_4 \Lambda)}
\varphi\bigl(q^{1/2}t^{1/2} T_2 x\bigr)\varphi\bigl(q^{1/2}t^{1/2} T_4 \Lambda/x\bigr) \mathsf{u}(\Lambda,x)
=\mathsf{v}^{(1)}(\Lambda,x),
\end{align}
the $qq$-Painlev\'e VI equation \eqref{qq-PVI-2}
is recast to $\widetilde{H} \mathsf{v}^{(1)}(\Lambda,x)=\mathsf{v}^{(1)}(\Lambda,x)$ with
\begin{align}\label{eq:Shiraishi}
\widetilde{H}={}&
{1\over
\varphi\bigl(T_1 q^{1/2}t^{1/2} x\bigr)
\varphi\bigl(T_3 q^{1/2}t^{1/2} \Lambda x^{-1}\bigr) }
 \Bor \nonumber \\
& \times
{\varphi(t T_1T_3 \Lambda) \varphi(q T_2T_4 \Lambda) \over
\varphi(-T_1T_2 x)\varphi\bigl(-Q^{-1} x\bigr)
\varphi\bigl(-T_3T_4 Q q t \Lambda x^{-1}\bigr)
\varphi\bigl(-q \Lambda x^{-1}\bigr)} \Bor \nonumber \\
&\times T_{q t Q,x}^{-1}T_{t,\Lambda}^{-1}
{1\over
\varphi\bigl(T_2 q^{1/2}t^{1/2} x\bigr)
 \varphi\bigl(T_4 q^{1/2}t^{1/2} \Lambda x^{-1} \bigr) }.
\end{align}
Recall that we have made the gauge transformation twice; $\mathsf{U}(\Lambda,x) \to \mathsf{u}(\Lambda,x) \to \mathsf{v}^{(1)}(\Lambda, x)$.
We can see the total gauge factor from the original system of Higgsed quiver gauge theory is simply
\begin{gather}\label{HtoL}
\frac{\Phi(qtT_2T_3\Lambda)\Phi\bigl(t^2T_1T_4\Lambda\bigr)}{\Phi\bigl(t^2T_1T_3\Lambda\bigr)\Phi(qtT_2T_4\Lambda)}
\frac{\Phi\bigl(T_3 q^{3/2} t^{-1/2} \Lambda/x\bigr)\Phi\bigl(T_4 q^{-1/2} t^{5/2}\Lambda/x\bigr)}
{\Phi\bigl(T_3 q^{1/2} t^{1/2} \Lambda/x\bigr)\Phi\bigl(T_4 q^{1/2}t^{3/2}\Lambda/x\bigr)}.
\end{gather}
Note that in the original definition of $\mathcal{A}_3$ in \cite{Shakirov:2021krl}
the four-dimensional limit is not well defined. But after the above gauge transformation we can take the four-dimensional limit.
In fact, let us look at the gauge transformation factor in the plethystic form
\begin{gather*}
\frac{\Phi(qtT_2T_3\Lambda)\Phi\bigl(t^2T_1T_4\Lambda\bigr)}{\Phi\bigl(t^2T_1T_3\Lambda\bigr)\Phi(qtT_2T_4\Lambda)}
= \exp \left( - \sum_{n=1}^\infty \frac{\bigl(q^nT_2^n-t^nT_1^n\bigr)\bigl(T_3^n-T_4^n\bigr)}{n\bigl(1-q^n\bigr)\bigl(1-t^n\bigr)} (t\Lambda)^n \right) \nonumber \\
\phantom{\frac{\Phi(qtT_2T_3\Lambda)\Phi\bigl(t^2T_1T_4\Lambda\bigr)}{\Phi\bigl(t^2T_1T_3\Lambda\bigr)\Phi(qtT_2T_4\Lambda)}
=
}{}\to (1- \Lambda)^{\frac{(\mu_3-\mu_4)(\mu_2 - \mu_1 - \epsilon_1 - \epsilon_2)}{\epsilon_1 \epsilon_2}}, \qquad \hbar \to 0,\\
 \frac{\Phi\bigl(T_3 q^{3/2} t^{-1/2} \Lambda/x\bigr)\Phi\bigl(T_4 q^{-1/2} t^{5/2}\Lambda/x\bigr)}
{\Phi\bigl(T_3 q^{1/2} t^{1/2} \Lambda/x\bigr)\Phi\bigl(T_4 q^{1/2}t^{3/2}\Lambda/x\bigr)} \nonumber \\
\qquad= \exp \left( - \sum_{n=1}^\infty \frac{\bigl(1-q^n t^{-n}\bigr)\bigl(T_4^n q^{-n/2}t^{5n/2} - T_3^n q^{n/2}t^{n/2}\bigr)}
{n\bigl(1-q^n\bigr)\bigl(1-t^n\bigr)} (\Lambda/x)^n \right) \nonumber \\
\qquad\quad\to (1- \Lambda/x)^{\frac{(\epsilon_1+ \epsilon_2)(\mu_4 - \mu_3 + \epsilon_1 + 2 \epsilon_2)}{\epsilon_1 \epsilon_2}}, \qquad \hbar \to 0,
\end{gather*}
where we set $q={\rm e}^{\hbar \epsilon_1}$, $t^{-\hbar\epsilon_2}$ and $T_i={\rm e}^{-\hbar \mu_i}$.
This four-dimensional limit completely matches with what we have found for the four-dimensional instanton partition functions.
This gauge transformation is also regarded as what is called the ${\rm U}(1)$ factor in the AGT correspondence~\cite{Alday:2009aq}.
Thus the gauge transformation \eqref{HtoL} is a five-dimensional uplift of the ${\rm U}(1)$ factor.
The five-dimensional ${\rm U}(1)$ factor for the Nekrasov partition function has been proposed in \cite{Bao:2013pwa,Hayashi:2013qwa},
where it is associated with pairs of parallel external lines in the fivebrane web diagram.
However, this~${\rm U}(1)$ factor does not recover that of \cite{Alday:2009aq} in the four-dimensional limit.
In this sense the relation of the ${\rm U}(1)$ factor in \cite{Bao:2013pwa,Hayashi:2013qwa} and our ${\rm U}(1)$ factor derived above
is not clear at the moment.

To simplify the computation with the quantum Seiberg--Witten curve,
let us make the change of parameters and variables:
\begin{alignat}{3}\label{d-variables1}
&\Lambda'= t T_1T_3 \Lambda, \qquad && x'= T_1 q^{-1/2} t^{1/2} x, & \\
&d_1 = T_1^{-1} q^{1/2} t^{-1/2} Q^{-1}, \qquad&&
d_2= T_2 q^{1/2} t^{-1/2},& \nonumber\\
&d_3 = T_3^{-1} q^{1/2} t^{-1/2}, \qquad&&
d_4 = T_4 q^{1/2} t^{1/2} Q,&\label{d-variables2}
\end{alignat}
so that $\tilde{H}$ can be written as $\tilde{H}=\SS T_{qtQ,x}^{-1}T_{t,\Lambda}^{-1}$ with
\begin{align}
\label{Shakirov}
\SS ={}& \frac{1}{\varphi(qx)\varphi(\Lambda/x)}  \Bor
\frac{\varphi(\Lambda)\varphi\bigl(q^{-1} d_1d_2d_3d_4\Lambda\bigr)}{\varphi(-d_1x)\varphi(-d_2x)\varphi(-d_3 \Lambda/x)\varphi(-d_4 \Lambda/x)}
\nonumber \\
&\times \Bor \frac{1}{\varphi\bigl(q^{-1}d_1d_2x\bigr)\varphi(d_3d_4 \Lambda/x)}.
\end{align}
where we delete $'$ of $x'$ and $\Lambda'$. Note that by the change of variables \eqref{d-variables1} and \eqref{d-variables2},
we have eliminated the parameter $t$ from \eqref{eq:Shiraishi}.
In fact, we can transform $\SS$ to $\Has$ defined by \eqref{Hasegawa} by the multiplication of $\varphi(\Lambda)\varphi\bigl(q^{-1} d_1d_2d_3d_4\Lambda\bigr)$,
which commutes with $\Bor$ and the adjoint action of $\varphi\bigl(q^{-1}d_1d_2x\bigr)\varphi(d_3d_4 \Lambda/x)$.
It is this form of $\tilde{H}$ that naturally leads to the solutions by the instanton counting
of the affine Laumon space with the following property. The instanton partition function is a formal double power series in $x$ and $\Lambda/x$:
\begin{equation}\label{double-series}
\mathcal{Z}(\Lambda, x) = \sum_{m,n \geq 0} c_{m,n} x^m (\Lambda/x)^n,
\end{equation}
where the coefficients $c_{m,n}$ are functions of $Q$, $T_i$, $q$ and $t$.
One of the characteristic features is that the ``boundary'' coefficients $c_{0,n}$ and $c_{m,0}$ factorize,
which is a consequence of the fact that there is a unique fixed point of the torus action on the instanton moduli space,
which corresponds to the topological number $n=0$ or $m=0$
This means $\SS$ should be regarded as the Hamiltonian on the gauge theory side.

\subsection{Five-dimensional quantum Seiberg--Witten curve}

When $t=1$, the $qq$-Painlev\'e VI equation becomes an autonomous system which admits a~conserved quantum curve.
The quantum curve, which is a Laurent polynomial in $(x,p)$, is identified
with a quantization of the Seiberg--Witten curve for the corresponding gauge theory.
Here we work out such a curve based on the commutativity with the operator $\SS$.

\begin{Proposition}\label{prp:5dSW}
Let $\calD=\calD(x,p)$ be an operator of the form
\begin{equation}\label{eq:D_with_c}
\calD=\sum_{i,j} c_{i,j}x^i p^j,
\end{equation}
with nonzero coefficients
\begin{equation*}
\left[
\begin{matrix}
c_{-1,1} & c_{0,1} & c_{1,1} \\
c_{-1,0} & c_{0,0} & c_{1,0} \\
c_{-1,-1} & c_{0,-1} & c_{1,-1}
\end{matrix}
\right]=\left[
\begin{matrix}
\Lambda \mu & -\frac{\mu q^2+\Lambda d_1 d_2}{q} & d_1 d_2 \\
-\Lambda \mu \left(d_3+d_4\right) & u & -d_1-d_2 \\
\Lambda \mu d_3 d_4 & -\Lambda \mu d_3 d_4-1 & 1
\end{matrix}
\right].
\end{equation*}
Then it satisfies the relation
\begin{equation}\label{eq:D_cons_cond}
\SS^{-1} \calD(x,p) \SS=\calD(\mu x,p),
\end{equation}
where $u, \mu \in \bbC$ are free parameters.
\end{Proposition}

\begin{proof}
We compute the successive transformations of the operators
\begin{equation*}
\calD \mapsto \calD_1 \mapsto \dots \mapsto \calD_6,
\end{equation*}
under the adjoint actions $\operatorname{Ad}(X)\colon \calD \mapsto X \calD X^{-1}$ for factors $X$ in $\SS^{-1}$.
In the following we will display the operators $\calD_i$ by their coefficient matrices $(c_{i,j})$.\footnote{Initially the range (support) of the indices $(i,j)$ is $-1\leq i$, $j \leq 1$. However, by the adjoint action
the range of the power of $x$ is extended to $-2 \leq i \leq 2$. See also the Newton polygons displayed below.}

First, under the gauge transformation
$\operatorname{Ad}\bigl({\varphi(qx)\varphi\bigl(\frac{\Lambda}{x}\bigr)}\bigr)\colon p \mapsto \frac{1-q x}{1-\frac{\Lambda}{qx}}p$,
we have
\begin{equation*}
\calD \mapsto \calD_1=\left[
\begin{matrix}
 0 & 0 & -q \mu & \mu q^2+d_1 d_2 & -q d_1 d_2 \\
 0 & -\Lambda \mu \left(d_3+d_4\right) & u & -d_1-d_2 & 0 \\
 -\Lambda ^2 \mu d_3 d_4 & \Lambda \left(\mu d_3 d_4+1\right) & -1 & 0 & 0
\end{matrix}
\right].
\end{equation*}
Next, by the action of $\operatorname{Ad}\bigl(\Bor^{-1}\bigr)\colon x^i \mapsto q^{-i(i+1)/2}x^i p^{-i}$, we obtain
\begin{equation*}
\calD_1 \mapsto \calD_2=\left[
\begin{matrix}
 -\frac{\Lambda ^2 \mu d_3 d_4}{q} & -\Lambda \mu \left(d_3+d_4\right) & -q \mu & 0 & 0 \\
 0 & \Lambda \left(\mu d_3 d_4+1\right) & u & \frac{\mu q^2+d_1 d_2}{q} & 0 \\
 0 & 0 & -1 & -\frac{d_1+d_2}{q} & -\frac{d_1 d_2}{q^2}
\end{matrix}
\right].
\end{equation*}
Then by the gauge transformation
$\operatorname{Ad}\bigl(\varphi(-d_1x)\varphi\bigl(-d_3\frac{\Lambda}{x}\bigr)\bigr)\colon
p \mapsto \frac{(1+d_1 x)}{\bigl(1+d_3 \frac{\Lambda}{q x}\bigr)} p$,
we have
\begin{equation*}
\calD_2 \mapsto \calD_3=\left[
\begin{matrix}
-\Lambda \mu d_4 & -\mu (q+\Lambda d_1 d_4) & -q \mu d_1 \\
\Lambda (\mu d_3 d_4+1) & u & \frac{\mu q^2+d_1 d_2}{q} \\
-\Lambda d_3 & -\frac{q+\Lambda d_2 d_3}{q} & -\frac{d_2}{q}
\end{matrix}
\right].
\end{equation*}
And by further gauge transformation
$
\operatorname{Ad}\bigl(\varphi(-d_2x)\varphi\bigl(-d_4\frac{\Lambda}{x}\bigr)\bigr)\colon
p \mapsto \frac{(1+d_2 x)}{\bigl(1+d_4 \frac{\Lambda}{q x}\bigr)} p,
$
we have
\begin{equation*}
\calD_3 \mapsto \calD_4=\left[
\begin{matrix}
 0 & 0 & -q \mu & -q \mu (d_1+d_2) & -q \mu d_1 d_2 \\
 0 & \Lambda (\mu d_3 d_4+1) & u & \frac{\mu q^2+d_1 d_2}{q} & 0 \\
 -\Lambda ^2 d_3 d_4 & -\Lambda (d_3+d_4) & -1 & 0 & 0 \\
\end{matrix}
\right].
\end{equation*}
Then the action of $\operatorname{Ad}\bigl(\Bor^{-1}\bigr)$ gives
\begin{equation*}
\calD_4 \mapsto \calD_5=\left[
\begin{matrix}
 -\frac{\Lambda ^2 d_3 d_4}{q} & \Lambda (\mu d_3 d_4+1) & -q \mu & 0 & 0 \\
 0 & -\Lambda (d_3+d_4) & u & -\mu (d_1+d_2) & 0 \\
 0 & 0 & -1 & \frac{\mu q^2+d_1 d_2}{q^2} & -\frac{\mu d_1 d_2}{q^2}
\end{matrix}
\right].
\end{equation*}
Finally, by the gauge transformation
$\operatorname{Ad}\bigl(\varphi\bigl(\frac{d_1d_2x}{q}\bigr)\varphi\bigl(d_3d_4 \frac{\Lambda}{x}\bigr)\bigr)\colon p \mapsto \frac{1-\frac{d_1 d_2}{q} x}{1-\frac{\Lambda d_3 d_4}{qx}}p$,
we have
\begin{equation*}
\calD_5 \mapsto \calD_6=\left[
\begin{matrix}
\Lambda & -\frac{\mu q^2+\Lambda d_1 d_2}{q} & \mu d_1 d_2 \\
-\Lambda (d_3+d_4) & u & -\mu (d_1+d_2) \\
\Lambda d_3 d_4 & -\Lambda \mu d_3 d_4-1 & \mu
\end{matrix}
\right].
\end{equation*}
In total, we obtained the relation
\begin{equation*}
\SS^{-1} \calD \SS=\calD_6=\calD |_{x \to \mu x},
\end{equation*}
as desired.
\end{proof}

\begin{Corollary} If we set $\mu=(t q Q)^{-1}$, then we have
\begin{equation*}
{\tilde H} ^{-1} \calD {\tilde H}=T_{t, \Lambda}(\calD),
\end{equation*}
where ${\tilde H} =\SS T_{tq Q,x}^{-1}T_{t,\Lambda}^{-1}$.
Namely, the operator $\calD$ with $\mu=(t q Q)^{-1}$ is conserved under the evolution by $\tilde{H}$
in autonomous case: $t=1$.
\end{Corollary}

\begin{Remark}
The converse of the Proposition \ref{prp:5dSW} is also true. Namely,
for a general operator~$\calD$ of the form \eqref{eq:D_with_c} with the initial support $i,j \in \{-1,0,1\}$,
the condition \eqref{eq:D_cons_cond} with fixed $\mu$ determines the coefficients $c_{i,j}$ up to two free parameters.
In fact, in order to keep the form of the operator $\calD$ under the successive transformations as above,
one has six linear constraints on the nine coefficients $c_{i,j}$. Then the condition \eqref{eq:D_cons_cond} gives one more constraint.
Hence, two coefficients (the constant term $u=c_{0,0}$ and overall normalization) remain free.
\end{Remark}
\begin{Remark}
The position of the nonzero coefficients $c_{i,j}$ shows the Newton polygon of the operators $\calD_k$.
Their transitions are as follows:
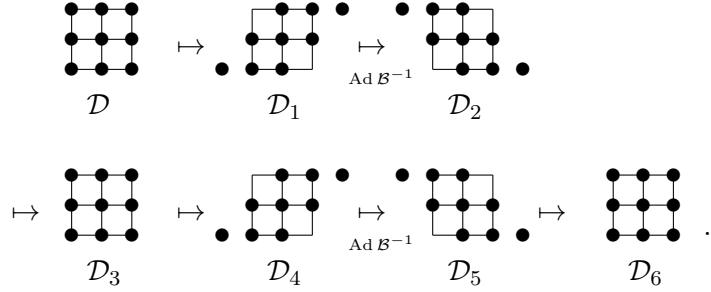
\begin{center}\setlength{\unitlength}{0.4mm}
\begin{picture}(140,45)(10,-15)
\put(-10,0){\line(0,1){20}}
\put(0,0){\line(0,1){20}}
\put(10,0){\line(0,1){20}}
\put(-10,0){\line(1,0){20}}
\put(-10,10){\line(1,0){20}}
\put(-10,20){\line(1,0){20}}
\multiput(-10,0)(10,0){3}{\circle*{4}}
\multiput(-10,10)(10,0){3}{\circle*{4}}
\multiput(-10,20)(10,0){3}{\circle*{4}}
\put(-5,-15){$\calD$}
\put(25,5){$\mapsto$}
\put(50,0){\line(0,1){20}}
\put(60,0){\line(0,1){20}}
\put(70,0){\line(0,1){20}}
\put(50,0){\line(1,0){20}}
\put(50,10){\line(1,0){20}}
\put(50,20){\line(1,0){20}}
\multiput(40,0)(10,0){3}{\circle*{4}}
\multiput(50,10)(10,0){3}{\circle*{4}}
\multiput(60,20)(10,0){3}{\circle*{4}}
\put(55,-15){$\calD_1$}
\put(85,5){$\mapsto$}
\put(82,-5){{\tiny $\operatorname{Ad}\Bor^{-1}$}}
\put(110,0){\line(0,1){20}}
\put(120,0){\line(0,1){20}}
\put(130,0){\line(0,1){20}}
\put(110,0){\line(1,0){20}}
\put(110,10){\line(1,0){20}}
\put(110,20){\line(1,0){20}}
\multiput(120,0)(10,0){3}{\circle*{4}}
\multiput(110,10)(10,0){3}{\circle*{4}}
\multiput(100,20)(10,0){3}{\circle*{4}}
\put(115,-15){$\calD_2$}
%
\end{picture}
\end{center}
\begin{center}
\setlength{\unitlength}{0.4mm}
\begin{picture}(140,45)(10,-15)
\put(-30,5){$\mapsto$}
\put(-10,0){\line(0,1){20}}
\put(0,0){\line(0,1){20}}
\put(10,0){\line(0,1){20}}
\put(-10,0){\line(1,0){20}}
\put(-10,10){\line(1,0){20}}
\put(-10,20){\line(1,0){20}}
\multiput(-10,0)(10,0){3}{\circle*{4}}
\multiput(-10,10)(10,0){3}{\circle*{4}}
\multiput(-10,20)(10,0){3}{\circle*{4}}
\put(-5,-15){$\calD_3$}
\put(25,5){$\mapsto$}
\put(50,0){\line(0,1){20}}
\put(60,0){\line(0,1){20}}
\put(70,0){\line(0,1){20}}
\put(50,0){\line(1,0){20}}
\put(50,10){\line(1,0){20}}
\put(50,20){\line(1,0){20}}
\multiput(40,0)(10,0){3}{\circle*{4}}
\multiput(50,10)(10,0){3}{\circle*{4}}
\multiput(60,20)(10,0){3}{\circle*{4}}
\put(55,-15){$\calD_4$}
\put(85,5){$\mapsto$}
\put(82,-5){{\tiny $\operatorname{Ad}\Bor^{-1}$}}
\put(110,0){\line(0,1){20}}
\put(120,0){\line(0,1){20}}
\put(130,0){\line(0,1){20}}
\put(110,0){\line(1,0){20}}
\put(110,10){\line(1,0){20}}
\put(110,20){\line(1,0){20}}
\multiput(120,0)(10,0){3}{\circle*{4}}
\multiput(110,10)(10,0){3}{\circle*{4}}
\multiput(100,20)(10,0){3}{\circle*{4}}
\put(115,-15){$\calD_5$}
\put(145,5){$\mapsto$}
\put(170,0){\line(0,1){20}}
\put(180,0){\line(0,1){20}}
\put(190,0){\line(0,1){20}}
\put(170,0){\line(1,0){20}}
\put(170,10){\line(1,0){20}}
\put(170,20){\line(1,0){20}}
\multiput(170,0)(10,0){3}{\circle*{4}}
\multiput(170,10)(10,0){3}{\circle*{4}}
\multiput(170,20)(10,0){3}{\circle*{4}}
\put(175,-15){$\calD_6$}
\put(200,0){.}
\end{picture}
\end{center}
We see the Newton polygon of $\calD_3$ has also a rectangular shape.
More precisely, we have
\begin{equation*}
\calD_3=\calD \Big{|}_{\bigl\{d_2\to \frac{\mu q^2}{d_2},\ d_4\to \frac{1}{d_4 \mu },\ \Lambda \to \frac{d_2 d_4 \Lambda }{q},
\ x\to -\frac{d_2 x}{q}\bigr\}}.
\end{equation*}
This corresponds to the factorization property Proposition \ref{prop:coupled}
which will be discussed in the next section.
\end{Remark}

Since the parameter $u$ is stable under the adjoint action, we can regard it as a trivial free parameter.
Setting $u=0$ the quantum Seiberg--Witten curve is explicitly
\begin{align}\label{5Dcurve}
\calD_{\rm SW}(\Lambda,x) ={}&
(d_1 d_2 x- \mu q ) \left(1-\frac{\Lambda }{q x}\right)p
- (d_1+d_2) x -\frac{(d_3+d_4) \Lambda \mu}{x}
\nonumber \\
&+(x-1) \left(1-\frac{d_3 d_4 \Lambda \mu}{x}\right)\frac{1}{p}\nonumber \\
={}& x\frac{(1-d_1 p)(1-d_2 p)}{p}
- (1+ d_3 d_4 \Lambda \mu)\frac{1}{p} - \left(\mu q +\frac{d_1 d_2\Lambda}{q}\right) p\nonumber \\
&+ \frac{\Lambda \mu}{x} \frac{(p-d_3)(p-d_4)}{p}.
\end{align}
This is a non-commutative Laurent polynomial in $(x,p)\colon px=q x p$.
Note that the highest and lowest terms in $x$ and $p$ are all factorized.
When $(x,p)$ are commutative, the curve \eqref{5Dcurve} reduces to
the $M$-theoretic curve, which is obtained from the toric diagram,
or the five-brane web \cite{Bonelli:2017gdk,Brandhuber:1997ua,Eguchi:2000fv,Kim:2014nqa}.

\subsection{Four-dimensional limit}

We study the $q \to 1$ limit of the operator \eqref{Shakirov}.
Let us use the short hand notation such as
\[
F^{a,b,b}_c=T_{q,a}T_{q,b}^2T_{q,c}^{-1} F=F\bigl(qa,q^2b,q^{-1}c, \ldots\bigr),
\]
to represent the $q$-shifts of a function $F=F(a,b,c,\ldots)$.

\begin{Proposition}
We have
\begin{gather}
(1-\Lambda) (\SS)^{\Lambda}_{d_1}=(1-x)(\SS)_{x}+d_2(x-\Lambda)\SS p,\qquad d_1 \leftrightarrow d_2,\nonumber\\
(1-\Lambda) (\SS)^{\Lambda}_{d_3}=\left(1-\dfrac{\Lambda}{x}\right)\SS+d_4\left(\dfrac{\Lambda}{x}-\Lambda\right)p^{-1}\SS,
\qquad d_3 \leftrightarrow d_4,\label{eq:S-cont}
\end{gather}
and
\begin{equation}\label{eq:S-cont-x}
\SS x=q^2 x (\SS)^{d_1,d_2}_{d_3, d_4} p^2.
\end{equation}
\end{Proposition}

\begin{proof}
We write $\SS$ as $\SS=A_1 \Bor A_2 \Bor A_3$, where
\begin{align*}
A_1 &=\frac{1}{\varphi(qx)\varphi(\Lambda/x)}, \qquad
A_2=\frac{\varphi(\Lambda)\varphi\bigl(q^{-1} d_1d_2d_3d_4\Lambda\bigr)}
{\varphi(-d_1x)\varphi(-d_2x)\varphi(-d_3 \Lambda/x)\varphi(-d_4 \Lambda/x)}, \\
A_3 &= \frac{1}{\varphi\bigl(q^{-1}d_1d_2x\bigr)\varphi(d_3d_4 \Lambda/x)}.
\end{align*}
Then we have
\begin{equation*}
{A_1}^{\Lambda,x}_{d_1}=(1-q x)A_1, \qquad
{A_2}^{\Lambda,x}_{d_1}=\dfrac{1+d_2 x}{1-\Lambda}A_2, \qquad
{A_3}^{\Lambda,x}_{d_1}=A_3.
\end{equation*}
Using these relations, we have
\begin{align*}
(\SS)^{\Lambda,x}_{d_1} \SS^{-1}&={A_1}^{\Lambda,x}_{d_1}\Bor{A_2}^{\Lambda,x}_{d_1}{A_2}^{-1}\Bor^{-1}A_1^{-1}
=\dfrac{1-q x}{1-\Lambda}A_1(1+d_2 p x)A_1^{-1} \\
&=\dfrac{1-q x}{1-\Lambda}\Biggl( 1+q d_2 x\dfrac{1-\dfrac{\Lambda}{qx}}{1-q x}p\Biggr),
\end{align*}
hence
\[
(1-\Lambda)(\SS)^{\Lambda,x}_{d_1}=(1-q x)\SS+q d_2 x\left(1-\dfrac{\Lambda}{qx}\right) (\SS)^x p.
\]
Putting $x \to x/q$, we obtain the first relation of \eqref{eq:S-cont}.
Similarly, the second relation of \eqref{eq:S-cont} follows from
\[
{A_1}^{\Lambda}_{d_3}=\left(1-\dfrac{\Lambda}{x}\right)A_1, \qquad
{A_2}^{\Lambda}_{d_3}=\dfrac{1+d_4\dfrac{\Lambda}{x}}{1-\Lambda}A_2, \qquad
{A_3}^{\Lambda}_{d_3}=A_3.
\]
The relation \eqref{eq:S-cont-x} follows easily as
\begin{align*}
\SS x&=A_1\Bor A_2 \Bor A_3 x
=A_1\Bor A_2 p x \Bor A_3
=q A_1 \Bor x A_2 \Bor {A_3}^x p =q A_1 p x \Bor A_2 \Bor {A_3}^x p
\\
&=q^2 x {A_1} \Bor {A_2}^{x} \Bor {A_3}^{x,x} p^2 =q^2 x (\SS)^{d_1,d_2}_{d_3, d_4} p^2.\tag*{\qed}
\end{align*} \renewcommand{\qed}{}
\end{proof}

\begin{Theorem} \label{th:4d} We put $q=e^h$, $d_i={\rm e}^{h m_i}$, then we have
\begin{align*}
&\SS\to 1+h H_{\rm 4d}+O\bigl(h^2\bigr), \qquad h \to 0, \\
&H_{\rm 4d}=\vartheta_x(\vartheta_x+1)+\dfrac{\Lambda-x}{1-\Lambda}
(\vartheta_x+m_1)(\vartheta_x+m_2)+
\dfrac{\Lambda}{x}\dfrac{x-1}{1-\Lambda}(\vartheta_x-m_3)(\vartheta_x-m_4).
\end{align*}
\end{Theorem}

\begin{proof}
First, we consider the expansion
\[
v\seteq \SS.1=v_0+h v_1+O\bigl(h^2\bigr).
\]
Obviously $v_0=1$. For the first order term $v_1$, we have from \eqref{eq:S-cont}
\begin{gather*}
(1-\Lambda) \{({v_1}|_{m_1\to m_1-1})-v_1\}=m_2(x-\Lambda)v_0,\qquad m_1 \leftrightarrow m_2,\\
(1-\Lambda) \{({v_1}|_{m_3\to m_3-1})-v_1\}=\frac{\Lambda}{x}m_4(1-x)v_0, \qquad m_3 \leftrightarrow m_4,
\end{gather*}
hence we have
\[
v_1=\dfrac{\Lambda-x}{1-\Lambda}m_1m_2+
\dfrac{\Lambda}{x}\dfrac{x-1}{1-\Lambda}m_3m_4.
\]
Then from this and iterative use of \eqref{eq:S-cont-x}, we obtain
\begin{align*}
\SS x^n={}&q^{n(n+1)}x^n \big\{\bigl(T_{q,d_1}T_{q,d_2}T^{-1}_{q,d_3}T^{-1}_{q,d_4}\bigr)^n \SS\big\}p^n.1\\
={}&q^{n(n+1)} x^n \bigl(T_{q,d_1}T_{q,d_2}T^{-1}_{q,d_3}T^{-1}_{q,d_4}\bigr)^n v \\
={}&1+ hn(n+1)+h \left( \dfrac{\Lambda-x}{1-\Lambda}(n+m_1)(n+m_2)+
\dfrac{\Lambda}{x}\dfrac{x-1}{1-\Lambda}(n-m_3)(n-m_4)\right)\\
&+O\bigl(h^2\bigr).
\end{align*}
This is the desired result.
\end{proof}

In a similar way, we can compute higher order corrections for $v=\SS.1$ as
\[
v=\exp\left( \frac{h}{1-\Lambda} C_1+\frac{h^2}{(1-\Lambda)^2} C_2+\frac{h^3}{(1-\Lambda )^3}C_3+O\bigl(h^4\bigr)\right),\\
\]
where
\begin{align*}
C_1={}&m_1 m_2 (\Lambda -x)+\frac{\Lambda}{x}m_3 m_4 (x-1),\\
C_2={}&\frac{(1-x)}{2x} \bigg[\frac{\Lambda}{x} m_3 m_4 \left( x+\Lambda -(m_3+m_4) (x-\Lambda )\right)
\\
&-m_1 m_2 x \left( x+\Lambda +(m_1+m_2) (x-\Lambda )\right) \bigg],\\
C_3={}&(x-1) (x-\Lambda)\bigg[
m_1 m_2
\left(-\frac{\Lambda ^2+\Lambda +8 x^2+3 \Lambda x-x}{12 (x-\Lambda )}\right.
\\
&\left. +\frac{(m_1+m_2) \bigl(\Lambda ^2+\Lambda +\Lambda x+x-4 x^2\bigr)}{4 (x-\Lambda )}
+\frac{(m_1+m_2){}^2 (1+\Lambda -2x)}{6} \right.\\
&\left.-\frac{m_1 m_2 (1+\Lambda +4 x)}{12}
\right)
+m_3 m_4
\left(
\frac{\Lambda \bigl(\Lambda x^2+x^2+8 \Lambda ^2+3 \Lambda x-\Lambda ^2 x\bigr)}{12 x^3 (x-\Lambda )}\right.\\
&\left.
-\frac{\Lambda (m_3+m_4) \bigl(\Lambda x^2+x^2+\Lambda ^2 x+\Lambda x-4 \Lambda ^2\bigr)}{4 x^3 (x-\Lambda )}
+\frac{\Lambda (m_3+m_4){}^2 \bigl(1+\Lambda-2 \frac{\Lambda}{x}\bigr)}{6x^4}\right.\\
&\left.
-\frac{\Lambda m_3 m_4 \bigl(1+\Lambda+4 \frac{\Lambda}{x}\bigr)}{12 x^4}\right)
+\frac{\Lambda m_1 m_2 m_3 m_4}{x}\bigg].
\end{align*}

\begin{Remark}
The quantum Seiberg--Witten curve for four-dimensional $\mathcal{N}=2$ ${\rm SU}(2)$ gauge theory with $N_f=4$ is
equivalent to the quantum $P_{\rm VI}$ equation $\mathcal{D}^{4d}_{\rm SW} \Psi(\Lambda, x) =0$,
where $\mathcal{D}^{4d}_{\rm SW}$ can be written as \cite{AFKMY}
\[
\mathcal{D}^{4d}_{\rm SW}=(1-\Lambda)\Lambda \partial_{\Lambda}+D_{\rm Heun} ,
\]
and the Heun operator is
\[
D_{\rm Heun}=(1-\Lambda)(v-a_1)(v-a_2)-\frac{\Lambda}{x}(1-x)(v-\mu_1)(v-\mu_2)-(x-\Lambda)(v-\mu_3)(v-\mu_4), \\
\]
where $v=b x \partial_x$ and $b^2 = \epsilon_1/\epsilon_2$.
The result of Theorem \ref{th:4d} is consistent with the correspondence of the quantum $P_{\rm VI}$ equation
and ${\rm SU}(2)$ Seiberg--Witten theory with $N_f=4$ in four dimensions.
\end{Remark}

In five dimensions there are two mutually commuting Hamiltonians, one of which requires the infinite product to
generate the discrete time evolution, the other is related to the conserved quantities and takes a simple expression (like the
relativistic affine Toda Hamiltonian).
We have explicitly seen this is the case in the decoupling limit of the hypermultiplets (the pure Yang--Mills case).
In four-dimensional limit these two Hamiltonians degenerate to a single Hamiltonian, which is obtained from the Seiberg--Witten curve.

\section{Relation to affine Laumon space}\label{Laumon}

\subsection{Factorization as a coupled system}

We may use the same gauge transformation as \eqref{StoS} but exchanging $T_1$ and $T_2$ in the gauge factor.
By the dictionary in Section \ref{sec:dictinary}, we can see this is nothing
but the action of the Weyl reflection $r_1$ (see Appendix~\ref{App.A}).
By this $r_1$-reflected gauge transformation, we obtain the following Hamiltonian
from \eqref{eq:Shiraishi} by exchanging $T_1$ and $T_2$:
\begin{align*}
\widetilde{\mathcal{H}}:={} &
{1\over
\varphi\bigl(T_2 q^{1/2}t^{1/2} x\bigr)
\varphi\bigl(T_3 q^{1/2}t^{1/2} \Lambda x^{-1}\bigr) }
 \Bor \\
& \times
{\varphi(t T_2T_3 \Lambda) \varphi(q T_1T_4 \Lambda) \over
\varphi(-T_1T_2 x)\varphi\bigl(-Q^{-1} x\bigr)
\varphi\bigl(-T_3T_4 Q q t \Lambda x^{-1}\bigr)
\varphi\bigl(-q \Lambda x^{-1}\bigr)} \Bor \\
& \times T_{q t Q,x}^{-1}T_{t,\Lambda}^{-1}
{1\over
\varphi\bigl(T_1 q^{1/2}t^{1/2} x\bigr)
 \varphi\bigl(T_4 q^{1/2}t^{1/2} \Lambda x^{-1} \bigr)}.
\end{align*}
It turns out that this exchange of $T_1$ and $T_2$ is better
for the purpose of factorizing the original non-stationary difference equation.\footnote{The virtue of
the gauge transformation \eqref{StoS} is that it cancels $\varphi\bigl(q^{1/2}t^{1/2}T_2x\bigr)$
appearing in \eqref{2Painleve} so that the total gauge factor is written in terms the double infinite product $\Phi(x)$ only,
as we have seen in the beginning of Section \ref{4dlimit}.}

Our main point is that one can transform the non-stationary difference equation
$\widetilde{\mathcal{H}} \mathsf{V}^{(1)}=\mathsf{V}^{(1)}$
to the following coupled system:
\begin{align}\label{coupled1}
\mathsf{V}^{(1)}={}&
{\Phi\bigl(qt^{-1} b_2/b_4\bigr)\Phi(b_1/b_3)
\over \Phi(t b_6/b_8)\Phi(q b_5/b_7) }
{1\over \varphi(- q b_6/G)\varphi(- G/b_8)}
 \widetilde{\Bor} \nonumber \\
& \times
{1\over \varphi\bigl(\mathsf{p}^{-2}q b_2/G\bigr)\varphi\bigl(\mathsf{p}^{-2} G/b_4\bigr)}
\bigl(T_{t^{1/2},x}^{-1}T_{t,\Lambda}^{-1}\bigr)\mathsf{V}^{(2)},\\
\label{coupled2}
\mathsf{V}^{(2)}={}&
{\Phi(t b_6/b_8)\Phi(qb_5/b_7)\over \Phi(q b_2/b_4)\Phi(t b_1/b_3) }
{1\over \varphi(q b_1/G)\varphi(G/b_3)}
 \widetilde{\Bor} \nonumber \\
& \times
{1\over \varphi(-q b_5/G)\varphi(- G/b_7)}\mathsf{V}^{(1)},
\end{align}
where $\widetilde{\Bor} := T_{(q t^{1/2}Q)^{1/2},x}^{-1}\Bor$.
Note that by making use of the equality
\begin{align*}
{\Phi\bigl(qt^{-1} b_2/b_4\bigr)\Phi(b_1/b_3)
\over \Phi(t b_6/b_8)\Phi(q b_5/b_7) } T_{t, \Lambda}^{-1}
{\Phi(t b_6/b_8)\Phi(qb_5/b_7)\over \Phi(q b_2/b_4)\Phi(t b_1/b_3) }&= \varphi(b_6/b_8) \varphi\bigl(qt^{-1}b_5/b_7\bigr) T_{t, \Lambda}^{-1} \\
&= \varphi(tT_2T_3\Lambda) \varphi(q T_1T_4 \Lambda) T_{t, \Lambda}^{-1},
\end{align*}
we have called back the double infinite product $\Phi(z)$, which we once eliminated
to reveal the relation to the $qq$-Painlev\'e VI equation, to factorize the original equation as a coupled system.
The possibility of such a factorization was already suggested in the proof of Conjecture \ref{conjecture}
in the special case of the ``Macdonald'' limit \cite{Shakirov:2021krl}.
We believe this factorization is a significant step towards a general proof of Conjecture \ref{conjecture}.
By the dictionary \eqref{dictionary} in Section \ref{sec:dictinary} and Definition \ref{b-variables} of the variables $b_i$,
it is straightforward to check the matching of parameters $\bigl(qt^{-1}b_2/b_4, b_1/b_3, t b_6/b_8,q b_5/b_7\bigr)$
to the parameters $\bigl(qt T_1 T_2 \Lambda, t^2 T_3 T_4 \Lambda, t^2 T_2 T_3 \Lambda, qt T_1 T_4 \Lambda\bigr)$
on the gauge theory side.
Note that the parameters $b_i$ involve neither $x$ nor $\Lambda$.
On the other hand, $G$ is a~monomial in $x$ and $\Lambda$. But it is also easy to see the matching of
$(-qb_6/\!G, -G\!/b_8, -q b_5/\!G, -G\!/b_7)$ to $\bigl(q^{1/2} t^{1/2} T_2 x, q^{1/2} t^{1/2} T_3 \Lambda/x,
q^{1/2} t^{1/2} T_1 x, q^{1/2} t^{1/2} T_4 \Lambda/x\bigr)$.
For remaining parameters we have to take the commutation with the shift operators into account.
Namely the parameters $\mathsf{p}^{-2} q b_2/G$ and $\mathsf{p}^{-2} G/b_4$ are affected by $T_{(qt^{1/2}Q)^{1/2},x}^{-1}$
and for $q b_1/G$ and $G/b_3$ the action of $T_{t^{1/2},x}^{-1}T_{t,\Lambda}^{-1}$ is also involved.

To rewrite the coupled system \eqref{coupled1} and \eqref{coupled2} to more symmetric one
in $\mathsf{V}^{(1)}$ and $\mathsf{V}^{(2)}$, let us introduce
the following shift operator $\widetilde{T}_{\mathsf{p},b}$ which commutes with $\widetilde{\Bor}$:
\begin{align}\label{widetildeT}
\widetilde{T}_{\mathsf{p},b} :={}& T_{-t^{-1/4},x}T_{t^{-1/2},\Lambda} \colon\ \nonumber \\
&b\mapsto \!\bigl( -\mathsf{p}^{-1}b_1, - \mathsf{p}^{-1}b_2, - \mathsf{p}b_3, -\mathsf{p}b_4,
- \mathsf{p}^{-1}b_5, -\mathsf{p}^{-1}b_6, - \mathsf{p}b_7, -\mathsf{p}b_8\!\bigr),
\end{align}
where $\mathsf{p} = {\rm e}^{\delta} = a_0 a_1 a_2^2 a_3^2 a_4 a_5 = t^{1/4}$.
Note that $G$ is invariant under $\widetilde{T}_{\mathsf{p},b}$. For the transformation~\eqref{widetildeT} of the parameters $b_i$,
the shift operator $T_{t^{-1/2},\Lambda}$ is enough, but to make $G$ invariant up to sign we need the combination with $T_{-t^{-1/4},x}$.
We also note that \smash{$\bigl(\widetilde{T}_{\mathsf{p},b}\bigr)^2 = T$} (the discrete time evolution) as far as the $b$ variables are concerned
(see Lemma \ref{lemA8}). Then we have
\begin{Proposition}\label{prop:coupled}
We can rewrite \eqref{coupled1} and \eqref{coupled2} as follows:
\begin{align}
\label{firsteq}
&\mathsf{V}^{(1)}=
{\Phi\bigl(q t^{-1} b_2/b_4\bigr)\Phi(b_1/b_3)
\over \Phi( t b_6/b_8)\Phi(q b_5/b_7) }
{1\over \varphi(- q b_6/G)\varphi(- G/b_8)}
 \nonumber \\
&\hphantom{\mathsf{V}^{(1)}=}{} \times (\widetilde{\Bor} \widetilde{T}_{\mathsf{p},b} )
{1\over \varphi\bigl(-\mathsf{p}^{-1} q b_2/G\bigr)\varphi\bigl(-\mathsf{p}^{-1} G/b_4\bigr)}
\widetilde{T}_{\mathsf{p},b} \mathsf{V}^{(2)},\\
\label{secondeq}
&\widetilde{T}_{\mathsf{p},b} \mathsf{V}^{(2)}=
{
\Phi\bigl(\mathsf{p}^{-2} t b_6/b_8\bigr)\Phi\bigl(\mathsf{p}^{-2}q b_5/b_7\bigr)\over
 \Phi\bigl(\mathsf{p}^{-2} q b_2/b_4\bigr)\Phi\bigl( \mathsf{p}^{-2} t b_1/b_3\bigr)
}
{1\over \varphi\bigl(-\mathsf{p}^{-1}q b_1/G\bigr)\varphi\bigl(-\mathsf{p}^{-1}G/b_3\bigr)} \nonumber \\
&\hphantom{\widetilde{T}_{\mathsf{p},b} \mathsf{V}^{(2)}=}{}\times \bigl(\widetilde{\Bor} \widetilde{T}_{\mathsf{p},b} \bigr)
{1\over \varphi(-q b_5/G)\varphi(-G/b_7)}\mathsf{V}^{(1)}.
\end{align}
The coupled system is gauge equivalent to the non-stationary difference equation \eqref{qq-PVI}.
\end{Proposition}

Recall that the time evolution $T$ of the discrete Painlev\'e VI equation is
a translation element in the extended affine Weyl group of $D_5^{(1)}$.
It is remarkable that $T$ allows a square root
\begin{equation}\label{Eq.(6.8)}
T :=r_2 r_1 r_0r_2 \sigma_{01}r_3r_4r_5r_3\sigma_{45}
= (r_2 r_1 r_0r_2 \sigma_{01} \tau)( r_2 r_1 r_0r_2 \sigma_{01} \tau).
\end{equation}
In fact, the factorization into \eqref{firsteq} and \eqref{secondeq} is not unrelated to
the existence of the square root $T^{1/2} := (r_2 r_1 r_0r_2 \sigma_{01} \tau)$, which
acts on the $b$ variables as follows:
\begin{equation}\label{squareroot}
T^{1/2}=(r_2 r_1 r_0r_2 \sigma_{01} \tau)\colon\ b\mapsto
\bigl(b_6,b_5,b_8,b_7,\mathsf{p}^{-2}b_2, \mathsf{p}^{-2}b_1, \mathsf{p}^{2}b_4, \mathsf{p}^{2}b_3\bigr).
\end{equation}
We define an operator
\[
X\colon\ b\mapsto
\bigl(\mathsf{p}b_6,\mathsf{p} b_5,\mathsf{p}^{-1}b_8, \mathsf{p}^{-1}b_7,
\mathsf{p}^{-1}b_2, \mathsf{p}^{-1}b_1, \mathsf{p}b_4, \mathsf{p}b_3\bigr),
\]
and assume that $X$ does not act on $(F,G)$.
Then we have

\begin{Lemma}
The action \eqref{squareroot} on the parameters $b_i$ is represented by
\[
T^{1/2} = (-1) X \widetilde{T}_{\mathsf{p},b},
\]
where $(-1)$ is the overall sign flip of $b_i$.
Two functions $\mathsf{V}^{(1)}$ and $\mathsf{V}^{(2)}$ are related by
\begin{equation}\label{F1F2}
X \mathsf{V}^{(1)} = \widetilde{T}_{\mathsf{p},b} \mathsf{V}^{(2)}.
\end{equation}
\end{Lemma}
Hence, the coupled system in the form of Proposition~\ref{prop:coupled} is quite natural from the viewpoint of
the $qq$-Painlev\'e VI equations, because the second equation \eqref{secondeq} is obtained by applying $X$ to \eqref{firsteq}.


\subsection{Instanton counting with a surface defect}

Now we want to point out that solutions to the coupled system \eqref{firsteq} and \eqref{secondeq} are given by the instanton partition function
of the affine Laumon space. In fact, we have already mentioned that the gauge transformation introduced in the beginning of Section \ref{4dlimit}
is a five-dimensional uplift of the gauge transformation from the Higgsed quiver gauge theory to the gauge theory with a surface defect.
A torus action on the affine Laumon space is induced by the standard torus action on $\mathbb{P}^1 \times \mathbb{P}^1$.
The fixed points of the torus action on the affine Laumon space of type~\smash{$A_{r}^{(1)}$} are labelled by $(r+1)$-tuples of partitions \cite{FFNR}.

\begin{Definition}
Set
\begin{align*}
[u;q]_n&=u^{-n/2}q^{-n(n-1)/4} (u;q)_n\\
&= \bigl(u^{-1/2}-u^{1/2}\bigr)\bigl(q^{-1/2}u^{-1/2}-q^{1/2}u^{1/2}\bigr)\cdots
\bigl(q^{-(n-1)/2}u^{-1/2}-q^{(n-1)/2}u^{1/2}\bigr).
\end{align*}
\end{Definition}

For a pair $(\lambda, \mu)$ of partitions, the $\mathbb{Z}_N$ orbifolded Nekrasov factor
with color $k$ is\footnote{We associate a $\sinh$ factor
with each monomial term in the equivariant character.}
\begin{align*}
\Nk^{(k|N)}_{\lambda,\mu}(u|q,\kappa)=
\Nk^{(k)}_{\lambda,\mu}(u|q,\kappa)={}&
 \prod_{j\geq i\geq 1 \atop j-i \equiv k ({\rm mod}\,N)}
\big[u q^{-\mu_i+\lambda_{j+1}} \kappa^{-i+j};q\big]_{\lambda_j-\lambda_{j+1}}\\
&\times
\prod_{\beta\geq \alpha \geq 1 \atop \beta-\alpha \equiv -k-1 ({\rm mod}\,N)}
\big[u q^{\lambda_{\alpha}-\mu_\beta} \kappa^{\alpha-\beta-1};q\big]_{\mu_{\beta}-\mu_{\beta+1}}. \nonumber
\end{align*}
Note that the equivariant parameters of the torus action on $\mathbb{P}^1 \times \mathbb{P}^1$ are
not $(q,t)$, but $(q, \kappa)$. We will substitute $\kappa= t^{-\frac{1}{2}}$ later.\footnote{The
square root comes from the $\mathbb{Z}_2$ orbifolding.}

From the equivariant character evaluated at each fixed point of the affine Laumon space of type $A_{1}^{(1)}$ \cite{FFNR}, we obtain
\begin{align}\label{Nf4}
&f(u_1,u_2;v_1,v_2;w_1,w_2|x_1,x_2|q,\kappa)
=f\left(  \begin{matrix}u_1,u_2 \\v_1,v_2\\w_1,w_2\end{matrix} \,\bigg|\,
x_1,x_2 \,\bigg|\,q,\kappa \right) \nonumber \\
&\quad=
\sum_{\lambda^{(1)},\lambda^{(2)}\in \mathsf{P}}
\prod_{i,j=1}^2
{\mathsf{N}^{(j-i|2)}_{\varnothing,\lambda^{(j)}} (u_i/v_j |q,\kappa)
\mathsf{N}^{(j-i|2)}_{\lambda^{(i)},\varnothing} (v_i/w_j|q,\kappa) \over
\mathsf{N}^{(j-i|2)}_{\lambda^{(i)},\lambda^{(j)}} (v_i/v_j |q,\kappa)}
 x_1^{|\lambda^{(1)}|_o+|\lambda^{(2)}|_e} x_2^{|\lambda^{(1)}|_e+|\lambda^{(2)}|_o},
\end{align}
where $\varnothing$ denotes the empty partition and for a partition $\lambda = (\lambda_1 \geq \lambda_2 \geq \cdots)$,
we set
\[
|\lambda|_o \seteq \sum_{k\geq 1} \lambda_{2k-1},
\qquad
|\lambda|_e \seteq \sum_{k\geq 1} \lambda_{2k}.
\]
Note that the function $f(u_1,u_2;v_1,v_2;w_1,w_2|x_1,x_2|q,\kappa)$ is invariant under
the overall scaling of the equivariant parameters $(u_1,u_2;v_1,v_2;w_1,w_2)$.
The partition function \eqref{Nf4} is a five-dimensional uplift of the instanton partition function
which is given, for example, in \cite{AFKMY}.
The parameters $(v_1, v_2)$ are the Coulomb moduli of ${\rm U}(2)$ gauge theory, or the equivariant
parameters of the Cartan subgroup ${\rm U}(1) \times {\rm U}(1) \subset {\rm U}(2)$.
The parameters $(u_1, u_2)$ and $(w_1, w_2)$ are exponentiated mass parameters
of the hypermultiplets in the fundamental and the anti-fundamental representations.
They are also regarded as equivariant parameters for the flavor symmetry.
The expansion parameters are parametrized as $x_1=x$, $x_2=\Lambda/x$,
where $x$ counts the monopole number (the first Chern number of the ${\rm U}(1)$ connection
on the defect), while $\Lambda$ counts the instanton number (the second Chern number).
When $\Lambda=0$, the terms with \smash{$\big|\lambda^{(1)}\big|_e+\big|\lambda^{(2)}\big|_o \neq 0$} do not contribute
to the partition function. This means the sum in \eqref{Nf4} is restricted to \smash{$\lambda^{(1)}=(m)$}
(a partition with a single row) and \smash{$\lambda^{(2)}=\varnothing$}. Physically this corresponds to
the topological sector with instanton number zero. As we see in Appendix~\ref{App.B},
the partition function is given by the Heine's $q$-hypergeometric series.
The six parameters $(u_1, u_2)$, $(w_1, w_2)$ and $(x_1, x_2)$ are ``external''
spectral parameters which correspond to the (independent) dynamical variables on Painlev\'e side.
On the other hand, the parameters $(v_1,v_2)$ are ``internal'' parameters
or the loop parameters.

The function \eqref{Nf4} should be compared with
the non-stationary Ruijsenaars function \cite{Shi}:
\begin{align}\label{nsR}
f^{\widehat{\mathfrak{ gl}}_N}(x,p|s,\kappa|q,t)={}&
\sum_{\lambda^{(1)},\ldots,\lambda^{(N)}\in {\mathsf P}}
\prod_{i,j=1}^N
{\Nk^{(j-i|N)}_{\lambda^{(i)},\lambda^{(j)}} (ts_j/s_i|q,\kappa) \over \Nk^{(j-i|N)}_{\lambda^{(i)},\lambda^{(j)}} (s_j/s_i|q,\kappa)}\nonumber\\
&\times \prod_{\beta=1}^N\prod_{\alpha\geq 1} ( p x_{\alpha+\beta}/x_{\alpha+\beta-1})^{\lambda^{(\beta)}_\alpha},
\end{align}
for $N=2$. Both functions come from the affine Laumon space of type $A_{1}^{(1)}$.
The non-stationary Ruijsenaars function \eqref{nsR} corresponds to the theory with an adjoint matter,
or the tangent bundle overt the affine Laumon space, while the partition function \eqref{Nf4}
is for the theory with four matter hypermultiplets in the (anti-)fundamental representation,
or the tautological bundle.
In the AGT correspondence, the former is identified with
the conformal block on a punctured torus and the latter on $\mathbb{P}^1$ with four punctures.
In the mass decoupling limit, both functions are conjectured to give solutions to the non-stationary affine Toda equation \cite{Shi}.

\begin{con}\label{secondconjecture}
The partition function \eqref{Nf4} gives a solution to the coupled system \eqref{firsteq} and~\eqref{secondeq}
by the following specialization of parameters:
\begin{align}
&\mathcal{F}^{(1)}=
f\left(  \begin{matrix}
q^{1/2}b_4/b_8,&q^{1/2} b_6/b_2\\
\bigl(\mathsf{p}^2Q\bigr)^{-1/2},&\bigl(\mathsf{p}^2Q\bigr)^{1/2}\\
q^{-1/2} b_2/b_5,&q^{-1/2}b_7/b_4\end{matrix} \,\Bigg|\,
q^{1/2} \mathsf{p}^{-1}\mathsf{t} G^{-1}, q^{-1/2} \mathsf{p}^{-1}\mathsf{t} G\,\Bigg|\,q,t^{-1/2}\right), \label{special1}\\
&\mathcal{F}^{(2)}=
f\left(  \begin{matrix}
q^{1/2}b_4/b_8,&q^{1/2} b_6/b_2\\
\bigl(\mathsf{p}^2Q\bigr)^{-1/2},&\bigl(\mathsf{p}^2Q\bigr)^{1/2}\\
q^{-1/2} b_1/b_6,&q^{-1/2}b_8/b_3\end{matrix} \,\Bigg|\,
- q^{1/2} \mathsf{t} G^{-1},- q^{-1/2} \mathsf{t} G\,\Bigg|\,q,t^{-1/2}\right). \label{special2}
\end{align}
\end{con}
It is remarkable that the only difference between $\mathcal{F}^{(1)}$ and $\mathcal{F}^{(2)}$
is the exchanges of $(b_1,b_2)$, $(b_3,b_4)$, $(b_5,b_6)$ and $(b_7,b_8)$ (see Figure~\ref{bvar})
in the specialization of $w_i$ and the scaling $-\mathsf{p}$ of $\mathsf{t}$.
Note that these exchanges of four pairs of $b_i$ are nothing but the action of
the Weyl reflections $r_4$, $r_5$, $r_1$ and $r_0$, respectively.
We can see that the specializations \eqref{special1} and \eqref{special2} are consistent with
the relation \eqref{F1F2}.
In fact, under the action of $X$
\begin{align*}
&\frac{b_4}{b_8} \mapsto \mathsf{p}^{-2}\frac{b_7}{b_3} = \frac{b_4}{b_8}, \qquad
\frac{b_2}{b_6} \mapsto \mathsf{p}^{-2}\frac{b_1}{b_5} = \frac{b_2}{b_6},\qquad
\frac{b_2}{b_5} \mapsto \mathsf{p}^{2}\frac{b_5}{b_2} = \frac{b_1}{b_6},\qquad
\frac{b_7}{b_4} \mapsto \mathsf{p}^{2}\frac{b_4}{b_7} = \frac{b_8}{b_3},
\end{align*}
where we have used the constraints \eqref{b-consts}.
On the other hand, these ratios are invariant by the action of $\widetilde{T}_{\mathsf{p},b}$.
Finally $\widetilde{T}_{\mathsf{p},b}$ generates the square root of the discrete time shift
$\mathsf{t} \mapsto -\mathsf{p}^{-1}\mathsf{t}$.

In terms of the variables on the gauge theory side, the specialization of parameters is
given as follows:
\begin{gather*}
\mathcal{F}^{(1)}\!
=\!f\!\left( \begin{matrix}
q^{1/2} t^{1/2} Q^{1/2} T_3,&q^{1/2} Q^{-1/2} T_1^{-1}\\
Q^{-1/2}&t^{1/2}Q^{1/2}\\
q^{-1/2} t^{1/2} Q^{1/2} T_2,&q^{-1/2} Q^{-1/2} T_4^{-1}\end{matrix} \,\Bigg| \,
{-}t^{1/2} T_1^{1/2}T_2^{1/2} x,-t^{1/2} T_3^{1/2}T_4^{1/2} \Lambda/x \,\Bigg|\, q,t^{-1/2}\!\right)\!, \\
 \mathcal{F}^{(2)}\!
=\!f\!\left(  \begin{matrix}
q^{1/2} t^{1/2} Q^{1/2} T_3,&q^{1/2} Q^{-1/2} T_1^{-1}\\
Q^{-1/2},&t^{1/2}Q^{1/2}\\
q^{-1/2} t^{1/2}Q^{-1/2} T_2^{-1},&q^{-1/2} t Q^{1/2} T_4\end{matrix} \,\Bigg|\,
t^{3/4} T_1^{1/2}T_2^{1/2} x,t^{3/4} T_3^{1/2}T_4^{1/2} \Lambda/x \,\Bigg|\,q,t^{-1/2}\!\right)\!,
\end{gather*}
where by making use of the scaling symmetry of the partition function \eqref{Nf4},
we have made the overall scaling of parameters by $\mathsf{p}$.
One can check that $\mathcal{F}^{(2)}$ is invariant under the exchange of~$T_1$ and $T_2$.
We have examined our conjecture in several cases. The results are summarized in Appendix \ref{App.B}.

Thus we can formulate Conjecture \ref{conjecture}
in terms of the instanton partition functions from the affine Laumon space.
According to the (four-dimensional) AGT correspondence,
the partition functions coming from the affine Laumon space should be
identified with the conformal blocks of the current algebra.
In the present case it is the affine algebra \smash{$A_1^{(1)} =\widehat{\mathfrak{sl}}_2$}.
In the original formulation in \cite{Shakirov:2021krl}, the five-dimensional
Nekrasov partition function are regarded as a~conformal block of
the deformed Virasoro algebra. In our formulation it is natural to expect
that the coupled system of non-stationary difference equations defines
a conformal block of the quantum deformation of \smash{$A_1^{(1)} =\widehat{\mathfrak{sl}}_2$},
namely the quantum affine algebra \smash{$U_q(\widehat{\mathfrak{sl}}_2)$}.
The advantage of~\smash{$U_q\big(\widehat{\mathfrak{sl}}_2\big)$} to the deformed Virasoro
algebra is that it is a~quantum group (in particular we have a~coproduct),
while the latter is not.

\appendix

\section{Difference analogue of Painlev\'e VI equation}
\label{App.A}

\subsection{Affine Weyl group and B\"acklund transformation}
Let $A=A\bigl(D^{(1)}_5\bigr)=(a_{ij})_{i,j=0}^5$ be the generalized Cartan matrix of type $D^{(1)}_5$
associated with the Dynkin diagram in Figure~\ref{D5}.

Fix a realization $\bigl(\mathfrak{h},\Pi,\Pi^\vee\bigr)$ of $A$,
where
$\Pi=\{\alpha_0,\ldots,\alpha_5\}\subset \mathfrak{h}^*$ denotes the set of simple roots,
$\Pi^\vee=\big\{\alpha_0^\vee,\ldots,\alpha_5^\vee\big\}\subset \mathfrak{h}$ the set of simple coroots,
with $\Pi$ and $\Pi^\vee$ being linearly independent.
We have
$\big\langle \alpha_i^\vee,\alpha_j\big\rangle =a_{ij}$ $(i,j=0,\ldots,5)$.
Let
$Q=\sum_{i=0}^5 \mathbb{Z} \alpha_i$ be the root lattice.
Denote by $\Delta$, $\Delta_+$ and $\Delta_-$ the sets of all roots, positive and negative roots respectively.
We have $\Delta=\Delta_+\cup \Delta_-$ (a disjoint union).
The center of \smash{$\mathfrak{g}=\mathfrak{g}\bigl(D^{(1)}_5\bigr)$} is 1-dimensional and is spanned by
the canonical central element
$
K= \alpha_0^\vee+\alpha_1^\vee+2 \alpha_2^\vee+2\alpha_3^\vee+\alpha_4^\vee+\alpha_5^\vee.
$
Denote by $\delta\in Q$ the null root
$\delta=\alpha_0+\alpha_1+2 \alpha_2+2\alpha_3+\alpha_4+\alpha_5$.
Let $d\in \mathfrak{h}$ be the scaling element (defined up to a~summand
proportional to $K$) satisfying
$\langle \alpha_i,d\rangle=0$ for $i=1,\ldots,5$ and $ \langle \alpha_0,d\rangle=1$.
The elements~$\alpha_0^\vee,\ldots,\alpha_5^\vee,d$ form a basis of $\mathfrak{h}$.
We have
$\mathfrak{g}=[\mathfrak{g},\mathfrak{g}]+\mathbb{C} d$, $\mathfrak{h}=\sum_{i=0}^5 \mathbb{C} \alpha_i^\vee+\mathbb{C} d$.
Define the element~$\Lambda_0\in \mathfrak{h}^*$ by
$\big\langle \Lambda_0, \alpha_i^\vee\big\rangle=\delta_{0i}$ for $i=0,\ldots,5$ and
$\langle \Lambda_0,d\rangle=0$.
The elements~$\alpha_0,\ldots,\alpha_5,\Lambda_0$ form a~basis of $\mathfrak{h}^*$
and we have $\mathfrak{h}^*=\sum_{i=0}^5 \mathbb{C} \alpha_i+\mathbb{C} \Lambda_0$.
The affine Weyl group~\smash{$W=W\bigl(D^{(1)}_5\bigr)$} is defined to be the group
generated by the fundamental reflections
$r_0,r_1,\ldots,r_5$ which act on~$\mathfrak{h}^*$~by
\begin{align*}
r_i(\lambda)=\lambda-\langle \lambda,\alpha_i^\vee\rangle \alpha_i, \qquad \lambda \in \mathfrak{h}^*.
\end{align*}

\begin{figure}[t]\centering
\vspace*{30mm}
\begin{picture}(40,30)\setlength{\unitlength}{1.5mm}

\put(0,10){\line(-1,1){7}}
\put(0,10){\line(-1,-1){7}}
\put(0,10){\line(1,0){10}}
\put(10,10){\line(1,1){7}}
\put(10,10){\line(1,-1){7}}

\put(-7.8,16.2){$\bullet$}
\put(-7.8,2.2){$\bullet$}

\put(-0.8,9.2){$\bullet$}
\put(9.4,9.2){$\bullet$}

\put(16.4,16.2){$\bullet$}
\put(16.4,2.2){$\bullet$}

\put(-10.8,16.2){$\scriptstyle 0$}
\put(-10.8,2.2){$\scriptstyle 1$}

\put(-0.8,6.2){$\scriptstyle 2$}
\put(9.4,6.2){$\scriptstyle 3$}

\put(19.4,16.2){$\scriptstyle 5$}
\put(19.4,2.2){$\scriptstyle 4$}

\put(-22,9.2){$ \sigma_{01}\,\,\,\updownarrow$}
\put(22,9.2){$ \updownarrow\,\,\, \sigma_{45}$}
\put(1.5,22.2){$ \longleftrightarrow$}
\put(4.0,25.2){$\tau$}

\end{picture}
\caption{Dynkin diagram of $D^{(1)}_5$ and its automorphism.}\label{D5}
\end{figure}
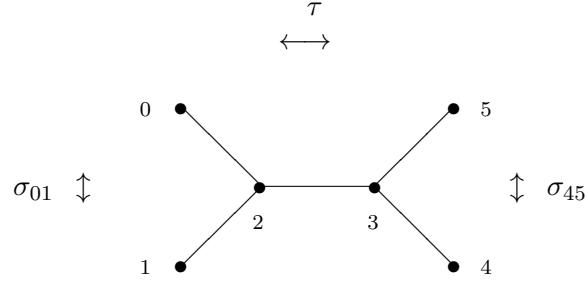

Let $a_i$ $(i=0,\ldots, 5)$ be the formal exponentials $a_i={\rm e}^{\alpha_i}$.
Write $a=(a_0,\ldots,a_5)$ for short.
Denote by $\mathbb{K}=\mathbb{C}(a)$ the field of rational functions in $a$.
Define the actions of the generators~${r_i\in W}$ on $\mathbb{K}$ by setting the rules
\begin{align*}
r_i \cdot a_j= a_j a_i^{-a_{ij}} ={\rm e}^{\alpha_j-a_{ij}\alpha_i}, \qquad 0\leq i, j\leq 5,
\end{align*}
and extending them as ring homomorphisms.
It is clear that these actions are compatible with the group structure,
namely they satisfy the Coxeter relations:
$r_i^2={\rm id}$ $(a_{ij}=2)$,
$r_ir_jr_i=r_jr_ir_j$ $(a_{ij}=-1)$, and
$r_ir_j=r_jr_i$ $(a_{ij}=0)$.
We regard $W$ as a group of birational isomorphisms of $\mathbb{K}$.

Let $\sigma_{01}$, $\sigma_{45}$ and $\tau$ be the automorphisms shown in Figure~\ref{D5}, namely
\begin{align*}
\begin{split}
&\sigma_{01}\colon\ (0,1,2,3,4,5)\mapsto (1,0,2,3,4,5), \qquad \sigma_{45}\colon\  (0,1,2,3,4,5)\mapsto (0,1,2,3,5,4), \\
&\tau\colon\ (0,1,2,3,4,5)\mapsto (5,4,3,2,1,0).
\end{split}
\end{align*}
Let $\widetilde{W}$ denotes the extended affine Weyl group generated by $W$
together with $\sigma_{01}$, $\sigma_{45}$ and $\tau$.
Let $f$ and $g$ be independent indeterminates, and consider the
rational function field \[\mathbb{K}(f,g) =\mathbb{C}(a)(f,g).\]

\begin{Definition}\label{birational-cl}
Define the actions of the generators $r_0,\ldots,r_5,\sigma_{01},\sigma_{45},\tau\in \widetilde{W}$ on
$\mathbb{K}(f,g)$
by setting the rules
\begin{align*}
&r_i \cdot a_j= a_j a_i^{-a_{ij}}, \qquad 0\leq i, j\leq 5,\\
& \sigma_{01} \cdot a_i=(a_{\sigma_{01}(i)})^{-1},
\qquad \sigma_{45} \cdot a_i=(a_{\sigma_{45}(i)})^{-1},
\qquad \tau \cdot a_i=(a_{\tau(i)})^{-1},\qquad 0\leq i\leq 5,\\
&r_i\cdot f=f,\qquad r_i\cdot g=g,\qquad i\neq 2,3,\\
&r_2\cdot f=f{a_0a_1^{-1}g+a_2^2\over a_0a_1^{-1}a_2^2g+1},\qquad r_2\cdot g=g,\qquad
r_3\cdot f=f,\qquad r_3\cdot g=
{a_3^2 a_4a_5^{-1} f+1\over a_4a_5^{-1}f+a_3^2}g,\\
&\sigma_{01}\cdot f=f^{-1},\qquad \sigma_{01}\cdot g=g,\qquad
\sigma_{45}\cdot f=f,\qquad \sigma_{45}\cdot g=g^{-1},\\
&\tau \cdot f=g,\qquad \tau \cdot g=f,
\end{align*}
and extending them as ring homomorphisms.
\end{Definition}

\begin{Proposition}\label{Cox}
The actions are compatible with the group structure of the
extended affine Weyl group $\widetilde{W}$. Namely they satisfy
the Coxeter relations:
\begin{align*}
&r_i^2={\rm id} \qquad \mbox{if }\quad a_{ij}=2,\\
&r_ir_jr_i=r_jr_ir_j \qquad \mbox{if }\quad a_{ij}=-1,\qquad
r_ir_j=r_jr_i \qquad \mbox{if }\quad a_{ij}=0,\\
& \sigma_{01}^2=\sigma_{45}^2=\tau^2={\rm id} ,\qquad
\sigma_{01}\sigma_{45}=\sigma_{45}\sigma_{01},
\qquad \sigma_{01}\tau=\tau\sigma_{45},\\
&\sigma_{01} r_0=r_1 \sigma_{01} ,\qquad \sigma_{01} r_i=r_i \sigma_{01}, \qquad i\neq 0,1,\\
&\sigma_{45} r_4=r_5 \sigma_{45} ,\qquad \sigma_{45} r_i=r_i \sigma_{45}, \qquad i\neq 4,5,\\
&\tau r_i=r_{5-i} \tau.
\end{align*}
\end{Proposition}

\subsection{Difference analogue of Painlev\'e VI equation}

Set for simplicity of display that
\begin{align*}
\mathsf{p}=a_0a_1a_2^2a_3^2a_4a_5=e^\delta,\qquad \mathsf{t}= a_3^2a_4a_5.
\end{align*}
Recall that the element
\begin{align*}
T=r_2r_1r_0r_2 \sigma_{01} r_3r_4r_5r_3 \sigma_{45}\in \widetilde{W}, \label{Trans}
\end{align*}
acts on $\mathfrak{h}^*$ as a translation. Namely, we have
\begin{align*}
T\colon\ (a_0,a_1,a_2,a_3,a_4,a_5) \mapsto \bigl(a_0,a_1,\mathsf{p} a_2,\mathsf{p}^{-1}a_3,a_4,a_5\bigr).
\end{align*}
Write $\overline{f}=T\cdot f$ and $\underline{g}=T^{-1}\cdot g$ for short.
\begin{Proposition}
We have the difference analogue of the Painlev\'e VI equation
\begin{align*}
&f \overline{f}=\mathsf{p}^2\mathsf{t}^{-2}
{g+ \mathsf{t}\mathsf{p}^{-1}a_1^2 \over
g+ \mathsf{t}^{-1}\mathsf{p}a_0^2}
{g+ \mathsf{t}\mathsf{p}^{-1}a_1^{-2} \over
g+ \mathsf{t}^{-1}\mathsf{p}a_0^{-2}},\\
&g \underline{g}=\mathsf{t}^{-2}
{f+ \mathsf{t}a_4^2 \over
f+ \mathsf{t}^{-1}a_5^2}
{f+ \mathsf{t}a_4^{-2} \over
f+ \mathsf{t}^{-1}a_5^{-2}}.
\end{align*}
\end{Proposition}

\begin{Proposition}
The translation $T$ admits the symmetry of the type $D^{(1)}_4$
in the following sense.
\begin{itemize}\itemsep=0pt
 \item[$(1)$] The translation $T$ commutes with the
action of the subgroup
\[\langle r_0,r_1,r_2r_3r_2=r_3r_2r_3,r_4,r_5 \rangle \simeq W\bigl(D^{(1)}_4\bigr) \qquad \text{of}\quad W\bigl(D^{(1)}_5\bigr),\]
\item[$(2)$] We have $ \sigma_{01}T=T\sigma_{01}$, $ \sigma_{45}T=T\sigma_{45}$ and
$ \tau T=T^{-1} \tau$.
\end{itemize}
\end{Proposition}

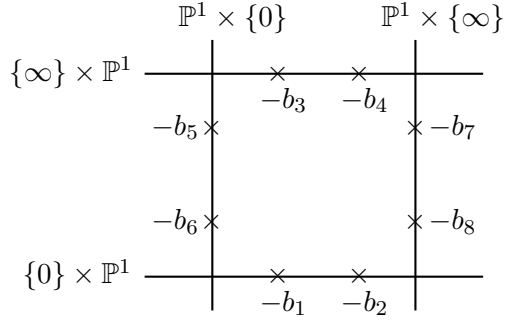
\begin{figure}[t]\centering
\vspace*{30mm}
\begin{picture}(50,40)
\setlength{\unitlength}{0.9mm}
\thicklines
\put(-10,5){\line(1,0){50}}
\put(-10,35){\line(1,0){50}}
\put(0,0){\line(0,1){40}}
\put(30,0){\line(0,1){40}}
\put(-5,42){$\mathbb{P}^1 \times \{0\}$}
\put(25,42){$\mathbb{P}^1 \times \{\infty\}$}
\put(-28,4){$\{0\} \times \mathbb{P}^1$}
\put(-30,34){$\{\infty\} \times \mathbb{P}^1 $}
\put(8,4){$\times$}
\put(20,4){$\times$}
\put(8,34){$\times$}
\put(20,34){$\times$}
\put(-1.8,12){$\times$}
\put(-1.8,26){$\times$}
\put(28.3,12){$\times$}
\put(28.3,26){$\times$}
\put(7,0){$-b_1$}
\put(19,0){$-b_2$}
\put(7,30){$-b_3$}
\put(19,30){$-b_4$}
\put(-9,26){$-b_5$}
\put(-9,12){$-b_6$}
\put(32,26){$-b_7$}
\put(32,12){$-b_8$}
\end{picture}
\caption{Space of initial data for Painlev\'e VI (8 points blow up of $\mathbb{P}^1 \times \mathbb{P}^1$).}
\label{bvar}
\end{figure}

\begin{Definition}\label{b-variables}
Introduce the following variables:
\begin{align*}
&b_1=\mathsf{t} a_4^{-2}={a_3^2 a_5\over a_4},\qquad
b_2=\mathsf{t} a_4^{2}=a_3^2a_4^3 a_5,\qquad
b_3=\mathsf{t} ^{-1}a_5^{2}={a_5\over a_3^2 a_4},\\
&
b_4=\mathsf{t} ^{-1}a_5^{-2}={1\over a_3^2 a_4 a_5^3},\qquad b_5=\mathsf{t} \mathsf{p}^{-1} a_1^{2}={a_1\over a_0 a_2^2},\qquad
b_6=\mathsf{t} \mathsf{p}^{-1} a_1^{-2}=
{1\over a_0a_1^3 a_2^2},\\
&
b_7=\mathsf{t}^{-1} \mathsf{p} a_0^{-2}=
 {a_1 a_2^2\over a_0},\qquad
b_8=\mathsf{t}^{-1} \mathsf{p} a_0^{2}=
a_0^3 a_1a_2^2.
\end{align*}
\end{Definition}

The parameters $b_i$ are not independent and they satisfy the constraints
\begin{equation}\label{b-consts}
b_1b_2b_3b_4= 1,\qquad
b_5b_6b_7b_8= 1,\qquad
b_1b_2b_7b_8= \mathsf{p}^2,\qquad
b_3b_4b_5b_6= \mathsf{p}^{-2}.
\end{equation}
Regarding $\mathsf{p}$ as a new parameter, we have six independent parameters.
The root variables are expressed as the ratios of $b_i$ (see Figure \ref{D5}):
\begin{align*}
a_0^4 = \frac{b_8}{b_7}, \qquad
a_1^4 = \frac{b_5}{b_6}, \qquad
a_2^4 = \frac{b_7}{b_5}, \qquad
a_3^4 = \frac{b_1}{b_3}, \qquad
a_4^4 = \frac{b_2}{b_1}, \qquad
a_5^4 = \frac{b_3}{b_4}.
\end{align*}

\begin{Corollary}
We have
\begin{align*}
&f \overline{f}=b_7b_8
{g+ b_5 \over
g+ b_7}
{g+ b_6 \over
g+ b_8},\qquad
g \underline{g}=b_3b_4
{f+b_1 \over
f+ b_3}
{f+ b_2\over
f+ b_4}.
\end{align*}
\end{Corollary}

\subsection[Tables of the action of W on K(f,g)]{Tables of the action of $\boldsymbol{\widetilde{W}}$ on $\boldsymbol{\mathbb{K}(f,g)}$}
\begin{Lemma}\label{lemA7}
Write $a=(a_0,a_1,a_2,a_3,a_4,a_5)$ for short. We have
\begin{gather*}
r_0\colon\ (a,f,g)
\mapsto
\bigl(a_0^{-1},a_1,a_0a_2,a_3,a_4,a_5,f,g\bigr),\\
r_1\colon\ (a,f,g)
\mapsto
\bigl(a_0,a_1^{-1},a_1a_2,a_3,a_4,a_5,f,g\bigr),\\
r_2\colon\ (a,f,g)
\mapsto
\left(a_0a_2,a_1a_2,a_2^{-1},a_2a_3,a_4,a_5,f{a_0a_1^{-1}g+a_2^2\over a_0a_1^{-1}a_2^2g+1},g\right),\\
r_3\colon\ (a,f,g)
\mapsto
\left(a_0,a_1,a_2a_3,a_3^{-1},a_3a_4,a_3a_5,f,{a_3^2 a_4a_5^{-1} f+1\over a_4a_5^{-1}f+a_3^2}g\right),\\
r_4\colon\ (a,f,g)
\mapsto
\bigl(a_0,a_1,a_2,a_3a_4,a_4^{-1},a_5,f,g\bigr),\\
r_5\colon\ (a,f,g)
\mapsto
\bigl(a_0,a_1,a_2,a_3a_5,a_4,a_5^{-1},f,g\bigr),\\
\sigma_{01}\colon\ (a,f,g)
\mapsto
\bigl(a_1^{-1},a_0^{-1},a_2^{-1},a_3^{-1},a_4^{-1},a_5^{-1},f^{-1},g\bigr),\\
\sigma_{45}\colon\ (a,f,g)
\mapsto
\bigl(a_0^{-1},a_1^{-1},a_2^{-1},a_3^{-1},a_5^{-1},a_4^{-1},f,g^{-1}\bigr),\\
\tau\colon\ (a,f,g)
\mapsto
\bigl(a_5^{-1},a_4^{-1},a_3^{-1},a_2^{-1},a_1^{-1},a_0^{-1},g,f\bigr),\\
T\colon\ a \mapsto \bigl(a_0,a_1,\mathsf{p} a_2,\mathsf
{p}^{-1} a_3,a_4,a_5\bigr).
\end{gather*}
\end{Lemma}

\begin{Lemma}\label{lemA8}
Write $b=(b_1,b_2,b_3,b_4,b_5,b_6,b_7,b_8)$ for short.
We have
\begin{gather*}
r_0\colon\ (b,f,g)
\mapsto
(b_1,b_2,b_3,b_4,b_5,b_6,b_8,b_7,f,g),\\
r_1\colon\ (b,f,g)
\mapsto
(b_1,b_2,b_3,b_4,b_6,b_5,b_7,b_8,f,g),\\
r_2\colon\ (b,f,g)
\mapsto
\bigg(b_1\sqrt{b_7\over b_5},b_2\sqrt{b_7\over b_5},b_3\sqrt{b_5\over b_7},b_4\sqrt{b_5\over b_7},
b_7,b_6,b_5,b_8,f\sqrt{b_5\over b_7}\frac{g+b_7}{g+b_5} ,g\bigg),\\
r_3\colon\ (b,f,g)
\mapsto
\bigg(b_3,b_2,b_1,b_4,
b_5\sqrt{b_3\over b_1},b_6\sqrt{b_3\over b_1},b_7\sqrt{b_1\over b_3},b_8\sqrt{b_1\over b_3},
f, \sqrt{b_1\over b_3}\frac{f+b_3}{f+b_1} g\bigg),\\
r_4\colon\ (b,f,g)
\mapsto
(b_2,b_1,b_3,b_4,b_5,b_6,b_7,b_8,f,g),\\
r_5\colon\ (b,f,g)
\mapsto
(b_1,b_2,b_4,b_3,b_5,b_6,b_7,b_8,f,g),\\
\sigma_{01}\colon\ (b,f,g)
\mapsto
\bigl(b_1^{-1},b_2^{-1},b_3^{-1},b_4^{-1},b_7,b_8,b_5,b_6,f^{-1},g\bigr),\\
\sigma_{45}\colon\ (b,f,g)
\mapsto
\bigl(b_3,b_4,b_1,b_2,b_5^{-1},b_6^{-1},b_7^{-1},b_8^{-1},f,g^{-1}\bigr),\\
\tau\colon\ (b,f,g)
\mapsto
(b_5,b_6,b_7,b_8,b_1,b_2,b_3,b_4,g,f),\\
T\colon\ b
\mapsto
\bigl(\mathsf{p}^{-2} b_1, \mathsf{p}^{-2}b_2, \mathsf{p}^{2}b_3, \mathsf{p}^{2}b_4,
\mathsf{p}^{-2}b_5,\mathsf{p}^{-2}b_6,\mathsf{p}^{2}b_7,\mathsf{p}^{2}b_8\bigr).
\end{gather*}
\end{Lemma}
The action of $T$ on $(f,g)$ is read from the Painlev\'e VI equation.

\section{Evidences for Conjecture \ref{secondconjecture}}
\label{App.B}

\subsection{Heine limit (the sector with vanishing instanton number)}

The power of the parameter $\Lambda$ counts the instanton number.
In the $\Lambda \to 0$ limit of the partition function \eqref{Nf4},
only the fixed points with instanton number zero survive\footnote{Recall that $x_1 = x$ and $x_2 = \Lambda/x$.}
and the partition function degenerates to the $q$-hypergeometric series.
Note that in this limit the non-stationary shift $T_{t, \Lambda}$ becomes trivial.

\begin{Lemma}\label{Heine}
The equation
\begin{align*}
S y(x)=y(x), \qquad
S:=\frac{1}{\varphi(x)} \Bor
\frac{1}{\varphi\bigl(-\frac{a}{q}x\bigr)\varphi\bigl(-\frac{b}{q}x\bigr)} \Bor \frac{1}{\varphi\bigl(\frac{ab}{q^2}x\bigr)} T_{\frac{c}{q^2},x}
\end{align*}
has a solution given by the Heine's $q$-hypergeometric series
\[
y(x)={}_2 \phi_{1}(x)= {}_2\phi_1\left[\left. {a,b \atop c}\right| q; x \right]
:=\sum_{n=0}^{\infty}\frac{(a;q)_n(b;q)_n}{(q,q)_n(c;q)_n}x^n.
\]
\end{Lemma}

\begin{proof}
Recall that the Heine's series ${}_2 \phi_{1}(x)$
satisfies the $q$-difference equation
\begin{align*}
D_{\rm Heine}& \ {}_2 \phi_{1}(x)=0, \\
D_{\rm Heine}&:=x(1-a p)(1-b p)p^{-1}-(1-p)(q-c p)p^{-1},
\qquad p=T_{q, x}.
\end{align*}
In a similar way to the proof of Proposition~\ref{prp:5dSW},
we see the adjoint actions on $D_{\rm Heine}$ as
\begin{center}\setlength{\unitlength}{0.4mm}
\begin{picture}(140,45)(45,-15)
\put(0,0){\line(0,1){20}}
\put(10,0){\line(0,1){20}}
\put(0,0){\line(1,0){10}}
\put(0,10){\line(1,0){10}}
\put(0,20){\line(1,0){10}}
\multiput(0,0)(10,0){2}{\circle*{4}}
\multiput(0,10)(10,0){2}{\circle*{4}}
\multiput(0,20)(10,0){2}{\circle*{4}}
\put(25,5){$\mapsto$}
\put(25,-5){{\tiny $\operatorname{Ad}\varphi(x)$}}
\put(60,0){\line(0,1){20}}
\put(70,0){\line(0,1){20}}
\put(60,0){\line(1,0){10}}
\put(60,10){\line(1,0){10}}
\put(60,20){\line(1,0){10}}
\multiput(60,0)(10,0){1}{\circle*{4}}
\multiput(60,10)(10,0){2}{\circle*{4}}
\multiput(60,20)(10,0){3}{\circle*{4}}
\put(85,5){$\mapsto$}
\put(82,-5){{\tiny $\operatorname{Ad}\Bor^{-1}$}}
\put(120,0){\line(0,1){20}}
\put(130,0){\line(0,1){20}}
\put(120,0){\line(1,0){10}}
\put(120,10){\line(1,0){10}}
\put(120,20){\line(1,0){10}}
\multiput(120,0)(10,0){3}{\circle*{4}}
\multiput(120,10)(10,0){2}{\circle*{4}}
\multiput(120,20)(10,0){1}{\circle*{4}}
\end{picture}
\setlength{\unitlength}{0.4mm}
\begin{picture}(120,45)(25,-15)
\put(5,5){$\mapsto$}
\put(-15,-5){{\tiny $\operatorname{Ad}\varphi(-\frac{a}{q}x)\varphi(-\frac{b}{q}x)$}}
\put(60,0){\line(0,1){20}}
\put(70,0){\line(0,1){20}}
\put(60,0){\line(1,0){10}}
\put(60,10){\line(1,0){10}}
\put(60,20){\line(1,0){10}}
\multiput(60,0)(10,0){1}{\circle*{4}}
\multiput(60,10)(10,0){2}{\circle*{4}}
\multiput(60,20)(10,0){3}{\circle*{4}}
\put(85,5){$\mapsto$}
\put(82,-5){{\tiny $\operatorname{Ad}\Bor^{-1}$}}
\put(115,0){\line(0,1){20}}
\put(125,0){\line(0,1){20}}
\put(115,0){\line(1,0){10}}
\put(115,10){\line(1,0){10}}
\put(115,20){\line(1,0){10}}
\multiput(115,0)(10,0){3}{\circle*{4}}
\multiput(115,10)(10,0){2}{\circle*{4}}
\multiput(115,20)(10,0){1}{\circle*{4}}
\put(150,5){$\mapsto$}
\put(145,-5){{\tiny $\operatorname{Ad}\varphi(\frac{ab}{q^2}x)$}}
\put(185,0){\line(0,1){20}}
\put(195,0){\line(0,1){20}}
\put(185,0){\line(1,0){10}}
\put(185,10){\line(1,0){10}}
\put(185,20){\line(1,0){10}}
\multiput(185,0)(10,0){2}{\circle*{4}}
\multiput(185,10)(10,0){2}{\circle*{4}}
\multiput(185,20)(10,0){2}{\circle*{4}}
\put(200,0){,}
\end{picture}
\end{center}
and we have
\[
S^{-1} D_{\rm Heine} \ S=D_{\rm Heine}.
\]
Hence, we obtain
\[
D_{\rm Heine}  S \, {}_2\phi_{1}(x)
=S D_{\rm Heine} \, {}_2 \phi_{1}(x)=0.
\]
Since the power series solution of the equation $ D_{\rm Heine} y(x)=0$ of the form $y(x)=1+O(x)$ is unique,
we have the conclusion $S \, {}_2\phi_{1}(x)= {}_2\phi_{1}(x)$.
\end{proof}

\begin{Remark}
Lemma \ref{Heine} shows the factorization of the coefficients $c_{m,0}$ of the ``boundary'' terms of
the solution \eqref{double-series}.
The other cases $c_{0,n}$ are similar.
\end{Remark}

\subsection{Macdonald limit}

\subsubsection[Macdonald function of types A\_1 and A\_2 \protect{[43]}]{Macdonald function of types $\boldsymbol{A_1}$ and $\boldsymbol{A_2}$ \cite{NShi}}

The asymptotically free Macdonald function $f^{\mathfrak{gl}_2}(x|s|q,t)$ of type $A_1$
is written as
\begin{align*}
f^{\mathfrak{gl}_2}(x|s|q,t)=\sum_{\theta\geq 0}
{(t;q)_\theta (ts_2/s_1;q)_\theta \over
(q;q)_\theta (qs_2/s_1;q)_\theta} (q x_2/tx_1)^\theta.
\end{align*}
We have the eigenvalue equation
\begin{align*}
&\mathcal{D}^{\mathfrak{gl}_2}(s|q,t)= s_1 {1-x_2/tx_1\over 1-x_2/x_1}T_{q,x_1}+s_2 {1-t x_2/x_1\over 1-x_2/x_1}T_{q,x_2},\\
&\mathcal{D}^{\mathfrak{gl}_2}(s|q,t) f^{\mathfrak{gl}_2}(x|s|q,t)=
(s_1+s_2)f^{\mathfrak{gl}_2}(x|s|q,t).
\end{align*}

The asymptotically free Macdonald function $f^{\mathfrak{gl}_3}(x|s|q,t)$ of type $A_2$
is written as
\begin{align}
f^{\mathfrak{gl}_3}(x|s|q,t)=&\sum_{\theta_{12}, \theta_{13}, \theta_{23} \geq 0}
{(t;q)_{\theta_{12}} \bigl(q^{-\theta_{23}+\theta_{13}} ts_2/s_1;q\bigr)_{\theta_{12}} \over
(q;q)_{\theta_{12}} \bigl(q^{-\theta_{23}+\theta_{13}} qs_2/s_1;q\bigr)_{\theta_{12}}} (q x_2/tx_1)^{\theta_{12}}\nonumber \\
& \times
{(t;q)_{\theta_{13}} ( ts_3/s_1;q)_{\theta_{13}} \over
(q;q)_{\theta_{13}} ( qs_3/s_1;q)_{\theta_{13}}}
{(ts_2/s_1;q)_{\theta_{13}} \bigl(q^{\theta_{23}-\theta_{13}} ts_1/s_2;q\bigr)_{\theta_{13}} \over
(qs_2/s_1;q)_{\theta_{13}} \bigl(q^{\theta_{23}-\theta_{13}} qs_1/s_2;q\bigr)_{\theta_{13}}}\nonumber \\
&\times \bigl(q^2 x_3/t^2x_1\bigr)^{\theta_{13}}
{(t;q)_{\theta_{23}} ( ts_3/s_2;q)_{\theta_{23}} \over
(q;q)_{\theta_{23}} ( qs_3/s_2;q)_{\theta_{23}}} (q x_3/tx_2)^{\theta_{23}}.\label{f3}
\end{align}
The eigenvalue equation is
\begin{align*}
&\mathcal{D}^{\mathfrak{gl}_3}(s|q,t)=
s_1 {1-x_2/tx_1\over 1-x_2/x_1} {1-x_3/tx_1\over 1-x_3/x_1}T_{q,x_1}\\
&\phantom{\mathcal{D}^{\mathfrak{gl}_3}(s|q,t)=}{} +s_2 {1-t x_2/x_1\over 1-x_2/x_1} {1- x_3/tx_2\over 1-x_3/x_2}T_{q,x_2}+
s_3 {1-t x_3/x_1\over 1-x_3/x_1} {1- t x_3/x_2\over 1-x_3/x_2}T_{q,x_3},\\
&\mathcal{D}^{\mathfrak{gl}_3}(s|q,t) f^{\mathfrak{gl}_3}(x|s|q,t)=
(s_1+s_2+s_3)f^{\mathfrak{gl}_3}(x|s|q,t).
\end{align*}

\subsubsection{Macdonald limit by tuning mass parameters}

Consider the particular choice of the mass parameters
\begin{align}
T_1=v^{-1},\qquad T_2=v,\qquad T_3=v^{-1},\qquad T_4=v t^{-1}. \label{Mac(1)}
\end{align}

\begin{Proposition}\label{gl3-dege}
Let \eqref{Mac(1)} be satisfied.
Then the Nekrasov partition function $\Psi( t\Lambda,x)$ degenerates to
the asymptotically free Macdonald function of type $\mathfrak{gl}_3$
with $(x_1,x_2,x_3)=(1/\Lambda,1/x,1)$ and $(s_1,s_2,s_3)=(1/Q,1,0)$,
\begin{align}
\Psi( t\Lambda,x)= f^{\mathfrak{gl}_3}(1/\Lambda,1/x,1|1/Q,1,0|q,t). \label{Z-f3}
\end{align}

Therefore, we have the eigenvalue equation
\begin{align*}
&\mathcal{D}^{\mathfrak{gl}_3}(1/Q,1,0 |q,q/t) \Psi( t\Lambda,x)
=(1/Q+1)\Psi( t\Lambda,x),
\end{align*}
where
\begin{align*}
&\mathcal{D}^{\mathfrak{gl}_3}(1/Q,1,0 |q,q/t)=
Q^{-1} {1-t\Lambda/q x\over 1-\Lambda/x} {1-t\Lambda/q\over 1-\Lambda}T_{q,\Lambda}^{-1}
+ {1-q\Lambda/t x\over 1-\Lambda/x} {1- tx/q\over 1-x}T_{q,x}^{-1}.
\end{align*}
\end{Proposition}

We need the following lemma which is easily confirmed by the explicit formula of
the Nekrasov factor.
\begin{Lemma}\quad
\begin{itemize}\itemsep=0pt
 \item[$(1)$]$\mathsf{N}_{\lambda,\mu}(1|q,\kappa)\neq 0$
if and only if $\lambda_i\leq \mu_i$ $(i\geq 1)$.
 \item[$(2)$]
$\mathsf{N}_{\lambda,\mu}\bigl(\kappa^{-1}|q,\kappa\bigr)\neq 0$
if and only if $\lambda_{i+1}\leq \mu_i$ $(i\geq 1)$.
\end{itemize}
\end{Lemma}

Then the proof of Proposition \ref{gl3-dege} is given as follows:
\begin{proof}
Under the condition \eqref{Mac(1)}, we have (see \eqref{Sh-parameters})
\begin{alignat*}{4}
&v \mathfrak{f}^+_2/\mathfrak{n}_1=1, \qquad &&v \mathfrak{f}^+_1/\mathfrak{n}_2=1,\qquad &&w \mathfrak{n}_1/\mathfrak{m}_1=t,&\\
& w \mathfrak{n}_2/\mathfrak{m}_2=1, \qquad&& v \mathfrak{m}_1/\mathfrak{f}^-_1=1,\qquad&& v \mathfrak{m}_2/\mathfrak{f}^-_2=t.&
\end{alignat*}
Hence, in view of the lemma above, we have the
following parametrization of the quadruples of partitions
which give rise to the nonvanishing contributions in the partition function~\eqref{Z}:
\begin{align*}
&\nu_1=(\theta_{23}),\qquad \nu_2=(\theta_{13}),\qquad \mu_1=\varnothing ,\qquad \mu_2=(\theta_{12}+\theta_{13}),
\end{align*}
where $\theta_{12},\theta_{13},\theta_{23}\geq 0$.
Then by making straightforward calculation and comparison with \eqref{f3}
under the identification $(x_1,x_2,x_3)=(1/\Lambda,1/x,1)$,
and $(s_1,s_2,s_3)=(1/Q,1,0)$, we obtain~\eqref{Z-f3}.
\end{proof}

\begin{Proposition}
Let \eqref{Mac(1)} be satisfied.
Then we have
\begin{align}
&\mathcal{F}^{(1)}=
\sum_{m\geq 0}{(q/t;q)_m (qQ/t;q)_m\over (q;q)_m (q Q;q)_m} t^m (\Lambda/x)^m=
 f^{\mathfrak{gl}_2}(1/\Lambda,1/x|1/Q,1|q,t),\label{F1-Mac}\\
&\mathcal{F}^{(2)}=
\sum_{m,n\geq 0}{(t;q)_m(q^n q/t;q)_m\over (q;q)_m(q^n qQ;q)_m} \bigl(-\bigl(q t^{1/2}Q\bigr)^{1/2}\bigr)^m (\Lambda/x)^m
{(q/t;q)_n(qQ/t;q)_n\over (q;q)_n (q Q;q)_n} t^{2n}\Lambda^n.\label{F2-Mac}
\end{align}
Hence, with the identification $(x_1,x_2)=(1/\Lambda,1/x)$, we have the eigenvalue equation
\begin{align*}
&\mathcal{D}^{\mathfrak{gl}_2}(1/Q,1|q,q/t)\mathcal{F}^{(1)}
=(1/Q+1)\mathcal{F}^{(1)},
\end{align*}
where
\begin{align*}
&\mathcal{D}^{\mathfrak{gl}_2}(1/Q,1 |q,q/t)=
Q^{-1} {1-t\Lambda/q x\over 1-\Lambda/x} T_{q,\Lambda}^{-1}
+ {1-q\Lambda/t x\over 1-\Lambda/x} T_{q,x}^{-1}.
\end{align*}
\end{Proposition}

\begin{proof}
Under the condition \eqref{Mac(1)}, we have
\begin{align*}
&\mathcal{F}^{(1)}=f\left( \left.\left.\begin{array}{ll}
t Q^{1/2} ,&qt^{-1/2} Q^{-1/2} \\
Q^{-1/2},&t^{1/2}Q^{1/2}\\
Q^{1/2} ,&q^{-1}t^{3/2} Q^{-1/2}\end{array} \right|
-t^{1/2} x,- \Lambda/x \right|q,t^{-1/2}\right),
\end{align*}
and
\begin{align*}
&\mathcal{F}^{(2)}=f\left( \left.\left.\begin{array}{ll}
t Q^{1/2} ,&qt^{-1/2} Q^{-1/2} \\
Q^{-1/2},&t^{1/2}Q^{1/2}\\
q^{-1} tQ^{-1/2},&t^{-1/2} Q^{1/2} \end{array} \right|
t^{3/4} x,t^{1/4} \Lambda/x \right|q,t^{-1/2}\right).
\end{align*}
One finds that $\mathcal{F}^{(1)}$ has the factor
\begin{align*}
\mathsf{N}^{(1|2)}_{\lambda^{(2)},\varnothing}\bigl(t^{1/2}|q,t^{-1/2}\bigr),
\end{align*}
which is nonvanishing, if and only if $\ell \bigl(\lambda^{(2)}\bigr)\leq 1$.
The $\mathcal{F}^{(1)}$ also has the factor
\begin{align*}
\mathsf{N}^{(1|2)}_{\varnothing,\lambda^{(1)}}\bigl(qt^{-1/2}|q,t^{-1/2}\bigr),
\end{align*}
which is nonvanishing, if and only if $\lambda^{(1)}=\varnothing$.
Hence, the nonvanishing contributions for $\mathcal{F}^{(1)}$ arise if and only if we have
$\bigl(\lambda^{(1)},\lambda^{(2)})=(\varnothing,(m)\bigr)$
$(m\geq 0)$. Making explicit calculation, we have~\eqref{F1-Mac}.

The $\mathcal{F}^{(2)}$ has the factor
\begin{align*}
\mathsf{N}^{(0|2)}_{\lambda^{(2)},\varnothing}\bigl(t|q,t^{-1/2}\bigr),
\end{align*}
which is nonvanishing, if and only if $\ell \bigl(\lambda^{(2)}\bigr)\leq 2$.
The $\mathcal{F}^{(2)}$ also has the factor
\begin{align*}
\mathsf{N}^{(1|2)}_{\varnothing,\lambda^{(1)}}\bigl(qt^{-1/2}|q,t^{-1/2}\bigr),
\end{align*}
which is nonvanishing, if and only if $\lambda^{(1)}=\varnothing$.
Hence, the nonvanishing contributions for $\mathcal{F}^{(2)}$ arise if and only if we have
$\bigl(\lambda^{(1)},\lambda^{(2)}\bigr)=(\varnothing,(n+m,n))$
$(m,n\geq 0)$. Making explicit calculation, we have \eqref{F2-Mac}.
\end{proof}

We have seen how both partition functions $\Psi(t\Lambda,x)$ and $\mathcal{F}^{(1)}$ are simplified to
the asymptotically free Macdonald functions in the limit \eqref{Mac(1)}.
If Conjectures \ref{conjecture} and \ref{secondconjecture} are true, they are related by the following gauge transformation.
\begin{con}\label{Gauge}
\begin{align*}
\mathcal{F}^{(1)}
&={1\over \Phi\bigl(t^2 T_2 T_3 \Lambda\bigr)\Phi(qt T_1 T_4 \Lambda)}
\varphi\bigl(q^{1/2}t^{1/2} T_1 x\bigr)\varphi\bigl(q^{1/2}t^{1/2} T_4 \Lambda/x\bigr)\psi(\Lambda,x)\\
&=
{\Phi(qt T_2 T_3 \Lambda)\Phi\bigl(t^2 T_1 T_4 \Lambda\bigr)\over \Phi\bigl(t^2 T_2 T_3 \Lambda\bigr)\Phi(qt T_1 T_4 \Lambda)}
\varphi\bigl(q^{1/2}t^{1/2} T_1 x\bigr)\varphi\bigl(q^{1/2}t^{1/2} T_4 \Lambda/x\bigr)
\mathcal{A}_3(t \Lambda, t q Q x)\Psi(t\Lambda,x).
\end{align*}
\end{con}

We can prove this is indeed the case.

\begin{Lemma}
We have
\begin{align}
 f^{\mathfrak{gl}_2}(1/\Lambda,1/x|1/Q,1|q,t)={ (t\Lambda;q)_\infty \over (q\Lambda;q)_\infty }
{ (t x;q)_\infty \over (q x;q)_\infty } f^{\mathfrak{gl}_3}(1/\Lambda,1/x,1|1/Q,1,0|q,t) .\label{Mac-Gauge}
\end{align}
\end{Lemma}

\begin{proof}
We find that the Macdonald operators $\mathcal{D}^{\mathfrak{gl}_3}(1/Q,1,0 |q,q/t)$
and $\mathcal{D}^{\mathfrak{gl}_2}(1/Q,1 |q,q/t)$ are gauge equivalent
\begin{align*}
{ (t\Lambda;q)_\infty \over (q\Lambda;q)_\infty }
{ (t x;q)_\infty \over (q x;q)_\infty }\mathcal{D}^{\mathfrak{gl}_3}(1/Q,1,0 |q,q/t)
{ (q\Lambda;q)_\infty \over (t\Lambda;q)_\infty }
{ (q x;q)_\infty \over (t x;q)_\infty }=\mathcal{D}^{\mathfrak{gl}_2}(1/Q,1 |q,q/t).
\end{align*}
Then the equality \eqref{Mac-Gauge} follows from the uniqueness of the normalized asymptotic eigenfunctions of
$\mathcal{D}^{\mathfrak{gl}_3}(1/Q,1,0 |q,q/t)$
and $\mathcal{D}^{\mathfrak{gl}_2}(1/Q,1 |q,q/t)$.
\end{proof}

\begin{Proposition}
Conjecture {\rm\ref{Gauge}} holds in the Macdonald limit \eqref{Mac(1)}.
\end{Proposition}
\begin{proof}
Under the condition \eqref{Mac(1)}, we have
\begin{align*}
&{\Phi(qt T_2 T_3 \Lambda)\Phi(t^2 T_1 T_4 \Lambda)\over \Phi(t^2 T_2 T_3 \Lambda)\Phi(qt T_1 T_4 \Lambda)}
\varphi\bigl(q^{1/2}t^{1/2} T_1 x\bigr)\varphi\bigl(q^{1/2}t^{1/2} T_4 \Lambda/x\bigr)
\mathcal{A}_3(t \Lambda, t q Q x)\\&\qquad =
{ (t\Lambda;q)_\infty \over (q\Lambda;q)_\infty }
{ (t x;q)_\infty \over (q x;q)_\infty }.\tag*{\qed}
\end{align*}\renewcommand{\qed}{}
\end{proof}

Now we are ready to prove Conjecture \ref{secondconjecture} in the Macdonald limit.
We need the following formulas for the proof.
\begin{Lemma}[\cite{Shakirov:2021krl}]
For $n\in \mathbb{Z}$,
we have
\begin{align}
&\mathcal{B} {1\over \varphi(\alpha x)\varphi(\beta \Lambda/x)}x^n
={ \varphi\bigl(-q^{1+n}\alpha x\bigr)\varphi\bigl(-q^{-n}\beta \Lambda/x\bigr)\over \varphi(\alpha \beta \Lambda)}
q^{n(n+1)/2}x^n,\label{lem-1}\\
&\mathcal{B}^{-1} {\varphi(\alpha x)\varphi(\beta \Lambda/x)}x^n
={ \varphi\bigl(q^{-1}\alpha \beta \Lambda\bigr)\over
\varphi\bigl(-q^{-1-n}\alpha x\bigr)\varphi\bigl(-q^{n}\beta \Lambda/x\bigr)}
q^{-n(n+1)/2}x^n.\label{lem-2}
\end{align}
\end{Lemma}

\begin{Proposition}[$q$-Chu--Vandermonde sums \cite{hypergeometric}] For $n\in \mathbb{Z}_{\geq 0}$, we have
\begin{align}
&{}_2\phi_1\left[\left. {a,q^{-n} \atop c}\right| q;q\right]={(c/a;q)_n\over (c;q)_n} a^n,\label{Chu-Vand-1}\\
&{}_2\phi_1\left[\left. {a,q^{-n} \atop c}\right| q;{cq^n\over a}\right]={(c/a;q)_n\over (c;q)_n}.\label{Chu-Vand-2}
\end{align}
\end{Proposition}

\begin{Proposition}
Conjecture {\rm\ref{secondconjecture}} holds in the Macdonald limit \eqref{Mac(1)}. Then
\begin{align}
&T_{(q t^{1/2}Q)^{1/2},x}\mathcal{F}^{(2)}=
\varphi(q\Lambda)
{1\over \varphi(- q t x)\varphi(- \Lambda/ x)} \mathcal{B}
{1\over \varphi( t x)\varphi(q \Lambda/t x)} \mathcal{F}^{(1)},\label{Mac-eq-1}\\
&\mathcal{F}^{(1)}=
\varphi(t\Lambda){1\over \varphi( q x)\varphi(t \Lambda/ x)}  \mathcal{B}
{1\over \varphi(- x)\varphi(- q\Lambda/ x)}
\bigl(T_{t^{1/2},x}^{-1}T_{t,\Lambda}^{-1}\bigr)T_{(q t^{1/2}Q)^{-1/2},x}\mathcal{F}^{(2)}.\label{Mac-eq-2}
\end{align}
\end{Proposition}

\begin{proof}
First, we show \eqref{Mac-eq-1}.
By using \eqref{lem-1} and the $q$-binomial formula, we have
\begin{gather*}
\varphi(q\Lambda)
{1\over \varphi(- q t x)\varphi(- \Lambda/ x)} \mathcal{B}
{1\over \varphi( t x)\varphi(q \Lambda/t x)} \mathcal{F}^{(1)}\\
\qquad=\varphi(q\Lambda)
{1\over \varphi(- q t x)\varphi(- \Lambda/ x)} \mathcal{B}
{1\over \varphi( t x)\varphi(q \Lambda/t x)}\sum_{n\geq 0}
{(q/t;q)_n (qQ/t;q)_n\over (q;q)_n (q Q;q)_n} t^n (\Lambda/x)^n\\
\qquad=\sum_{n\geq 0}
{\varphi\bigl(- q^{1-n} t x\bigr)\varphi\bigl(- q^{1+n}\Lambda/t x\bigr)\over \varphi(- q t x)\varphi(- \Lambda/ x)}
{(q/t;q)_n (qQ/t;q)_n\over (q;q)_n (q Q;q)_n} t^n q^{n(n-1)/2}(\Lambda/x)^n\\
\qquad=\sum_{k,l,n\geq 0}
{(q^{-n};q)_k\over (q;q)_k} (-q t x)^k {\bigl(q^{n+1}/t;q\bigr)_l\over (q;q)_l}(-\Lambda/x)^l
{(q/t;q)_n (qQ/t;q)_n\over (q;q)_n (q Q;q)_n} t^n q^{n(n-1)/2}(\Lambda/x)^n.
\end{gather*}
Note that we have the truncation of the summation as $\sum_{k,l,n\geq 0} =\sum_{k,l\geq 0}\sum_{n=k}^\infty $
because of the factor $(q^{-n};q)_k$. Shifting the running index $n$ by $k$ as $n\rightarrow n+k$
(hence changing the summation range as $\sum_{k,l\geq 0}\sum_{n=k}^\infty \rightarrow
\sum_{k,l\geq 0}\sum_{n=0}^\infty $ accordingly), we have
\begin{align*}
={}&\sum_{k,l,n\geq 0}
{\bigl(q^{-k-n};q\bigr)_k\over (q;q)_k} (-q t )^k {\bigl(q^{k+n+1}/t;q\bigr)_l\over (q;q)_l}(-1)^l\\
&\times
{(q/t;q)_{k+n} (qQ/t;q)_{k+n}\over (q;q)_{k+n} (q Q;q)_{k+n}} t^{k+n} q^{(k+n)(k+n-1)/2} (\Lambda/x)^{l+n} \Lambda^k.
\end{align*}
Then we change the running index $l$ to $m$ (where $m=l+n$)
as $\sum_{l,n\geq 0} =\sum_{m\geq 0}\sum_{n=0}^m$.
Simplifying the factors, we have
\begin{align*}
={}&\sum_{k\geq 0}\sum_{m\geq 0}
{\bigl(q^{k+1}/t;q\bigr)_m\over (q;q)_m} (-\Lambda/x)^m
{(q/t;q)_k\over (q;q)_k} {(qQ/t;q)_k \over (qQ;q)_k} \bigl(t^{2}\Lambda\bigr)^k\\
&
\times\sum_{n=0}^m
{(q^{-m};q)_n\over (q;q)_n}{\bigl(q^{k+1}Q/t;q\bigr)_n\over \bigl(q^{k+1}Q;q\bigr)_n}q^{mn}t^{n}\\
={}&\sum_{k\geq 0}\sum_{m\geq 0}
{\bigl(q^{k+1}/t;q\bigr)_m\over (q;q)_m}{(t;q)_m\over \bigl(q^{k+1}Q;q\bigr)_m} (-\Lambda/x)^m
{(q/t;q)_k\over (q;q)_k} {(qQ/t;q)_k \over (qQ;q)_k} \bigl(t^{2}\Lambda\bigr)^k\\
={}&T_{(q t^{1/2}Q)^{1/2},x}\mathcal{F}^{(2)}.
\end{align*}
Here, we have used the
$q$-Chu--Vandermonde summation formula \eqref{Chu-Vand-2} in the last step.

The second equation \eqref{Mac-eq-2} can be shown in exactly the same manner as above.
Note that~\eqref{Mac-eq-2} is equivalent to
\begin{align*}
\bigl(T_{t^{1/2},x}^{-1}T_{t,\Lambda}^{-1}\bigr)T_{(q t^{1/2}Q)^{-1/2},x}\mathcal{F}^{(2)}=
{1\over \varphi(t\Lambda)}
{\varphi(- x)\varphi(- q\Lambda/ x)}
\mathcal{B}^{-1}
{\varphi( q x)\varphi(t \Lambda/ x)}
\mathcal{F}^{(1)}.
\end{align*}
By using \eqref{lem-2} and the $q$-binomial formula, we have
\begin{gather*}
 {1\over \varphi(t\Lambda)}
{\varphi(- x)\varphi(- q\Lambda/ x)}
\mathcal{B}^{-1}
{\varphi( q x)\varphi(t \Lambda/ x)}
\mathcal{F}^{(1)}\\
\qquad= {1\over \varphi(t\Lambda)}
{\varphi(- x)\varphi(- q\Lambda/ x)}
\mathcal{B}^{-1}
{\varphi( q x)\varphi(t \Lambda/ x)}
\sum_{n\geq 0}
{(q/t;q)_n (qQ/t;q)_n\over (q;q)_n (q Q;q)_n} t^n (\Lambda/x)^n\\
\qquad=\sum_{n\geq 0}
{\varphi(- x)\varphi(- q\Lambda/ x)\over \varphi(- q^n x)\varphi(-q^{-n}t \Lambda/ x)}
{(q/t;q)_n (qQ/t;q)_n\over (q;q)_n (q Q;q)_n} t^n q^{-n(n-1)/2}(\Lambda/x)^n\\
\qquad=\sum_{k,l,n\geq 0}
{(q^{-n};q)_k\over (q;q)_k} (-q^n x)^k {(q^{n+1}/t;q)_l\over (q;q)_l}(-q^{-n}t \Lambda/x)^l\\
\phantom{\qquad=}{}\times{(q/t;q)_n (qQ/t;q)_n\over (q;q)_n (q Q;q)_n} t^n q^{-n(n-1)/2}(\Lambda/x)^n.
\end{gather*}
Note that we have $\sum_{k,l,n\geq 0} =\sum_{k,l\geq 0}\sum_{n=k}^\infty $
because of the factor $(q^{-n};q)_k$. Shifting the running index $n$ by $k$ as $n\rightarrow n+k$
(and changing the summation range as $\sum_{k,l\geq 0}\sum_{n=k}^\infty \rightarrow
\sum_{k,l\geq 0}\sum_{n=0}^\infty $ accordingly), we have
\begin{align*}
={}&\sum_{k,l,n\geq 0}
{\bigl(q^{-k-n};q\bigr)_k\over (q;q)_k} \bigl(-q^{n+k} \bigr)^k {\bigl(q^{k+n+1}/t;q\bigr)_l\over (q;q)_l}\bigl(-q^{-n-k}t\bigr)^l\\
&\times
{(q/t;q)_{k+n} (qQ/t;q)_{k+n}\over (q;q)_{k+n} (q Q;q)_{k+n}} t^{k+n} q^{-(k+n)(k+n-1)/2} (\Lambda/x)^{l+n} \Lambda^k.
\end{align*}
Changing the running index $l$ to $m$ (where $m=l+n$)
as $\sum_{l,n\geq 0} =\sum_{m\geq 0}\sum_{n=0}^m$, and
simplifying the factors, we have
\begin{align*}
={}&\sum_{k\geq 0}\sum_{m\geq 0}
{\bigl(q^{k+1}/t;q\bigr)_m\over (q;q)_m} (-\Lambda/x)^m
{(q/t;q)_k\over (q;q)_k} {(qQ/t;q)_k \over (qQ;q)_k} \Lambda^k q^{-km} t^{k+m}\\
&\times\sum_{n=0}^m
{(q^{-m};q)_n\over (q;q)_n}{\bigl(q^{k+1}Q/t;q\bigr)_n\over \bigl(q^{k+1}Q;q\bigr)_n}q^{n}\\
={}&\sum_{k\geq 0}\sum_{m\geq 0}
{\bigl(q^{k+1}/t;q\bigr)_m\over (q;q)_m}{(t;q)_m\over \bigl(q^{k+1}Q;q\bigr)_m} (-qQ\Lambda/x)^m
{(q/t;q)_k\over (q;q)_k} {(qQ/t;q)_k \over (qQ;q)_k} (t\Lambda)^k\\
={}&\bigl(T_{t^{1/2},x}^{-1}T_{t,\Lambda}^{-1}\bigr)T_{(q t^{1/2}Q)^{-1/2},x}\mathcal{F}^{(2)}.
\end{align*}
Here, we have used the
$q$-Chu--Vandermonde summation formula \eqref{Chu-Vand-1} in the last step.
\end{proof}

\subsubsection[Macdonald limit in \protect{[46]}]{Macdonald limit in \cite{Shakirov:2021krl}}

For the sake of readers' convenience,
we recollect the facts concerning the choice of the parameters investigated in \cite{Shakirov:2021krl}
\begin{align}
T_1=v t^{-1},\qquad T_2=v^{-1},\qquad T_3=v^{-1},\qquad T_4=v t^{-1}. \label{Mac(2)}
\end{align}
Note that in the limit \eqref{Mac(2)} we have (See Conjecture \ref{conjecture});
\begin{align*}
&\mathcal{A}_1(\Lambda,x)={1\over \varphi\bigl(qt^{-1} x\bigr)\varphi(t\Lambda/x)},\\
&\mathcal{A}_2(\Lambda,x)={\varphi(t \Lambda)\varphi\bigl(q t^{-2}\Lambda\bigr)
\over \varphi\bigl(- t^{-1}x\bigr)\varphi\bigl(-Q^{-1}x\bigr) \varphi(-qQ \Lambda/x) \varphi(-q \Lambda/x) },\\
&\mathcal{A}_3(\Lambda,x)={1\over \varphi\bigl(q^{-1}Q^{-1} x\bigr)\varphi\bigl( q^2 t^{-1}Q\Lambda/x\bigr)}.
\end{align*}

Set
\begin{align*}
&U(\Lambda,x)=
\sum_{k,l\geq 0}
{(q/t;q)_k\bigl(q^l q Q/t;q\bigr)_k \over (q;q)_k\bigl(q^l q Q;q\bigr)_k} (\Lambda/x)^k
{(t;q)_l( t Q;q)_l \over (q;q)_l( q Q;q)_l} \bigl(q\Lambda/t^2\bigr)^l,\\
&V(\Lambda,x)=
\sum_{k,l\geq 0}
{(q/t;q)_k(t;q)_k \over (q;q)_k(q Q;q)_k} (-qQ\Lambda/x)^k
{\bigl(q^k t^2;q\bigr)_l( t Q;q)_l \over (q;q)_l\bigl(q^k q Q;q\bigr)_l} \bigl(q\Lambda/t^2\bigr)^l.
\end{align*}

\begin{Proposition}[\cite{Shakirov:2021krl}]
When \eqref{Mac(2)} is satisfied, we have
\begin{align*}
\Psi(\Lambda,x)=U(\Lambda,x),
\end{align*}
and
\begin{align*}
&V(\Lambda,x)={\varphi(t\Lambda)\over \varphi\bigl(-Q^{-1}x\bigr)\varphi(-q Q\Lambda/x)}\mathcal{B}
{1\over \varphi\bigl(q^{-1}Q^{-1}x\bigr)\varphi\bigl(q^2t^{-1} Q\Lambda/x\bigr)}
U\left(\Lambda,{x\over t q Q}\right),\\
&U(t\Lambda,x)={\varphi\bigl(qt^{-2}\Lambda\bigr)\over \varphi\bigl(q t^{-1}x\bigr)\varphi(t \Lambda/x)}\mathcal{B}
{1\over \varphi\bigl(-t^{-1}x\bigr)\varphi(-q\Lambda/x)}
V\left(\Lambda,x\right).
\end{align*}
\end{Proposition}

\begin{Proposition}
In the limit \eqref{Mac(2)},
we have
\begin{gather*}
{\varphi(t\Lambda/x)\over \varphi(q \Lambda/t x)}
\Psi(t \Lambda,x)=
\sum_{k,l\geq 0}
{(t;q)_k\bigl(q^l t Q;q\bigr)_k \over (q;q)_k\bigl(q^l q Q;q\bigr)_k} (q\Lambda/tx)^k
{(t;q)_l( t Q;q)_l \over (q;q)_l(q Q;q)_l} (q\Lambda/t)^l\\
\phantom{{\varphi(t\Lambda/x)\over \varphi(q \Lambda/t x)}
\Psi(t \Lambda,x)}{}=f^{\mathfrak{gl}_3}\bigl(1/\Lambda,1/x,1|tQ^{-1},t,1|q,t\bigr),
\mathcal{D}^{\mathfrak{gl}_3}\bigl(tQ^{-1},t,1 |q,t\bigr),\\
{\varphi(t\Lambda/x)\over \varphi(q \Lambda/t x)}
\Psi(t \Lambda,x)=
\bigl(tQ^{-1}+t+1\bigr){\varphi(t\Lambda/x)\over \varphi(q \Lambda/t x)}
\Psi(t \Lambda,x),
\end{gather*}
where
\begin{align*}
\mathcal{D}^{\mathfrak{gl}_3}\bigl(tQ^{-1},t,1 |q,t\bigr)={}&
tQ^{-1} {1-\Lambda/t x\over 1-\Lambda/x} {1-\Lambda/t\over 1-\Lambda}T_{q,\Lambda}^{-1}\\
&
+ t {1-t\Lambda/ x\over 1-\Lambda/x} {1- x/t\over 1-x}T_{q,x}^{-1}+
{1-t\Lambda/ x\over 1-\Lambda/x} {1- tx\over 1-x}T_{q,\Lambda}T_{q,x}.
\end{align*}

\end{Proposition}

\begin{Remark}
In the limit \eqref{Mac(2)}, the set of partitions giving rise to the nonvanishing
contributions for $\mathcal{F}^{(1)}$ remains rather big. This seems to be consistent with the
appearance of double infinite products in the gauge transformation
\begin{align*}
&
{\Phi\bigl(q^{-1} t^{3} \Lambda\bigr)\Phi\bigl(q^2 t^{-2} \Lambda\bigr)\over
\Phi\bigl(q t^{-1} \Lambda\bigr)\Phi\bigl(t^2 \Lambda\bigr)
}{\varphi(t x) \over \varphi\bigl(q t^{-1} x\bigr)}
\mathcal{F}^{(1)}=
\Psi(t\Lambda,x).
\end{align*}
In four-dimensional theory, the set of partitions giving rise to the nonvanishing contributions is the same.
The four-dimensional limit of the ratio of double infinite products is
$(1- \Lambda)^{-\frac{2(\epsilon_1 + \epsilon_2)^2}{\epsilon_1\epsilon_2}}$.
On the other hand the corresponding set of partitions for $\mathcal{F}^{(2)}$ becomes small and gives us
\begin{equation*}
T_{(q t^{1/2}Q)^{1/2},x}\mathcal{F}^{(2)}=
\sum_{m,n\geq 0}{(t;1)_m(q^n q/t;1)_m\over (q;q)_m(q^n qQ;q)_m} (-\Lambda/x)^m
{(q/t;1)_n(qQ/t;1)_n\over (q;q)_n (q Q;q)_n} \bigl(t^2\Lambda\bigr)^n.
\end{equation*}
\end{Remark}

\section{Four-dimensional limit in a factorized form}
\label{sec:App.C}

In Section \ref{4dlimit}, we computed the four-dimensional limit of the operator $\SS$.
Here, we will consider the limit of each factor
\begin{gather*}
K_1=\frac{1}{\varphi(\Lambda)}\varphi(-d_1x)\varphi\left(-d_3\frac{\Lambda}{x}\right)\Bor^{-1}\varphi(qx)\varphi
\left(\frac{\Lambda}{x}\right),\\
K_2=\frac{1}{\varphi\bigl(q^{-1}d_2d_4\Lambda\bigr)}\varphi\bigl(q^{-1}d_1d_2x\bigr)\varphi
\left(d_3d_4\frac{\Lambda}{x}\right)\Bor^{-1}\varphi(-d_2x)\varphi\left(-d_4\frac{\Lambda}{x}\right),\\
N=\frac{\varphi\bigl(q^{-1}d_2d_4\Lambda\bigr)}{\varphi\bigl(q^{-1}d_1d_2d_3d_4\Lambda\bigr)},
\end{gather*}
in the decomposition
\[
\SS^{-1}=N K_2 K_1,
\]
separately. We show that each factor $K_i$ has also a well defined four-dimensional limit, though
the result is not usual differential operator any more. To describe the results in a simple form,
it is convenient to use the normal ordering $: \ :$
which is a $\bbC$-linear map
from a commutative ring~$\bbC(x, \vartheta_x)_{/ \vartheta_x x=x \vartheta_x}$ to the corresponding
non-commutative ring~$\bbC(x, \vartheta_x)_{/ \vartheta_x x=x (\vartheta_x+1)}$
defined~by
\[
{:}x^n F(x, \vartheta_x){:}  =x^n {:}F(x, \vartheta_x){:}, \qquad
{:}F(x, \vartheta_x) \vartheta_x^n{:} = {:}F(x, \vartheta_x){:} {\vartheta_x}^n, \qquad
{:}1{:}=1.
\]

\begin{Theorem}\label{ThmC}
For $q=e^h$, $d_i=q^{m_i}$, we have
\begin{gather*}
K_1 ={:}(1+x)^{-m_1-\vartheta_x}\left(1+\frac{\Lambda}{x}\right)^{-m_3+\vartheta_x}\left(1+\frac{h}{2} A_1\right)+{\mathcal O}\bigl(h^2\bigr){:}, \\
K_2 ={:}(1-x)^{-m_1-\vartheta_x}\left(1-\frac{\Lambda}{x}\right)^{-m_3+\vartheta_x}\left(1+\frac{h}{2} A_2\right)+{\mathcal O}\bigl(h^2\bigr){:},\\
N =(1-\Lambda)^{-(m_1+m_3)}\left(1+\frac{h}{2}\frac{(m_1+m_3)(m_1+2m_2+m_3+2m_4-3)}{(1-\Lambda)}\right)+{\mathcal O}\bigl(h^2\bigr),
\end{gather*}
where
\begin{gather*}
A_1= \frac{x(\vartheta_x-m_1+1)(\vartheta_x+m_1)}{1+x}
+\frac{\frac{\Lambda}{x}(\vartheta_x+m_3-1)(\vartheta_x-m_3)}{1+\frac{\Lambda}{x}}-\vartheta_x(\vartheta_x+1),\\
A_2= -\frac{x(\vartheta_x-m_1-2m_2+3)(\vartheta_x+m_1)}{1-x}
-\frac{\frac{\Lambda}{x}(\vartheta_x+m_3+2m_4-1)(\vartheta_x-m_3)}{1-\frac{\Lambda}{x}}\\
\hphantom{A_2=}{} -\vartheta_x(\vartheta_x+1).
\end{gather*}
\end{Theorem}

\begin{proof}
From \eqref{lem-2} in Appendix~\ref{App.B}, we have
\[
\Bor^{-1}\varphi(q x)\varphi\left(\frac{\Lambda}{x}\right) x^n=\frac{\varphi(\Lambda)}{\varphi(-q^{-n}x)\varphi\bigl(-q^n \frac{\Lambda}{x}\bigr)}q^{-n(n+1)}x^n,
\]
and hence
\[
K_1 x^n=\frac{\varphi(-d_1 x)\varphi\bigl(-d_3 \frac{\Lambda}{x}\bigr)}{\varphi(-q^{-n}x)\varphi\bigl(-q^n \frac{\Lambda}{x}\bigr)}q^{-n(n+1)}x^n.
\]
Then, using the limiting formula of the $q$-binomial theorem
\[
\frac{\varphi\bigl(q^jx\bigr)}{\varphi(x)}=
\sum_{k=1}^{\infty}\frac{\bigl(q^j\bigr)_k}{(q)_k}x^k=(1-x)^{-j}\left\{1+\frac{hj(j-1)}{2}\frac{x}{1-x}\right\}+{\mathcal O}\bigl(h^2\bigr),
\]
we obtain
\begin{gather*}
K_1 x^n=(1+x)^{-m_1-n}\left(1+\frac{\Lambda}{x}\right)^{-m_3+n}\left(1+\frac{h}{2} A_1\right)+{\mathcal O}\bigl(h^2\bigr),\\
A_1=\frac{x(n-m_1+1)(n+m_1)}{1+x}
+\frac{\frac{\Lambda}{x}(n+m_3-1)(n-m_3)}{1+\frac{\Lambda}{x}}-n(n+1),
\end{gather*}
as desired.

The expression for $K_2$ follows from that for $K_1$ by the substitution
\[
x \to - q^{-1}d_2 x, \qquad \Lambda \to q^{-1}d_2d_4 \Lambda.
\]
The limit of $N$ follows directly as
\begin{align*}
\displaystyle
N={}&\exp \sum_{n=1}^{\infty}\frac{\bigl(q^{-1}d_2d_4\Lambda\bigr)^n-\bigl(q^{-1}d_1d_2d_3d_4\Lambda\bigr)^n}{n(1-q^n)}
=\exp \sum_{n=1}^\infty \frac{\bigl(q^{m_2+m_4-1}\Lambda\bigr)^n\bigl(1-q^{m_1+m_3+n}\bigr)}{n(1-q^n)}\\
\displaystyle
={}&\exp \bigg\{\sum_{n=1}^\infty \frac{1}{n}(1+nh(m_2+m_4-1))\Lambda^n (m_1+m_3)\left(1+\frac{nh}{2}(m_1+m_3-1)\right)
\bigg\}\\&+{\mathcal O}\bigl(h^2\bigr),\\
\displaystyle
={}&\exp \{(m_1+m_3)\sum_{n=1}^\infty \frac{\Lambda^n}{n}\left(1+\frac{nh (m_1+2m_2+m_3+2m_4-3)}{2}\right)
\}+{\mathcal O}\bigl(h^2\bigr),\\
\displaystyle
={}&\left(1+\frac{h (m_1+2m_2+m_3+2m_4-3)}{2}\vartheta_\Lambda\right)\exp \bigg\{(m_1+m_3)\sum_{n=1}^\infty \frac{\Lambda^n}{n}
\bigg\}+{\mathcal O}\bigl(h^2\bigr),\\
\displaystyle
={}&(1-\Lambda)^{-(m_1+m_3)}\left(1+\frac{h}{2}\frac{(m_1+m_3)(m_1+2m_2+m_3+2m_4-3)}{(1-\Lambda)}\right)+{\mathcal O}\bigl(h^2\bigr).\tag*{\qed}
\end{align*}\renewcommand{\qed}{}
\end{proof}

\begin{Remark}
The result of Theorem \ref{ThmC} is consistent with the result in Section~\ref{4dlimit}.
For instance, the consistency in the leading order is given by the identity
\[
{:}(1-x)^{-m_1-\vartheta_x}\left(1-\frac{\Lambda}{x}\right)^{-m_3+\vartheta_x}{:}\,
{:}(1+x)^{-m_1-\vartheta_x}\left(1+\frac{\Lambda}{x}\right)^{-m_3+\vartheta_x}{:} =
(1-\Lambda)^{m_1+m_3},
\]
which follows from formal computations such as
\begin{equation}\label{formal}
{:}(1+x)^{-m_1-\vartheta_x}\left(1+\frac{\Lambda}{x}\right)^{-m_3+\vartheta_x}{:} f(x)
=(1+x)^{-m_1}\left(1+\frac{\Lambda}{x}\right)^{-m_3}f\left(x\frac{1+\frac{\Lambda}{x}}{1+x}\right).
\end{equation}
In relation to \eqref{formal}, the formula
\[
{:}(1-x)^{-b -\vartheta_x}{:} F(a,b,c;-x)=(1-x)^{-b}F\left(a,b,c,\frac{x}{x-1}\right)=F(c-a,b,c;x).
\]
for the Gauss Hypergeometric series $F(a,b,c;x)$ will be instructive.
\end{Remark}

\subsection*{Acknowledgements}

We would like to thank H.~Hayashi, A.N.~Kirillov, G.~Kuroki and H.~Nakajima for useful discussions.
We are also grateful to anonymous referees for useful comments and suggestions.
Our work is supported in part by Grants-in-Aid for Scientific Research (Kakenhi);
18K03274 (H.K.), 21K03180 (R.O.), 19K03512 (J.S.), 19K03530 (J.S.) and 22H01116 (Y.Y.).
The work of R.O. was partly supported by Osaka Central Advanced Mathematical
Institute: MEXT Joint Usage/Research Center on Mathematics and
Theoretical Physics JPMXP0619217849, and the Research Institute for Mathematical Sciences,
an International Joint Usage/Research Center located in Kyoto University.

\pdfbookmark[1]{References}{ref}
\LastPageEnding

\end{document}